\documentclass[11pt]{article}
\usepackage{amsmath,amsthm,amssymb,amscd,fancybox,ifthen,float,epsfig,subcaption}
\usepackage[all]{xy}
\usepackage{times}
\usepackage{color}
\usepackage{hyperref}
\floatstyle{plain}
\restylefloat{figure}

\usepackage{color}

\definecolor{gr}{rgb}   {0.,   0.8,   0. }
\definecolor{bl}{rgb}   {0.,   0.5,   1. }
\definecolor{mg}{rgb}   {0.7,  0.,    0.7}

\newcommand{\Bk}{\color{black}}
\newcommand{\Rd}{\color{black}}

\newcommand\lot{\operatorname{l.o.t.}}

\newcommand\loc{\operatorname{loc}}

\newcommand\tot{\operatorname{tot}}
\newcommand\scat{\operatorname{sc}}
\newcommand\rad{\operatorname{rad}}
\newcommand\In{\operatorname{in}}

\newcommand\fw{\mathfrak w}
\newcommand\fW{\mathfrak W}
\newcommand\fA{\mathfrak A}
\newcommand\fB{\mathfrak B}
\newcommand\fC{\mathfrak C}

\newcommand\tc{\tilde c}

\newcommand\trho{\widetilde \rho}
\newcommand\tR{\widetilde R}
\newcommand\tsigma{\widetilde \sigma}

\newcommand\ttau{\widetilde \tau}

\newcommand\tpsi{\widetilde \psi}

\newcommand\cF{\mathcal F}

\newcommand\cW{\mathcal W}

\newcommand\tgamma{\widetilde{\gamma}}
\newcommand\ta{\widetilde{a}}
\newcommand\tb{\widetilde{b}}

\newcommand\cC{\mathcal{C}}

\newcommand\cS{\mathcal{S}}

\newcommand\cD{\mathcal{D}}

\newcommand\Ha{\hat a}
\newcommand\hb{\hat b}

\newcommand\hsigma{\hat\sigma}
\newcommand\htau{\hat\tau}

\newcommand\tw{\tilde w_{0+}}

\renewcommand\Re{\operatorname{Re}}
\renewcommand\Im{\operatorname{Im}}

\newcommand\fD{\mathfrak D}

\newcommand\bbN{\mathbb N}

\newcommand\bbR{\mathbb R}
\newcommand\bbZ{\mathbb Z}

\newcommand\pa{\partial}

\newcommand\supp{\operatorname{supp}}

\newcommand\Id{\operatorname{Id}}

\newtheorem{theorem}{Theorem}
\newtheorem*{theorem*}{Theorem}
\newtheorem{proposition}{Proposition}

\newtheorem{lemma}{Lemma}
\newtheorem*{lemma*}{Lemma}

\theoremstyle{definition}

\theoremstyle{remark}
\newtheorem{remark}{Remark}

\begin{document}

\title{Solving the  Scattering Problem for Open Wave-Guide Networks, II
\\Outgoing Estimates}

\author{Charles L. Epstein\footnote{ Center for Computational Mathematics, 
  Flatiron Institute, 162 Fifth Avenue, New York, NY 10010. E-mail:
  {cepstein@flatironinstitute.org}. }
  }
\date{November 5, 2025}

\maketitle

\begin{abstract} The paper continues the analysis, started
  in~\cite{EpWG2023_1} (Part I), of the model open wave-guide
  network\footnote{\Rd In the Applied Math, Engineering and Physics literature
  an ``open wave-guide'' usually refers to a translationally invariant device.
  We call these {\em bi-infinite wave-guides.}  The main point of our work is
  that we consider an assemblage of devices that are asymptotically modeled by
  bi-infinite wave-guides, which we call a {\em wave-guide network.} \Bk}
  scattering problem defined by 2 semi-infinite, rectangular wave-guides meeting
  along a common perpendicular line. In Part I we reduce the solution of the
  scattering problem to a transmission problem rephrased as a system of integral
  equations on the common perpendicular line. In this part we show that
  solutions of the integral equations introduced in Part I have asymptotic
  expansions, if the data allows it. Using these expansions we show that the
  solutions to the PDE found in each half space have asymptotic expansions that
  imply that they satisfy appropriate outgoing radiation conditions. The
  radiation conditions are given in Part III, where we show that they imply
  uniqueness of the solution to the PDE, as well as uniqueness for our system of
  integral equations.
\end{abstract}

\tableofcontents
\section{Introduction}
This paper continues the analysis begun in~\cite{EpWG2023_1} of the scattering
problem for an open wave-guide network defined by two rectangular channels
meeting along a common perpendicular line, see Figure~\ref{fig0}.
\begin{figure}
  \centering \includegraphics[width= 10cm]{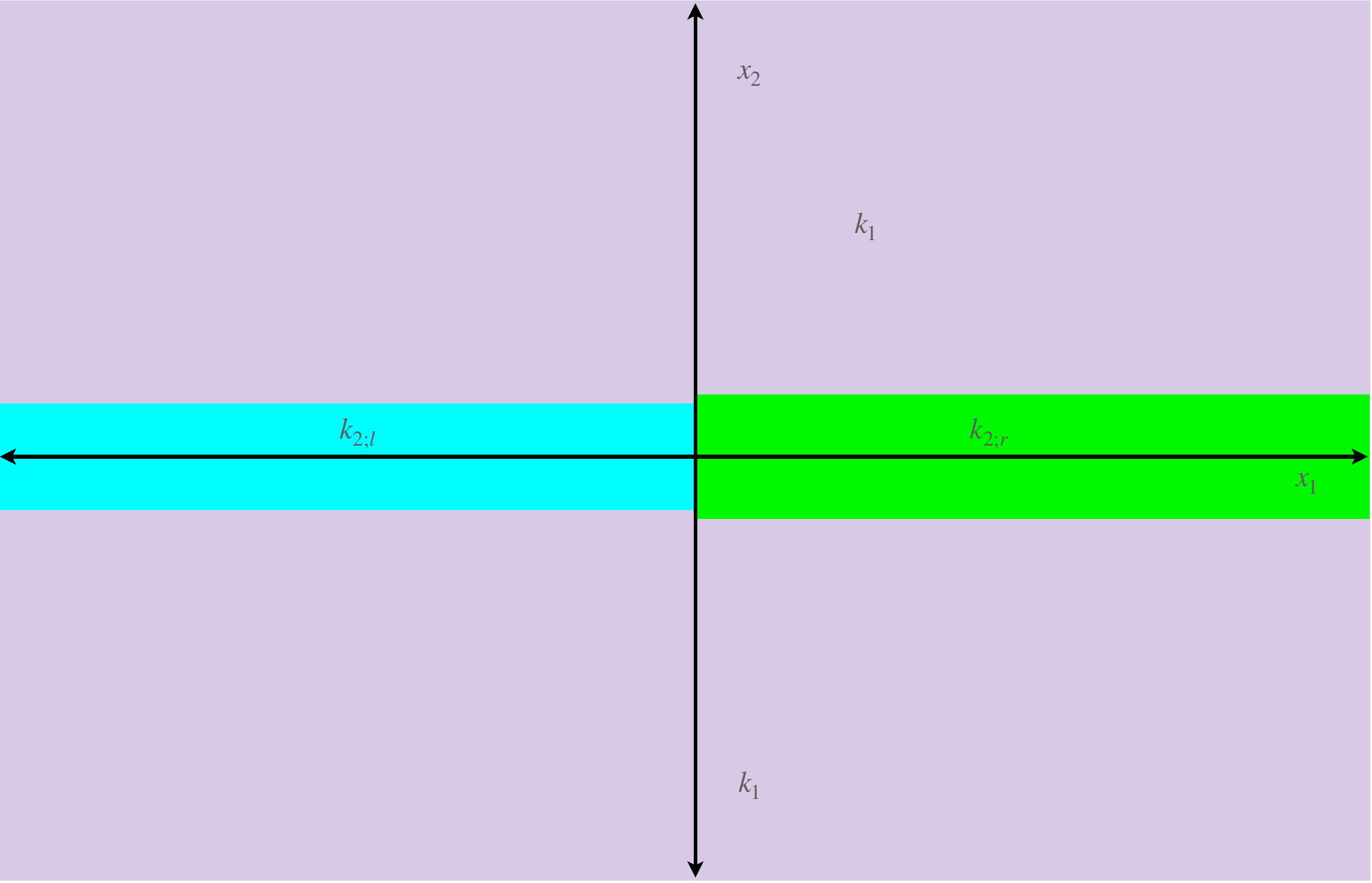}
    \caption{Two dielectric channels meeting along a straight interface. The
      $x_3$-axis is orthogonal to the plane of the image.}  
   \label{fig0}
\end{figure}
In the pages that follow we adopt the notation, and make extensive use of
results in our earlier paper,~\cite{EpWG2023_1}.  In this paper we obtain
precise \emph{uniform} asymptotics for the solution to the open wave-guide
network scattering problem found in~\cite{EpWG2023_1}. \Rd The notion of uniform
asymptotics is classical, appearing in~\cite{Erdelyi}, and explained in
Section~\ref{ss.asympt}.\Bk

\Rd We begin by recalling the scattering problem from Part I, and our approach
to solving it.  Generalities of scattering problems are discussed in Sections
1.1 and 1.2 of Part I. Here we focus  on our approach to scattering
from a pair of open wave guides that meet along a common perpendicular line.

The pair of wave-guides meeting along a common
perpendicular line is modeled by a Schr\"odinger operator with a piecewise
constant potential
\begin{multline}\label{eqn1.200}
  q(x) =q_l(x)+q_r(x)\overset{d}{=}\\
  \chi_{(-\infty,0]}(x_1)(k_{2;l}^2-k_1^2)\chi_{[y_l-d_l,y_l+d_l]}(x_2)+
    \chi_{[0,\infty)}(x_1)(k_{2;r}^2-k_1^2)\chi_{[y_r-d_r,y_r+d_r]}(x_2).
\end{multline}
We assume that $k_{2;l,r}\geq k_1;$ there is no assumption about the
relationship between $(y_l,d_l)$ and $(y_r,d_r).$ Throughout this paper, the
sub- and super-scripts $l,r$ refer to the half-planes: $l\leftrightarrow\{x_1<0\},
r\leftrightarrow \{x_1>0\}.$ We look for solutions to the equation
\begin{equation}\label{eqn2.761}
  (\Delta+k_1^2+q(x))u^{\tot}(x)=0,
\end{equation}
which belong to $H^2_{\loc}(\bbR^2).$  A scattering problem is specified by
  defining an incoming field, $u^{\In}$  and then finding the scattered field
  produced by the variations in the material parameters produced by the
  potential.  As described in Part I, the incoming field is described by a pair
  functions $u^{l,r}_{\In}$ that are solutions to
  \begin{equation}
    (\Delta+k_1^2+q_{l,r}(x))u^{\In}_{l,r}(x)=0.
  \end{equation}

  A bi-infinite wave-guide is modeled by a operator of the form
  $\Delta+q(x_2)+k^2,$ where $q$ is a compactly supported function. If $q$ is
  somewhere positive, then the the 1$d$ operator $\pa_{x_2}^2+q(x_2)$ typically
  has a finite number of $L^2$-eigenfunctions, $\{v_1(x_2),\dots,v_N(x_2)\}$
  with positive energies
  \begin{equation}\label{eqn4.765}
    (\pa_{x_2}^2+q(x_2))v_j(x_2)=E_j^2v_j(x_2).
  \end{equation}
  Setting $\xi_j=\sqrt{E_j^2+k_1^2},$ the functions
  \begin{equation}\label{eqn5.765}
    u^{\pm}_j(x_1,x_2)=e^{\pm   i\xi_j x_1}v_j(x_2)\text{ for }j=1,\dots, N,
  \end{equation}
  define solutions to $(\Delta+k^2+q(x_2))u^{\pm}_j=0,$ which are called
  \emph{wave-guide modes.} They are traveling waves, which do not decay as
  $x_1\to\pm\infty,$ but are exponentially localized within $\supp q.$ The
  wave-guide modes for $q_l$ and $q_r$ are often used as incoming data for the
  wave-guide network defined above. We denote these eigenpairs by
 $\{(v^{l,r}_{j}(x_2),\xi^{l,r}_j):\: j=1,\dots, N_{l,r}\}.$
  \Bk

The scattered field is then defined by a pair of functions $u^{\scat}_{l,r}(x)$
that satisfy the equations
\begin{equation}
  \begin{split}
    & (\Delta+k_1^2+q_{l}(x))u^{\scat}_{l}(x)=0,\text{ where }x_1<0,\\
     & (\Delta+k_1^2+q_{r}(x))u^{\scat}_{r}(x)=0,\text{ where }x_1>0,
  \end{split}
\end{equation}
along with the transmission boundary conditions along the line $\{x_1=0\}$
\begin{equation}\label{eqn4.213}
  \begin{split}
    u^l_{\scat}(0,x_2)+u_{\In}^{l}(0,x_2&=u^r_{\scat}(0,x_2))+u_{\In}^r(0,x_2),\\
    \pa_{x_1}[u^l_{\scat}(0,x_2)+u_{\In}^{l}(0,x_2)]&=\pa_{x_1}[u^r_{\scat}(0,x_2)+u_{\In}^r(0,x_2)];
  \end{split}
\end{equation}
the data for the scattered field is given $\{x_1=0\}$  by
\begin{equation}\label{eqn6.761}
  \begin{split}
    &g(x)=u_{\In}^r(0,x_2)-u_{\In}^{l}(0,x_2),\\
    &h(x)=\pa_{x_1}u_{\In}^{l}(0,x_2)-\pa_{x_1}u_{\In}^{l}(0,x_2).
  \end{split}
\end{equation}
In addition, the scattered fields are expected to satisfy an outgoing radiation
condition as $|x|\to \infty.$ While the precise nature of the radiation condition is
not given until Part III, it is just the standard Sommerfeld radiation condition in directions
not parallel to the wave-guide channels. In this paper the wave-guide directions are
$\{(\pm 1,0)\}.$ If we let $\eta=(\eta_1,\eta_2)$ denote a point on the unit
circle,  $S^1,$ we then
need to show that for $\eta\notin\{(-1,0), (0,\pm 1)\}$ for $l$, and
$\eta\notin\{(+1,0), (0,\pm 1)\}$ for $r$ we have the estimates
\begin{equation}
  (\pa_r-ik_1)u^{l,r}_{\scat}(r\eta)=O(r^{-\frac 32})
\end{equation}
as $r\to\infty,$ in some sense {\em uniformly} as $\eta$ approaches the excluded
points.  Establishing these estimates is the main goal of this paper.  The
function
\begin{equation}
   u^{\tot}(x_1,x_2)=
  \begin{cases}
    u_{\In}^l(x)+u^l_{\scat}(x)\text{ for }x_1<0,\\
      u_{\In}^r(x)+u^r_{\scat}(x)\text{ for }x_1>0,
  \end{cases}
\end{equation}
then belongs to $H^2_{\loc}(\bbR^2)$ and defines a weak solution to~\eqref{eqn2.761}.

Having reformulated the scattering problem as a transmission problem along the
$x_2$-axis, $\{x_1=0\},$ we show, in Part I, that the solution of this
problem can be reduced to the solution of a system of integral equations along
the $x_2$-axis.  To that end we construct the outgoing fundamental solutions for
the operators $(\Delta+k_1^2+q_{l,r}(x_2)),$ which can then be used to represent the
scattered fields $u^{l,r}_{\scat}$ as layer potentials along the $x_2$-axis, in the
respective half-planes. We let $(\sigma,\tau)$ denote the densities in these
representations, see~\eqref{eqn160.08}.\Bk

These densities are obtained by solving a system of integral equations on the $x_2$-axis,
which take the form
 \begin{equation}\label{eqn143.35}
     \left(\begin{matrix}\Id & D\\C&\Id\end{matrix}\right)
              \left(\begin{matrix}\sigma\\\tau\end{matrix}\right)=
        \left(\begin{matrix}g\\h\end{matrix}\right).
 \end{equation}
 \Rd This is equation (65) in Part I, where we show (Corollary 3) that this
 system is Fredholm of index 0 when acting on the function spaces
 $\cC_{\alpha}(\bbR)\oplus\cC_{\alpha+\frac 12}(\bbR),$ for any $0<\alpha<\frac
 12,$ see~\eqref{eqn9.556}.  The null-space of~\eqref{eqn143.35} is trivial, but
 the proof of this statement requires the radiation condition, and is therefore
 postponed to Part III.  The kernel functions that define the integral operators
 $C,D$ are constructed in Part I, and have very useful asymptotic expansions,
 which we describe below in~\eqref{eqn171.65}--\eqref{eqn173.65}.  \Bk

 \Rd The integral equations are solvable for data
 $(g,h)\in\cC_{\alpha}(\bbR)\oplus\cC_{\alpha+\frac 12}(\bbR),$ which may or may not
 come from incoming data, as in~\eqref{eqn6.761}. \Bk Assuming that the data,
 $(g,h),$ have appropriate asymptotic expansions, we  show that the
 solutions to the integral equation, \Rd
 $(\sigma,\tau)\in\cC_{\alpha}(\bbR)\oplus\cC_{\alpha+\frac 12}(\bbR),$ \Bk also
 have asymptotic expansions:
 \begin{equation}\label{eqn5.205}
\begin{split}
   &\sigma(x_2)=\frac{e^{ik_1|x_2|}}{|x_2|^{\frac
       12}}\sum_{l=0}^N\frac{a^{\pm}_l}{|x_2|^l}+O\left(|x_2|^{-N-\frac
    32}\right),\\
  &\tau(x_2)=\frac{e^{ik_1|x_2|}}{|x_2|^{\frac
       32}}\sum_{l=0}^N\frac{b^{\pm}_l}{|x_2|^l}+O\left(|x_2|^{-N-\frac
     52}\right), \text{ as }\pm x_2\to\infty.
 \end{split}
 \end{equation}
Using these expansions we can show that the scattered fields, $u^{l,r}_{\scat}$
have uniform (see Section~\ref{ss.asympt}) asymptotic expansions.

In our representation, the solution in each half plane naturally splits into a
radiation part, $u^{l,r}_{\scat,\rad},$ and a wave-guide mode part,
$u^{l,r}_{\scat,g},$ so that
\begin{equation}
u^{l,r}_{\scat}=  u^{l,r}_{\scat,\rad}+u^{l,r}_{\scat,g}.
\end{equation}
The wave-guide mode part is given as a finite sum
\begin{equation}
  u^{l,r}_{\scat,g}(x_1,x_2)=\sum_{m=1}^{N_{l,r}}a_m^{l,r}e^{\mp i\xi_m^{l,r}x_1}v^{l,r}_m(x_2),
\end{equation}
where $0<k_1<\xi_m^{l,r}<k_{2;l,r},$  $\{v_m^{l,r}\}\subset L^2(\bbR)$ and
\begin{equation}
  (\pa_{x_2}^2+q_{l,r}(x_2))v^{l,r}_m(x_2)=[(\xi_m^{l,r})^2-k_1^2]v^{l,r}_m(x_2),
\end{equation}
\Rd which implies  $|v_m^{l,r}(x)|=O(e^{-|x_2|\sqrt{(\xi_m^{l,r})^2-k_1^2}}),$ as
$|x_2|\to\infty.$ Note that the wave-guide mode part does not decay as
$x_1\to\pm\infty,$  which partly explains the need for more complicated radiation conditions.\Bk

Let $x=r\eta,$ where $\eta=(\eta_1,\eta_2)\in S^1.$ Assuming that $N$ can be
taken arbitrarily large in~\eqref{eqn5.205}, in each half plane we show that the
radiation parts have uniform asymptotic expansions
\begin{equation}\label{eqn9.400}
  u^{l,r}_{\scat,\rad}(r\eta)\sim\frac{e^{ik_1 r}}{\sqrt{r}}\sum_{j=0}^{\infty}\frac{a^{l,r}(\eta)}{r^j},
\end{equation}
where the coefficients $\{a^{l,r}(\eta)\}$ are smooth where
$\eta_1\cdot\eta_2\neq 0,$ and  extend smoothly to $\eta_1, \eta_2\to
0^{\pm}.$ \Rd There is some subtlety to the sense in which these expansions are
uniform, which is explained in \Bk  Remark~\ref{rmk9.200}. In addition there are somewhat different
expansions as $x_1\to\pm\infty,$ with $x_2$ remaining bounded.

In the course of proving these expansions we provide an answer to a question of
independent interest. Suppose that we define single and double layer potentials
in the half plane $\{x_1>0\},$ by integrating along  the boundary of the half plane:
\begin{equation}
  \begin{split}
    v(x)&=\cD_k(\sigma)[x]=\int_{-\infty}^{\infty}\pa_{y_1}g_k(x;0,y_2)\sigma(y_2)dy_2,\\
  u(x)&=\cS_k(\tau)[x]=\int_{-\infty}^{\infty}g_k(x;0,y_2)\tau(y_2)dy_2,
  \end{split}
\end{equation}
\Rd where $(\sigma,\tau)\in\cC_{\alpha}(\bbR)\oplus\cC_{\alpha+\frac 12}(\bbR),$  see~\eqref{eqn9.556}.
\Bk  Under what hypotheses on
$\sigma,\tau$ do these layer potentials  have uniform asymptotic expansions like those
in~\eqref{eqn9.400}? The analysis in Section~\ref{sec3.207} proves
the following result.
\begin{theorem*}[Theorem~\ref{thm3.75}]
  If $\sigma,$ resp. $\tau,$ belongs to $\cC_{\alpha}(\bbR),$
  resp. $\cC_{\alpha+\frac 12}(\bbR),$ for $0<\alpha<\frac 12,$ has an
    asymptotic expansion like that given in~\eqref{eqn5.205}, then the layer
    potential $v(r\eta),$ resp. $u(r\eta),$ has an asymptotic expansion, as
    $r\to\infty$ like that in~\eqref{eqn9.400}, with smooth coefficients and
    uniformly bounded errors as $\eta_2\to \pm 1.$
\end{theorem*}

As usual, the expansions are proved using a stationary phase argument.  The
difficulty in proving uniform expansions comes from the fact that the location
of the stationary phase varies with $\eta,$ and, in \emph{natural coordinates},
tends to the intersection of two orthogonal segments on the contour of
integration, see Figure~\ref{fig6.203}. The proof that these expansions extend
uniformly uses complex contour deformations.  In some cases there is a fixed
smooth curve so that all the stationary phase points remain in the interior of
the curve. In other cases there is a smooth family of contours so that the
stationary phase lies at the orthogonal intersection of two curve segments. The
uniformity of the expansions then follows from the lemma:
\begin{lemma*}[Lemma~\ref{lem5.202}]
   Let $f(x,\theta)\in\cC^{\infty}_c([0,1)\times [0,\delta]),$ for a
     $\delta>0.$ \Rd The function $F(r,\theta)$ defined by the integral
     \begin{equation}
       F(r,\theta)=\int_{0}^{1}e^{ir x^2}f(x,\theta)dx
     \end{equation}
    is infinitely differentiable and,  all derivatives  have
    uniform asymptotic expansions
     \begin{equation}
      \pa_r^m\pa_{\theta}^l F(r,\theta)\sim
      \sum_{j=0}^{\infty}\pa_{\theta}^la_j(\theta)
      \pa_r^m\left(\frac{1}{r^{\frac{(j+1)}{2}}}\right),
     \end{equation}
     where $0\leq l,m$ and
     $a_{j}(\theta)\in\cC^{\infty}([0,\delta]).$\Bk
\end{lemma*}
\Rd
\begin{remark}
  The subscript $c$ in $\cC^{\infty}_c([0,1)\times [0,\delta])$ means that the
    data is this space has compact support in $[0,1)\times [0,\delta].$ That is,
        the support is contained in a set of the form $[0,a]\times [0,\delta],$
        for an $a<1.$
\end{remark}
\Bk

These expansions suffice to show that the solutions
$u^{l,r}_{\scat,\rad}+u^{l,r}_{\scat,g}$ satisfy the natural outgoing radiation
conditions for this problem, which are essentially given in the work of Isozaki,
Melrose and Vasy, see~\cite{Isozaki94,Melrose94,Vasy97,Vasy2000}. These are
analyzed and explained in detail in Part III, see~\cite{EpMaSR2023}. The
radiation conditions imply uniqueness results, which show that our solutions
agree with the limiting absorption solutions, when they exist.  A similar
construction for the limiting absorption resolvents for a bi-infinite wave-guide
is considered in~\cite{Nosich1994}, which gives a physics-style discussion of
the outgoing radiation conditions. It contains neither a mathematically rigorous
treatment, nor a precise radiation condition valid in a neighborhood of the
channels.

While we only explicitly consider the case that $q_{l,r}(x_2)$ are of the simple form given
in~\eqref{eqn1.200}, the asymptotics established herein are valid for any pair
of piecewise smooth, bounded potentials with bounded support, for which 0 is not
a {\em threshold,}  that is the equations
$(\pa_{x_2}^2+q_{l,r}(x_2))v(x_2)=0$ do \emph{ not } have bounded
solutions. \Rd This is proved in~\cite{EGHQR2025}.\Bk

\subsection{Asymptotic Expansions}\label{ss.asympt}

With $\eta\in S^1$ and $r>0,$ we say that a function $f(r\eta)$ has an
asymptotic expansion as $r\to\infty,$ and write
   \begin{equation}
     f(r\eta)\sim\frac{e^{ik_1 r}}{r^{\alpha}}\sum_{j=0}^{\infty}\frac{a_j(\eta)}{r^j},
   \end{equation}
   provided that, for any $N>0$
   \begin{equation}\label{eqn17.760}
     \left|f(r\eta)-\frac{e^{ik_1
         r}}{r^{\alpha}}\sum_{j=0}^{N}\frac{a_j(\eta)}{r^j}\right|=o(r^{-(N+\alpha)}). 
   \end{equation}
   \Rd If $\Xi\subset S^1,$ then the expansion is \emph{uniform} for $\eta\in
   \Xi$ provided the implied constants in the error terms in~\eqref{eqn17.760}
   are uniformly bounded for $\eta\in\Xi.$ Smoothness in $\eta$ means that the
   coefficients $\{a_j(\eta)\}$ are smooth, we can differentiate the
   expansion with respect to $\eta,$ and the error terms remain the same order.
   That is, for any $N,$
\begin{equation}\label{eqn18.760}
     \left|\pa^l_{\eta}f(r\eta)-\frac{e^{ik_1
         r}}{r^{\alpha}}\sum_{j=0}^{N}\frac{\pa^l_{\eta}a_j(\eta)}{r^j}\right|=o(r^{-(N+\alpha)}). 
\end{equation}\Bk

   An important fact about the error terms in asymptotic expansions that arise
   from stationary phase calculations is contained in the following lemma.
   \begin{lemma}\label{lem0}
     Let $f\in\cC^{\infty}_c((-1,1)),$ the function
     \begin{equation}
       F(r)=\int_{-1}^1e^{irx^2}f(x)dx,
     \end{equation}
     has the asymptotic expansion
     \begin{equation}
       F(r)\sim\sum_{j=0}^{\infty}\frac{a_j}{r^{j+\frac 12}}.
     \end{equation}
     For each $N$ let
     \begin{equation}
       R_N(r)= F(r)-\sum_{j=0}^{N}\frac{a_j}{r^{j+\frac 12}},
     \end{equation}
     then for each $l\geq 0$
     \begin{equation}\label{eqn22.760}
         \pa_r^{l}F(r)\sim\sum_{j=0}^{\infty}a_j\pa_r^l\left[\frac{1}{r^{j+\frac 12}}\right],
     \end{equation}
     and we have the estimate
     \begin{equation}
       \left|\pa_{r}^{l}R_N(r)\right|=O(r^{-(N+l+\frac 32)}).
     \end{equation}
   \end{lemma}
   \begin{remark}
     Briefly we can differentiate the asymptotic expansion and the remainder
     terms, $R_N(r),$
     satisfy symbolic estimates.
   \end{remark}
   \begin{proof}
   \Rd  Choose $\varphi\in\cC^{\infty}_c((-1,1)),$ an even function with $\varphi(x)=1$ for
     $x\in\supp f.$ We first observe that, for any $j\in\bbN,$
   \begin{equation}\label{eqn20.45}
     \begin{split}
       &\int_{-1}^{1}e^{irx^2}x^{2j}\varphi(x)dx=\frac{a_j}{r^{j+\frac
           12}}+w_j(r).\\
       & \int_{-1}^{1}e^{irx^2}x^{2j-1}\varphi(x)dx=0.
       \end{split}
     \end{equation}
     The error terms $w_j(r)=O(r^{-M}),$ for any $M\in\bbN.$ The formula for
     $j>0$ is obtained by differentiating the $j=0$-formula $j$ times,
     see~\cite{ZworskiSA2012}.
    
     Let
     \begin{equation}
       \rho_N(x)=f(x)-\sum_{j=0}^{2N+1} \frac{f^{[j]}(0)}{j!}x^j,
     \end{equation}
     then
     \begin{equation}\label{eqn26.555}
       F(r)=\int_{-1}^1e^{irx^2}\varphi(x)\left[\sum_{j=0}^{2N+1}
       \frac{f^{[j]}(0)}{j!}x^j\right]dx+
       \int_{-1}^1e^{irx^2}\rho_N(x)\varphi(x)dx.
     \end{equation}
     It follows from~\eqref{eqn20.45}  that the first $N+1$
     terms of the asymptotic expansion of $F(r)$ come from the first integral
     in~\eqref{eqn26.555}, with an error term that satisfies the conclusion of
     the lemma.   As
     \begin{equation}
       \pa_r^l F(r)=\int_{-1}^1e^{irx^2}(ix^2)^lf(x)dx,
     \end{equation}
     this also verifies~\eqref{eqn22.760}.

     Taylor's theorem shows that
     \begin{equation}
       \rho_N(x)=x^{2(N+1)}\trho_N(x),
     \end{equation}
    where $\trho_N\in\cC^{\infty}((-1,1)).$ It remains to show that
     \begin{equation}\label{eqn28.555}
       R^{(2)}_N(r)=\int_{-1}^1e^{irx^2}\rho_N(x)\varphi(x)dx=\int_{-1}^1e^{irx^2}x^{2(N+1)}\trho_N(x)\varphi(x)dx
     \end{equation}
     also satisfies the conclusion of the lemma. 
      Integrating by parts $N+1$ times
     shows that
     \begin{equation}\label{eqn29.555}
       \begin{split}
         R^{(2)}_N(r)&=\int_{-1}^1e^{irx^2}\left(\pa_x\frac{-1}{2irx}\right)^{N+1}[x^{2(N+1)}\trho_N(x)\varphi(x)]dx\\
         &=\frac{1}{r^{N+1}}\int_{-1}^1e^{irx^2} \trho_{N,N}(x)dx,
         \end{split}
     \end{equation}
     for a function $\trho_{N,N}\in\cC^{\infty}_c((-1,1)).$  It follows from 
     standard stationary phase results that the integral is $O(r^{-\frac 12}),$
     and therefore
     \begin{equation}
       R^{(2)}_N(r)=O\left(r^{-(N+\frac 32)}\right).
     \end{equation}
     The  formula for $R^{(2)}_N(r)$ can be differentiated under the integral sign, as
     many times as we please, giving
     \begin{equation}
         \pa_r^lR^{(2)}_N(r)=\int_{-1}^1e^{irx^2}(ix^2)^l\rho_N(x)\varphi(x)dx.
     \end{equation}
     This is precisely the sort of integral appearing in~\eqref{eqn28.555} with
     $N$ replaced by $N+l,$ and therefore the foregoing argument shows that
     \begin{equation}
        \pa_r^lR^{(2)}_N(r)=O\left(r^{-(N+l+\frac 32)}\right).
     \end{equation}
     \Bk
   \end{proof}
   \Rd
   \begin{remark}
     This is an elaboration of result of Coddington and Levinson,  Theorem 3.2 in Chapter 5
     of~\cite{CoddingtonLevinson}, which states
     \begin{theorem*}[ Theorem 3.2 (b), Chapter 5 of~\cite{CoddingtonLevinson}]
       If $f(t)\sim \sum_{k=0}^{\infty}p_kt^{-k},$ $f(t)$ is continuously
       differentiable for $t>t_0$ and $f'(t)$ has an asymptotic expansion, then
       \begin{equation}
         f'(t)\sim -\sum_{k=1}^{\infty}kp_kt^{-(k+1)}.
       \end{equation}
     \end{theorem*}
     We will have several occasions to use this theorem below, as it allows us
     to identify the coefficients in the asymptotic expansions of derivatives of
     functions with asymptotic expansions, without any careful accounting.
   \end{remark}
\Bk
   
   \subsection{Some Notation}
\Rd  
   We define some non-standard notation that we use throughout this paper.
   \begin{enumerate}
   \item The notation $A\overset{d}{=}B$ means that $A$ is defined by $B.$
   \item $B_d$ is the subset of $\bbR^2$ defined by $B_d=[-d,d]\times [-d,d].$
   \item For $\alpha\in\bbR,$ $\cC_{\alpha}(\bbR)$ is the Banach space defined
     as the subspace of functions $f\in\cC^0(\bbR)$ for which
      \begin{equation}\label{eqn9.556}
    |f|_{\alpha}=\sup_{x\in\bbR}\{(1+|x|)^{\alpha}|f(x)|\}<\infty,
      \end{equation}
      with $|\cdot|_{\alpha}$  the norm on  $\cC_{\alpha}(\bbR).$
      \item We let $\eta=(\eta_1,\eta_2)$ denote a point on the unit circle
        $S^1\subset\bbR^2.$
        \item We let $\Id$ denote the identity operator: $\Id f=f.$
   \end{enumerate}
   \Bk
   
   {\small
  
\centerline{\sc Acknowledgments} I would like to thank Leslie Greengard for
suggesting this problem and for many interesting conversations along the way. I
would also like to thank Manas Rachh, Shidong Jiang, Felipe Vico, and Alex Barnett for many
helpful discussions of this material and pointers to the literature on this
problem. I am very grateful to Manas for carefully reading this manuscript
and providing very useful comments. I would like to thank David Jerison, Rafe
Mazzeo and Andras Vasy for useful discussions about these issues. I am 
very grateful for the support of the Flatiron Institute of the Simons
Foundation, and for the support of Stanford University through the Bergman
Visiting Scholarship.  I would finally like the thank the referees for Parts I,
II and III whose careful reading and thoughtful comments have lead to very substantial
improvements in all 3 papers. }
  
 \section{Outgoing Estimates for $\sigma$  and $\tau$ }\label{sec8.71}
 In this section we show that a solution, $(\sigma,\tau)$ to the integral
 equations
 \begin{equation}\label{eqn33.557}
   \left(
   \begin{matrix}
     \Id&D\\ C&\Id
   \end{matrix}\right) \left(
   \begin{matrix}
     \sigma\\ \tau
   \end{matrix}\right)=\left(
   \begin{matrix}
     g \\ h
   \end{matrix}\right),
 \end{equation}
 introduced in Section 5 of Part I, equations (64), (65) and (66), has an
 asymptotic expansion like that in~\eqref{eqn5.205}, if the data $(g,h)$ allows
 it. \Rd The kernel functions defining the operator $C,D$ are described below
 in~\eqref{eqn14.206},~\eqref{eqn40.764}.  When $(\sigma,\tau)$ has an expansion
 like that in~\eqref{eqn5.205} the resultant solutions, $u^{l,r}(r\eta),$ also
 have asymptotic expansions that are uniform in the asymptotic direction
 $\eta\in S^1.$ \Bk

  As a consequence of the non-compactness of the domain over which we integrate
  the layer potentials, uniform estimates are rather delicate to prove.  In Part
  III of this series we state the radiation conditions for wave-guide networks precisely
  \Rd and show that the solutions we have constructed satisfy these
  conditions. They imply uniqueness for the solution of the original PDE
  problem, which allows us to show that the integral equations
  in~\eqref{eqn33.557} have a trivial null-space and are therefore always
  solvable for data $(g,h)\in\cC_{\alpha}(\bbR)\oplus\cC_{\frac 12
    +\alpha}(\bbR),\,0<\alpha<\frac 12.$\Bk

  Estimates for $(\sigma,\tau)$ are obtained by a bootstrap argument: We use
  estimates for the kernel functions defining the operators $C,D$
  to first show that, if $(g,h)$ are \Rd differentiable \Bk and decrease rapidly enough,
  then the solution $(\sigma,\tau)$ to~\eqref{eqn33.557}, which a priori belongs
  to $\cC_{\alpha}(\bbR)\oplus\cC_{\frac 12 +\alpha}(\bbR),$ for an
  $0<\alpha<\frac 12,$ actually satisfies the optimal decay estimates
 \begin{equation}\label{eqn164.76}
   |\sigma(x_2)|\leq \frac{M}{(1+|x_2|)^{\frac 12}}\,\text{ and }\,
    |\tau(x_2)|\leq \frac{M}{(1+|x_2|)^{\frac 32}}.
 \end{equation}
 Using this first step, in Lemmas~\ref{lem3.62} and~\ref{lem4.62} we then show
 that $C\sigma, D\tau$ satisfy symbolic estimates of all orders. Using these
 estimates we finally show that $\sigma(x_2)$ and $\tau(x_2)$ have asymptotic
 expansions as $|x_2|\to\infty,$ see~\eqref{eqn200.61}.
 
The kernel functions, $k_C(x_2,y_2), k_D(x_2,y_2),$ for the operators $C$ and
$D$ \Rd are  introduced in equation (64) of Part I. Their properties are
established in Theorem 1 of Part I.  \Bk The kernels are given by
   \begin{equation}\label{eqn14.206}
     \begin{split}
       k_C(x_2,y_2)&=\fw_l^{[2]}(x_2,y_2)-\fw_r^{[2]}(x_2,y_2)+\\ &
       \phantom{lllllllllllmmmmmm}\pa^2_{x_1y_1}w^g_{l}(0,x_2;0,y_2)-\pa^2_{x_1y_1}w^g_{r}(0,x_2;0,y_2),\\
         k_D(x_2,y_2)&=\fw_l^{[0]}(x_2,y_2)-\fw_r^{[0]}(x_2,y_2)+w^g_{l}(0,x_2;0,y_2)-w^g_{r}(0,x_2;0,y_2).
     \end{split}
   \end{equation}
   \Rd The terms $w^g_{l,r}, \pa^2_{x_1y_1}w^g_{l,r},$ come from the wave-guide
   modes, see~\eqref{eqn4.765},~\eqref{eqn5.765}. If the $\{v^{l,r}_j(x_2)\}$
   are properly normalized, then
   \begin{equation}\label{eqn40.764}
     w^g_{l,r}(x_1,x_2;y_1,y_2)=\sum_{j=1}^{N_{l,r}}e^{\pm
       i\xi^{l,r}_j(x_1-y_1)}v^{l,r}_j(x_2)v^{l,r}_j(y_2),
     \text{ with }+\leftrightarrow r,\,-\leftrightarrow l.
   \end{equation}
These terms are infinitely differentiable except at the boundaries of the
supports of $q_{l,r},$ where they are $\cC^1.$ They are exponentially decaying,
along with all of their derivatives as $|x_2|\to\infty.$\Bk

 For the convenience of the reader we recall
 the expansions proved in~\cite{EpWG2023_1}. If the support of the potential
 lies in $|x_2|<d,$ then as proved in Theorem I of Part I,  we have the expansions:
   \begin{enumerate}
   \item  If both $\pm x_2>d$ and  $\pm y_2>d, \,j=0,1,2,$ then
     \begin{equation}\label{eqn171.65}
        \fw^{[j]}_{l,r}(x_2,y_2)\sim
        \frac{e^{ik_1(|x_2|+|y_2|)}}{(|x_2|+|y_2|)^{\frac{j+1}{2}}}\left[M^{\pm,\pm}_{j0;l,r}+\sum_{n=1}^{\infty}
          \frac{M^{\pm,\pm}_{jn;l,r}}{(|x_2|+|y_2|)^n}\right].
     \end{equation}
     In these subsets of $\bbR^2$ the kernel functions depend only on
    $|x_2|+|y_2|.$  We also have
     \begin{equation}\label{eqn122.56}
       \begin{split}
        \fw^{[j]}_{l,r}(x_2,y_2)&\sim \frac{e^{ik_1|y_2|}}
           {|y_2|^{\frac{j+1}{2}}}\left[\sum_{n=0}^{\infty}\frac{b^{\pm}_{jn;l,r}(x_2)}{|y_2|^n}\right],
           \text{
           where }\pm y_2>d,\, |x_2|<d;\\
              \fw^{[j]}_{l,r}(x_2,y_2)&\sim
       \frac{e^{ik_1|x_2|}}{|x_2|^{\frac{j+1}{2}}}\left[\sum_{n=0}^{\infty}
       \frac{c^{\pm}_{jn;l,r}(y_2)}{|x_2|^n}\right],       \text{
           where }\pm x_2>d,\, |y_2|<d.
       \end{split}
     \end{equation}
   \Rd The coefficients
   $b^{\pm}_{jn;l,r},c^{\pm}_{jn;l,r}\in\cC^{\infty}([-d,d]).$

   \item In $B_d^c\subset \bbR^2$ the functions $\fw^{[j]}_{l,r}(x_2,y_2)$ are
     infinitely differentiable, and expansions for their derivatives are
     obtained by differentiating the expansions in~\eqref{eqn171.65} and
     ~\eqref{eqn122.56}. Along the diagonal in $B_d$ the functions
     $\fw^{[j]}_{l,r}$ have $|x_2-y_2|^{2-j}\log|x_2-y_2|$-singularities. More
     precise descriptions of the diagonal singularities are given in equation (271) of Part I.
  \item
     In particular, applying the differential operators $\pa_{x_2}\mp ik_1,$ or
     $\pa_{y_2}\mp ik_1$ to the appropriate expansion above give asymptotic
     expansions for the functions $\pa_{x_2}\fw^{[j]}_{l,r}(x_2,y_2) \mp
     ik_1\fw^{[j]}_{l,r}(x_2,y_2),$ $\pa_{y_2}\fw^{[j]}_{l,r}(x_2,y_2) \mp
     ik_1\fw^{[j]}_{l,r}(x_2,y_2).$ These expansions easily imply that the
     kernels are outgoing, that is
       \begin{equation}\label{eqn173.65} 
         \begin{split}
           \pa_{x_2}\fw^{[j]}_{l,r}(x_2,y_2) \mp    ik_1\fw^{[j]}_{l,r}(x_2,y_2)=O((|x_2|+|y_2|)^{-\frac{j+3}{2}}),\\
           \pa_{y_2}\fw^{[j]}_{l,r}(x_2,y_2) \mp ik_1\fw^{[j]}_{l,r}(x_2,y_2)=O((|x_2|+|y_2|)^{-\frac{j+3}{2}}),
         \end{split}
       \end{equation}
       \Rd as the appropriate variable tends to $\pm\infty.$ \Bk  
      
   \end{enumerate}

   \subsection{The basic estimate}
 We assume that $\supp q_{l,r}\subset [-d,d];$ using  the relations
 \begin{equation}\label{eqn213.65}
   \sigma(x_2)=g(x_2)-D\tau(x_2),\quad\tau(x_2)=h(x_2)-C\sigma(x_2),
 \end{equation}
 we prove the following proposition.
 \begin{proposition}\label{prop1.206} Assume that
   $(\sigma,\tau)\in\cC_{\alpha}(\bbR)\oplus\cC_{\alpha+\frac 12}(\bbR),$ for an
   $0<\alpha<\frac 12,$ satisfy~\eqref{eqn213.65},
    with the data $g,h \in \cC^0(\bbR)\cap
   \cC^1((-\infty,-d)\cup (d,\infty))$  satisfying the estimates
 \begin{equation}\label{eqn20.211}
   \begin{split}
     |g(x_2)|\leq \frac{M}{(1+|x_2|)^{\frac 12}},&\quad
      |(\pa_{x_2}\mp ik_1)g(x_2)|\leq \frac{M}{(1+|x_2|)^{\frac 32}}\text{ for  }
        \pm x_2>d,\\
         |h(x_2)|\leq \frac{M}{(1+|x_2|)^{\frac 32}},&\quad
      |(\pa_{x_2}\mp ik_1)h(x_2)|\leq \frac{M}{(1+|x_2|)^{2}}\text{ for }
        \pm x_2> d.
        \end{split}
 \end{equation}
 Then $(\sigma,\tau),$  satisfy the estimates\Rd
 \begin{equation}\label{eqn42.557}
   |\tau(x_2)|\leq \frac{m_{\alpha,K}(|\sigma|_{\alpha}+M)}{(1+|x_2|)^{\frac 32}},\quad
     |\sigma(x_2)|\leq \frac{m_{\alpha,K}(|\tau|_{\alpha+\frac 12}+M)}{(1+|x_2|)^{\frac 12}},
 \end{equation}
 for all $x_2,$ the constants $m_{\alpha,K}$ depend $\alpha,$ and bounds on
 the kernels, $k_C,k_D.$\Bk
 \end{proposition}
 \begin{remark}
 \Rd  
   In the proof below $m_{\alpha,K}$ denotes a variety of positive constants that
   depend on $\alpha,$ and bounds on the kernel functions, but not the data.
   While the estimates below are valid for $|x_2|>d,$ there main
   interest lies in what they say as $|x_2|\to\infty.$ Note that we estimate
   different quantities depending on whether $x_2\to +\infty$ or $x_2\to
   -\infty$\Bk
 \end{remark}
 \begin{proof}
 Outside a neighborhood of $B_d$ the kernels $k_C(x_2,y_2),
 k_D(x_2,y_2)$ are infinitely differentiable.  \Rd Using the estimates on $\tau$
 and on the kernel $k_D$ along with the mean value theorem and the Lebesgue
 dominated convergence theorem, one can easily show to show that $D\tau(x_2)$ is
 differentiable, for $|x_2|>d,$ and that we can differentiate $D\tau(x_2)$ under
 the integral sign. Indeed this argument can be used, along with the asymptotic
 expansions for the kernel and its derivatives, to see that $D\tau(x_2)$ is
 infinitely differentiable where $|x_2|>d,$ and can be repeatedly differentiated
 under the integral sign. Similar remarks apply to 
 $C\sigma(x_2).$ \Bk Using the estimate in~\eqref{eqn173.65} for the kernel
 function, and the fact that
 \begin{equation}
   |\tau(x_2)|\leq\frac{|\tau|_{\alpha+\frac 12}}{(1+|x_2|)^{\alpha+\frac 12}}
 \end{equation}
we see that, for
 $\pm x_2\to \infty,$
 \begin{equation}\label{eqn196.30}
   \begin{split}
     |(\pa_{x_2}\mp ik_1)D\tau(x_2)|&\leq
     \int_{-\infty}^{\infty}\frac{m_{\alpha,K}|\tau|_{\alpha+\frac 12}dy_2}{|y_2|^{\alpha+\frac
         12}(|x_2|+|y_2|)^{\frac 32}}\\
     &= \frac{|\tau|_{\alpha+\frac 12}}{|x_2|^{\alpha+1}}\int_{-\infty}^{\infty}\frac{m_{\alpha,K}d t}{|t|^{\alpha+\frac
         12}(1+|t|)^{\frac 32}}\leq \frac{m_{\alpha,K}|\tau|_{\alpha+\frac 12}}{|x_2|^{\alpha+1}}.
   \end{split}
 \end{equation}
 
 Similarly we can show that, as $\pm x_2\to\infty,$
 \begin{equation}\label{eqn197.30}
   \begin{split}
     |(\pa_{x_2}\mp ik_1)C\sigma(x_2)|&\leq
     \int_{-\infty}^{\infty}\frac{m_{\alpha,K}|\sigma|_{\alpha} dy_2}{|y_2|^{\alpha}(|x_2|+|y_2|)^{\frac 52}}\\
     &= \frac{|\sigma|_{\alpha}}{|x_2|^{\alpha+\frac
         32}}\int_{-\infty}^{\infty}\frac{m_{\alpha,K}d t}{|t|^{\alpha}
       (1+|t|)^{\frac 52}}\leq
     \frac{m_{\alpha,K}|\sigma|_{\alpha}}{|x_2|^{\alpha+\frac 32}}.
   \end{split}
 \end{equation}
Equation~\eqref{eqn213.65} implies that, for $|x_2|>d,$
 \begin{equation}
   \pa_{x_2}\sigma=\pa_{x_2}g-\pa_{x_2}D\tau,\quad
   \pa_{x_2}\tau=\pa_{x_2}h-\pa_{x_2}C\sigma,
    \end{equation}
 and  therefore these estimates, along with~\eqref{eqn20.211},  show that, for
 $\pm x_2 >d,$ we have the estimates
 \begin{equation}\label{eqn218.66}
   \begin{split}
   &|\pa_{x_2}\sigma(x_2)\mp ik_1\sigma(x_2)|\leq
   \frac{m_{\alpha,K}|\sigma|_{\alpha}+M}{(1+|x_2|)^{\alpha+1}}\\
   &|\pa_{x_2}\tau(x_2)\mp ik_1\tau(x_2)|\leq
   \frac{m_{\alpha,K}|\tau|_{\alpha+\frac 12}+M}{(1+|x_2|)^{\alpha+\frac 32}}.
   \end{split}
 \end{equation}

 Let $\psi_{\pm}\in\cC^{\infty}(\bbR)$ be monotone and non-negative with 
 \begin{equation}\label{eqn190.76}
   \psi_{\pm}(x)=
   \begin{cases}
     &0\text{ for }\pm x<d,\\
     &1\text{ for }\pm x>d+\frac 12.
   \end{cases}
 \end{equation}
 Using integration by parts and~\eqref{eqn173.65}  we see
 that for $|x_2|>d,$ 
 \begin{equation}\label{eqn218.65}
   \begin{split}
    \Bigg|\int_{d}^{\infty}k_D(x_2,y_2)[\pa_{y_2}(\psi_+\tau)(y_2)+&ik_1\psi_+\tau(y_2)]
     dy_2\Bigg|\\
     &=\left|\int_{d}^{\infty}(\pa_{y_2}-ik_1)k_D(x_2,y_2)\psi_+\tau(y_2)dy_2\right|\\
     &\leq \int_{d}^{\infty}\frac{m_{\alpha,K}|\tau|_{\alpha+\frac 12}dy_2}{|y_2|^{\alpha+\frac
         12}(|x_2|+|y_2|)^{\frac 32}}\leq \frac{m_{\alpha,K}|\tau|_{\alpha+\frac 12}}{|x_2|^{\alpha+1}}.
     \end{split}
 \end{equation}
 A similar argument with $\psi_-(x_2)$ shows that
 \begin{equation}\label{eqn219.65}
   \left|\int_{-\infty}^{-d}k_D(x_2,y_2)[\pa_{y_2}(\psi_-\tau)(y_2)-ik_1\psi_-\tau(y_2)]dy_2\right|
   \leq \frac{m_{\alpha,K}|\tau|_{\alpha+\frac 12}}{|x_2|^{\alpha+1}}.
 \end{equation}

 The estimates in~\eqref{eqn218.66} show that,  for $\pm x_2>d,$ we have
 \begin{multline}
   \pa_{x_2}(\psi_{\pm}\tau)(x_2)=(\pa_{x_2}\psi_{\pm}(x_2))\tau(x_2)+\psi_{\pm}(x_2)\pa_{x_2}\tau(x_2)
   =\\
   (\pa_{x_2}\psi_{\pm}(x_2))\tau\pm ik_1\psi_{\pm}\tau(x_2)+O(|x_2|^{-(\alpha+\frac 32)}).
 \end{multline}
 Here the implicit constant in the $O$-term is that in~\eqref{eqn218.66}.  Using
 this in~\eqref{eqn218.65}, along with the triangle inequality and the leading
 order estimate for $|k_D(x_2,y_2)|,$ we see that
 \begin{multline}\label{eqn222.66}
   \left|2ik_1\int_{d}^{\infty}k_D(x_2,y_2)\psi_+\tau(y_2)dy_2\right|\leq\\
  m_{\alpha,K}(|\tau|_{\alpha+\frac 12}+M)\left[ \frac{1}{|x_2|^{1+\alpha}}+\frac{1}{(1+|x_2|)^{\frac 12}}+
    \int_{d}^{\infty}\frac{dy_2}{y_2^{\alpha+\frac 32}(x_2+y_2)^{\frac 12}}\right].
 \end{multline}
 An elementary estimate shows that the integral on the r.h.s. is
 bounded by $m_{\alpha,K}x_2^{-\frac 12}.$ There is a similar estimate arising
 from~\eqref{eqn219.65} for the
 integral from $-\infty$ to $-d.$

 As
 \begin{multline}
   D\tau(x_2)=\int_{-\infty}^{-d}k_D(x_2,y_2)\psi_-\tau(y_2)dy_2+
   \int^{\infty}_{d} k_D(x_2,y_2)\psi_+\tau(y_2)dy_2+\\
   \int_{-(d+1)}^{d+1}k_D(x_2,y_2)(1-\psi_+(y_2)-\psi_-(y_2))\tau(y_2)dy_2,
 \end{multline}
from~\eqref{eqn222.66} and the analogous estimate for $\psi_-\tau$  we conclude that
 \begin{equation}
   |D\tau(x_2)|\leq \frac{m_{\alpha}K(|\tau|_{\alpha+\frac 12}+M)}{(1+|x_2|)^{\frac 12}},
 \end{equation}
 and therefore~\eqref{eqn213.65} and~\eqref{eqn20.211} imply that
 \begin{equation}\label{eqn206.30}
   |\sigma(x_2)|\leq \frac{m_{\alpha}K(|\tau|_{\alpha+\frac 12}+M)+M}{(1+|x_2|)^{\frac 12}}
 \end{equation}
 as well.

 We can apply this argument for the integrals appearing in $C\sigma$ to
 conclude that
 \begin{equation}
   \begin{split}
     &\left|\int_{d}^{\infty}k_C(x_2,y_2)
     [\pa_{y_2}(\psi_+\sigma)(y_2)+ik_1\psi_+\sigma(y_2)]dy_2\right|\leq
     \frac{m_{\alpha,K}(|\sigma|_{\alpha}+M)}{x_2^2},\\
     &\left|\int_{-\infty}^{-d}k_C(x_2,y_2)
          [\pa_{y_2}(\psi_-\sigma)(y_2)-ik_1\psi_-\sigma(y_2)]dy_2\right|\leq
          \frac{m_{\alpha,K}(|\sigma|_{\alpha}+M)}{x_2^2}.
   \end{split}
 \end{equation}
 Combining these estimates with the estimate in~\eqref{eqn218.66} we 
 see that
 \begin{equation}
   |C\sigma(x_2)|\leq \frac{m_{\alpha,K}(|\sigma|_{\alpha}+M)}{(1+|x_2|)^{\frac 32}},
 \end{equation}
 and therefore~\eqref{eqn213.65} and~\eqref{eqn20.211} imply that
 \begin{equation}\label{eqn209.30}
   |\tau(x_2)|\leq \frac{m_{\alpha,K}(|\sigma|_{\alpha}+M)+M}{(1+|x_2|)^{\frac 32}}
 \end{equation}
 as well. This completes the proof of the proposition.
\end{proof}
 \subsection{The asymptotic expansion}
 From the form of the kernels it is reasonable to expect that $\sigma$ and $\tau$  have asymptotic
 expansions of the form
 \begin{equation}\label{eqn200.61}
   \sigma(x_2)\sim\frac{e^{ik_1|x_2|}}{|x_2|^{\frac 12}}\left[\sum_{l=0}^N\frac{
     a^{\pm}_l}{|x_2|^l}\right]\text{ and }
\tau(x_2)\sim\frac{e^{ik_1|x_2|}}{|x_2|^{\frac 32}}\left[\sum_{l=0}^N\frac{ b^{\pm}_l}{|x_2|^l}\right],
 \end{equation}
 for a choice of $N$ largely determined by the data $(g,h).$
 Choose a non-negative $\varphi\in\cC^{\infty}((-d-2\epsilon,d+2\epsilon))$ with
 \begin{equation}\label{eqn202.76}
    \varphi(x_2)=1 \text{ for }|x_2|\leq d+\epsilon;
 \end{equation}
 define
  \begin{equation}
   \begin{split}
     C_0\sigma(x_2)&=\int_{-d-\epsilon}^{d+\epsilon}k_C(x_2,y_2)\varphi(y_2)\sigma(y_2)dy_2,\quad C_1\sigma=(C-C_0)\sigma,\\
     D_0\tau(x_2)&=\int_{-d-\epsilon}^{d+\epsilon}k_D(x_2,y_2)\varphi(y_2)\sigma(y_2)dy_2,
     \quad D_1\tau=(D-D_0)\tau.
\end{split}
  \end{equation}
  
From~\eqref{eqn14.206},~\eqref{eqn171.65},~\eqref{eqn122.56} and the smoothness
of the kernels outside of $B_d$ it follows that $C_0\sigma(x_2), D_0\tau(x_2)$
have asymptotic expansions of exactly the sort given in~\eqref{eqn200.61}.
 To prove that $\sigma$ and $\tau$ do as well we
 first remove the oscillations from these functions by defining
 \begin{equation}
   \tsigma(x_2)=e^{-ik_1|x_2|}\sigma(x_2),\quad \ttau(x_2)=e^{-ik_1|x_2|}\tau(x_2).
 \end{equation}
 To obtain the asymptotic expansion we first  show that these functions
 satisfy symbolic estimates
 \begin{equation}
   | \pa_{x_2}^l\tsigma(x_2)|\leq \frac{C_l}{(1+|x_2|)^{l+\frac 12}}\text{ and }
    | \pa_{x_2}^l\ttau(x_2)|\leq \frac{C'_l}{(1+|x_2|)^{l+\frac 32}}.
 \end{equation}
 With these estimates in hand we show that expansions like those
 in~\eqref{eqn200.61} are in fact correct.

 We begin with an estimate on $D\tau.$
 \begin{lemma}\label{lem3.62} 
   Suppose that $\tau\in L^1(\bbR),$ then for each $l\geq 0,$ there is a $C_{l}$
   so that, for large $|x_2|,$ we have the estimate
   \begin{equation}
    \left| \pa_{x_2}^l[e^{-i|x_2|k_1}D\tau(x_2)]\right|\leq\frac{C_l\|\tau\|_{L^1}}{(1+|x_2|)^{l+\frac 12}}.
   \end{equation}
 \end{lemma}
 \begin{remark}
   If $g,h$ satisfy the hypotheses of Proposition~\ref{prop1.206},
   then~\eqref{eqn42.557} implies that $\tau\in L^1(\bbR),$ \Rd but note that this
   is not generally true for functions in $\cC_{\alpha+\frac 12}(\bbR).$\Bk
 \end{remark}
 \begin{proof}
   With $\varphi$ as above, we write
   \begin{multline}
     e^{-ik_1|x_2|}D\tau(x_2)=\int_{-d-\epsilon}^{d+\epsilon}e^{-ik_1|x_2|}k_D(x_2,y_2)\tau(y_2)\varphi(y_2)dy_2+\\
     \int_{|y_2|>d}e^{-ik_1|x_2|}k_D(x_2,y_2)\tau(y_2)(1-\varphi(y_2))dy_2.
   \end{multline}
   The estimates for the compactly supported part follow immediately from~\eqref{eqn171.65}
   and~\eqref{eqn122.56}  and the fact that $k_D(x_2,y_2)$ is smooth outside of
   $B_d.$ Now suppose that $|x_2|>d;$ using~\eqref{eqn171.65}, \Rd and the fact that
   we can differentiate these asymptotic expansions, \Bk  shows that
   \begin{equation}\begin{split}
     \left|\pa_{x_2}^l\left[
       \int_{|y_2|>d}e^{-ik_1|x_2|}k_D(x_2,y_2)\tau(y_2)(1-\varphi(y_2))dy_2\right]\right|\leq&
     \int_{|y_2|>d}\frac{C_l|\tau(y_2)|dy_2}{(|x_2|+|y_2|)^{l+\frac 12}}\\
    \leq &\frac{C_l\|\tau\|_{L^1}}{|x_2|^{l+\frac 12}}.
    \end{split}
   \end{equation}
   The contribution of the wave-guide modes decays exponentially as $|x_2|\to\infty.$
 \end{proof}

 With slightly different hypotheses, we estimate $C\sigma.$
  \begin{lemma}\label{lem4.62} 
   Suppose that $\sigma\in\cC^1(\bbR)\cap L^{\infty}(\bbR),$ and
   $\pa_{x_2}\tsigma\in L^1(\bbR),$ then for each $l\geq 0,$ there is a $C_{l}$
   so that, for large $|x_2|,$ we have the estimate
   \begin{equation}
     \left|\pa_{x_2}^l[e^{-i|x_2|k_1}C\sigma(x_2)]\right|\leq\frac{C_l(\|\tsigma\|_{L^{\infty}}+\|\tsigma'\|_{L^1})}{(1+|x_2|)^{l+\frac 32}}.
   \end{equation}
  \end{lemma}
  \begin{remark}
    Since
    \begin{equation}
      \tsigma(x_2)=e^{-ik_1|x_2|}g(x_2)-e^{-ik_1|x_2|}D\tau(x_2),
    \end{equation}
    if $\tau$  and $\pa_{x_2}[e^{-ik_1|x_2|}g]$ are  in $L^1(\bbR)$ then $\pa_{x_2}\tsigma\in L^1(\bbR)$ as well.
  \end{remark}
 \begin{proof}  With $\varphi$ defined in~\eqref{eqn202.76},
   we write
   \begin{multline}
     e^{-ik_1|x_2|}C\sigma(x_2)=\int_{-d-\epsilon}^{d+\epsilon}e^{-ik_1|x_2|}k_C(x_2,y_2)\sigma(y_2)\varphi(y_2)dy_2+\\
     \int_{|y_2|>d}e^{-ik_1|x_2|}k_C(x_2,y_2)\sigma(y_2)(1-\varphi(y_2))dy_2.
   \end{multline}
   The estimates for the compactly supported part follow as before.

  Now suppose that $|x_2|>d.$ We first consider the $l=0$ case;
  using~\eqref{eqn171.65} we integrate by parts once
  \begin{equation}\label{eqn208.62}
    \begin{split}
          &\int_{|y_2|>d}e^{2ik_1|y_2|}[e^{-ik_1(|x_2|+|y_2|)}k_C(x_2,y_2)]\tsigma(y_2)(1-\varphi(y_2))dy_2=\\
   &\frac{e^{2ik_1 |y_2|}}{2ik_1}\tsigma(y_2)(1-\varphi(y_2))[e^{-ik_1(|x_2|+|y_2|)}k_C(x_2,y_2)]
         \Bigg|_{-\infty}^{-d}+\Bigg|^{\infty}_{d}\\
         &+\int_{|y_2|>d} \frac{e^{2ik_1 |y_2|}}{2ik_1}\pa_{y_2}\left[\tsigma(y_2)(1-\varphi(y_2))
           [e^{-ik_1(|x_2|+|y_2|)}k_C(x_2,y_2)]\right]dy_2.
    \end{split}
   \end{equation}
   The boundary terms on the second line are zero.  From~\eqref{eqn171.65} it follows
   that the term with derivative placed on
   $[e^{-ik_1(|x_2|+|y_2|)}k_C(x_2,y_2)]$ is bounded by
   $C\|\sigma\|_{L^{\infty}}|x_2|^{-\frac 32}.$  If the
   derivative is placed on $\tsigma(y_2)(1-\varphi(y_2)),$ then the argument
   used in the previous proof shows that this term is
   $O(\|\tsigma'\|_{L^1}|x_2|^{-\frac 32}).$ Applying $\pa_{x_2}^l$ to the
   expression on the last line in~\eqref{eqn208.62} we easily establish the
   remaining estimates.
 \end{proof}

 Using these estimates we can now prove symbolic estimates on $\tsigma,\ttau;$
 we begin with $\tsigma.$ The equations in~\eqref{eqn213.65} imply that
 \begin{equation}
   \tsigma(x_2)=e^{-ik_1|x_2|}g(x_2)-e^{-ik_1|x_2|}D\tau(x_2).
 \end{equation}
 As $\tau\in L^1(\bbR),$ Lemma~\ref{lem3.62} shows that
 \begin{equation}
   |\pa_{x_2}^le^{-ik_1|x_2|}D\tau(x_2)|\leq\frac{C_l}{(1+|x_2|)^{l+\frac 12}}.
 \end{equation}
 If we assume that the data also satisfies symbolic estimates
 $$|\pa_{x_2}^le^{-ik_1|x_2|}g(x_2)|\leq C_l'(1+|x_2|)^{-(l+\frac 12)}\text{ for }l=0,1,\dots,N,$$
 then it follows that
 \begin{equation}\label{eqn211.62}
      |\pa_{x_2}^l\tsigma(x_2)|\leq\frac{C''_l}{(1+|x_2|)^{l+\frac 12}},\text{
        for }l=0,\dots,N.
 \end{equation}

 Similarly we use the relation
 \begin{equation}
   \ttau(x_2)=e^{-ik_1|x_2|}h(x_2)-e^{-ik_1|x_2|}C\sigma(x_2),
 \end{equation}
 and Lemma~\ref{lem4.62} to prove symbolic estimates for $\ttau.$ We must assume
 that $\tsigma'$ is integrable (which holds if~\eqref{eqn211.62} is valid for $l=1$) and
 that
 $$|\pa_{x_2}^l[e^{-ik_1|x_2|}h(x_2)]|=O((1+|x_2|)^{-(l+\frac 32)}),\text{ for }l=0,\dots,N. $$
In this case we have the estimates
 \begin{equation}\label{eqn214.63}
        |\pa_{x_2}^l\ttau(x_2)|\leq\frac{C''_l}{(1+|x_2|)^{l+\frac 32}},\text{
        for }l=0,\dots,N.
 \end{equation}
 Using these estimates we can now prove that $(\sigma,\tau)$ have asymptotic
 expansions.

 As noted above, it is obvious that $C_0\sigma$ and
 $D_0\tau$ have the desired asymptotic expansions.  To prove the
 existence of the asymptotic expansions for $\sigma,\tau,$ we consider the functions
 \begin{equation}\label{eqn77.557}
   \begin{split}
   s(w)&=w^{-\frac 12}[e^{-ik_1|x_2|}D_1\tau(x_2)]_{x_2=w^{-1}}\\
   t(w)&=w^{-\frac 32}[e^{-ik_1|x_2|}C_1\sigma(x_2)]_{x_2=w^{-1}}.
   \end{split}
   \end{equation}
 To prove the existence of asymptotic expansions of order $N$ as $\pm
 x_2\to\infty,$ it suffices to show the existence of the limits
 \begin{equation}
   \lim_{w\to 0^{\pm}}\pa_{w}^ls(w) \text{ resp. }
    \lim_{w\to 0^{\pm}}\pa_{w}^lt(w)\text{ for }l=0,\dots,N+1.
 \end{equation}
 The proof that these limits exist makes usage of the asymptotic expansions
 in~\eqref{eqn171.65}; as before the contribution of the wave-guide modes decays
 exponentially as $|x_2|\to\infty.$ We assume that the data have asymptotic expansions:
 \begin{equation}\label{eqn217.63}
   \begin{split}
     g(x_2)&=\frac{e^{ik_1|x_2|}}{|x_2|^{\frac 12}}\sum_{j=0}^N\frac{g_j^{\pm}}{|x_2|^j}+O\left(|x_2|^{-N-\frac 32}\right),\\
       h(x_2)&=\frac{e^{ik_1|x_2|}}{|x_2|^{\frac 32}}\sum_{j=0}^N\frac{h_j^{\pm}}{|x_2|^j}+O\left(|x_2|^{-N-\frac 52}\right).
   \end{split}
 \end{equation}
 If the incoming data is defined by wave-guide modes, then this hypothesis is
   trivially satisfied, with all coefficients zero.

 We now prove the following theorem:
 \begin{theorem}\label{thm2.75}
   Let $(\sigma,\tau)\in\cC_{\alpha}(\bbR)\oplus\cC_{\alpha+\frac 12}(\bbR)$
   solve~\eqref{eqn143.35}. Assuming~\eqref{eqn217.63} holds for some $N\geq 1,$
   we have the following asymptotic expansions
   \begin{equation}\label{eqn230.63}
     \begin{split}
   &\sigma(x_2)=\frac{e^{ik_1|x_2|}}{|x_2|^{\frac
       12}}\sum_{l=0}^N\frac{a^{\pm}_l}{|x_2|^l}+O\left(|x_2|^{-N-\frac
         32}\right),\\
       &\tau(x_2)=\frac{e^{ik_1|x_2|}}{|x_2|^{\frac
       32}}\sum_{l=0}^{N-1}\frac{b^{\pm}_l}{|x_2|^l}+O\left(|x_2|^{-N-\frac
     32}\right), \text{ as }\pm x_2\to\infty.
     \end{split}
     \end{equation}
 \end{theorem}

 \begin{proof}
With $\varphi$ defined in~\eqref{eqn202.76}, we let
   \begin{equation}\label{eqn220.76}
      \varphi_{\pm}(y_2)=\chi_{(d,\infty)}(\pm y_2)[1-\varphi(y_2)].
   \end{equation}
   We first consider $s(w),$ defined in~\eqref{eqn77.557}, for $w>0.$ Using the
   expansions for the kernel function in~\eqref{eqn171.65}
   and~\eqref{eqn122.56}, we see that $s(w)$ has an expansion with terms that
   are constant multiples of
 \begin{equation}\label{eqn80.558}
   s_{\pm}^k(w)\overset{d}{=}\int_{|y_2|>d}\frac{e^{2ik_1
       y_2}\ttau(y_2)\varphi_{\pm}(y_2)w^kdy_2}{(1\pm w y_2)^{k+\frac
       12}},\quad k=0,1,\dots
 \end{equation}
 We give the details for $y_2>0,$ the other cases are essentially identical.

 Terms with $k>l$ do not contribute to
 \begin{equation}
   \lim_{w\to 0^+}\pa_w^ls(w).
 \end{equation}
 \Rd
 Applying the Leibniz formula to the integrand in~\eqref{eqn80.558} we see that,
 for $k>l,$
 \begin{equation}
   \pa_{w}^l\left(\frac{w^k}{(1\pm wy_2)^{k+\frac 12}}\right)=\sum_{j=0}^l
   C_{kl}\frac{w^{k-l}(wy_2)^{l-j}}{(1\pm wy_2)^{l-j+k+\frac 12}}\leq Cw^{k-l},
 \end{equation}
 showing that these terms tend to zero as $w\to 0^{\pm}.$
 \Bk

 The error term from truncating the expansion for $k_D$ after $N$
 terms is
 \begin{multline}
   R_N(x_2,y_2)=\\
   [\fw_l^{[0]}(x_2,y_2)-\fw_r^{[0]}(x_2,y_2)]-\frac{e^{ik_1(x_1+y_2)}}{(x_2+y_2)^{\frac
       12}}\left[\sum_{n=0}^{N}\frac{M^{+,+}_{0n;l}-M^{+,+}_{0n;r}}
   {(x_2+y_2)^n}\right],
 \end{multline}
 which is a function of $x_2+y_2,$ that is $O((x_2+y_2)^{-(N+\frac
   32)}).$ With $\tR_N(x_2,y_2)=e^{-ik(x_2+y_2)}R_N(x_2,y_2),$ its contribution to
 $s(w)$ is
 \begin{equation}
  E_N(w)= \int_{d}^{\infty}\tR_N(w^{-1}+y_2)e^{2ik_1y_2}\varphi_+(y_2)\ttau(y_2)dy_2;
 \end{equation}
if we differentiate w.r.t. $w$ we obtain
 \begin{equation}
   \pa_w E_N(w)=-\int_{d}^{\infty}w^{-2}\tR^{[1]}_N(w^{-1}+y_2)e^{2ik_1y_2}\varphi_+(y_2)\ttau(y_2)dy_2.
 \end{equation}
 Lemma~\ref{lem0} shows that $\tR_N^{[1]}(x)=O(x^{-(N+\frac 52)}),$ so the integral is
 bounded by
 \begin{equation}
   C\int_{d}^{\infty}\frac{w^{N+\frac 12}}{(1+wy_2)^{N+\frac 52}}\varphi_+(y_2)|\ttau(y_2)|dy_2.
 \end{equation}
 As $\ttau\in L^1$ we see that the limit is zero as $w\to 0^+.$ Using the Leibniz formula, and
 Lemma~\ref{lem0} repeatedly we can show that $\tR_N^{[l]}(x)=O(x^{-(N+\frac 32+l)}),$ and therefore
 \begin{equation}
   \lim_{w\to 0^+}\pa_w^l E_N(w)=0\text{ for }l\leq N+1.
 \end{equation}

 It is clear that for $w>0$ we can differentiate under the integral sign and
 apply the Leibniz rule to obtain
 \begin{equation}
   \pa_w^ls_+^k(w)=
   \sum_{m=0}^kC^k_{lm}\int_{d}^{\infty}e^{2ik_1y_2}\frac{w^{k-m}y_2^{l-m}\ttau(y_2)\varphi_+(y_2)dy_2}{(1+wy_2)^{k+l-m+\frac 12}},
 \end{equation}
 for some constants $\{C^k_{lm}\}.$ As noted above, we only need to consider
 $k\leq l.$ Using the fact that $\pa_{y_2}e^{2ik_1y_2}=2ik_1e^{2ik_1y_2},$ we
 integrate this expression by parts $(l-k)$-times to obtain
 \begin{equation}\label{eqn220.63} 
   \pa_w^ls_+^k(w)=
   \frac{1}{(-2ik_2)^{l-k}}
   \sum_{m=0}^kC^k_{lm}\int_{d}^{\infty}e^{2ik_1y_2}\pa_{y_2}^{l-k}\left[\frac{w^{k-m}y_2^{l-m}\ttau(y_2)
       \varphi_+(y_2)dy_2}{(1+wy_2)^{k+l-m+\frac 12}}\right].
 \end{equation}
 Applying the Leibniz formula again we see that
 \begin{equation}\label{eqn221.63}
   \pa_{y_2}^{l-k}\left[\frac{w^{k-m}y_2^{l-m}\ttau(y_2)
       \varphi_+(y_2)}{(1+wy_2)^{k+l-m+\frac 12}}\right]=
   \sum_{j=0}^{l-k}b_j\frac{w^{k-m+j}\pa_{y_2}^{l-k-j}[y_2^{l-m}\ttau(y_2)\varphi_+(y_2)]}{(1+wy_2)^{k+l+j-m+\frac 12}}.
 \end{equation}
 Using the estimates in~\eqref{eqn214.63} we can show that
 \begin{equation}
   |\pa_{y_2}^{l-k-j}[y_2^{l-m}\ttau(y_2)\varphi_+(y_2)]|\leq Cy_2^{k-m+j-\frac 32}.
 \end{equation}
 This estimate shows that the $j$th term in~\eqref{eqn221.63} is bounded by
 \begin{equation}
   C\frac{(wy_2)^{k-m+j}}{(1+wy_2)^{k-m+j+l+\frac 12}}\cdot\frac{1}{y_2^{\frac
       32}}\leq C \frac{1}{(1+wy_2)^{l+\frac 12}}\cdot\frac{1}{y_2^{\frac
       32}},
 \end{equation}
 which shows that the integrand in~\eqref{eqn220.63} is uniformly integrable as
 $w\to 0^+.$ The Lebesgue dominated convergence theorem then implies that
 \begin{equation}
   \lim_{w\to 0^+}\pa_w^ls_+(w)\text{ exists for }l=0,\dots,N.
 \end{equation}

 Applying the same argument to $s_-(w),$ and for $w<0,$ we conclude that $s(w)$ has  order $N$
 Taylor expansions at $w=0^{\pm},$ and therefore, assuming~\eqref{eqn217.63}, we have
 \begin{equation}\label{eqn226.63}
   \sigma(x_2)=\frac{e^{ik_1|x_2|}}{|x_2|^{\frac
       12}}\sum_{l=0}^N\frac{a^{\pm}_l}{|x_2|^l}+O\left(|x_2|^{-N-\frac
     32}\right),
   \text{ as }\pm x_2\to\infty.
 \end{equation}

 A very similar argument applies to analyze $\pa_w^lt(w).$ An additional step is
 needed as $\tsigma(y_2)$ is not in $L^1(\bbR).$ We integrate the terms in the asymptotic
 expansion of $C_1\sigma$ by parts to obtain
 \begin{equation}\label{eqn87.45}
   \begin{split}
    & \int_{\pm y_2>d}\frac{e^{2ik_1y_2}\tsigma(y_2)\varphi_{\pm}(y_2)dy_2}{(x_2\pm
       y_2)^{l+\frac 32}}=\\
     & -\int_{\pm y_2>d}\frac{e^{2ik_1y_2}}{2ik_1}\left[\frac{\pa_{y_2}(\tsigma(y_2)\varphi_{\pm}(y_2))}{(x_2\pm
       y_2)^{l+\frac 32}}\mp \left(l+\frac 32\right)\frac{\tsigma(y_2)\varphi_{\pm}(y_2)}{(x_2\pm
         y_2)^{l+\frac 52}}\right]dy_2.
   \end{split}
 \end{equation}
 Since $|\pa_{y_2}(\tsigma(y_2)\varphi_{\pm}(y_2))|\leq C/|y_2|^{\frac 32}$ the
 arguments used to analyze $s_{\pm}$ apply to these terms as well. The
 contribution of the second term on the r.h.s.  to $t_{\pm}(w)$  takes the form
 \begin{multline}
   \int_{\pm y_2>d}\frac{e^{2ik_1y_2}}{2ik_1}\left[\frac{w^{l+\frac 52}\tsigma(y_2)\varphi_{\pm}(y_2)}{(1\pm
       wy_2)^{l+\frac 52}}\right]dy_2=\\
    \int_{\pm y_2>d}\frac{e^{2ik_1y_2}}{2ik_1}\left[\frac{w^{l+\frac 32}\tsigma(y_2)\varphi_{\pm}(y_2)}{y_2(1\pm
         wy_2)^{l+\frac 32}}\cdot\frac{wy_2}{(1\pm wy_2)}\right]dy_2.
 \end{multline}
 In the end we loose one term in the asymptotic expansion, and, with small
 modifications, our earlier arguments apply to show that
 \begin{equation}
  \lim_{w\to 0^{\pm}} \pa_w^lt(w)\text{ exists for }l=0,\dots,N-1.
 \end{equation}
 This completes the proof of the theorem.
 \end{proof}

  \begin{remark}
   If $(\sigma,\tau)$ solve~\eqref{eqn143.35} with $(g,h)=(0,0),$ then they
   automatically have asymptotic expansions like those in~\eqref{eqn230.63} for
   any $N>0.$ This proves useful in Part III when we prove that the null-space
   of~\eqref{eqn143.35} is trivial. \Rd The expansions for $\sigma$ and $\tau$
   can be differentiated, which follows from the fact that the expansions for
   the kernels $k_C(x_2,y_2), k_D(x_2,y_2),$ can be differentiated.  To prove
   these expansions for derivatives of the densities we replace $s(w), t(w)$ in
   the proof of Theorem~\ref{thm2.75} with
   \begin{equation}
     \begin{split}
     &s^{(l)}(w)=w^{-(l+\frac
       12)}\left[\pa_{x_2}^l(e^{-ik_1|x_2|}D_1\tau(x_2))\right]_{x_2=w^{-1}},\\
      &t^{(l)}(w)=w^{-(l+\frac
         32)}\left[\pa_{x_2}^l(e^{-ik_1|x_2|}C_1\sigma(x_2))\right]_{x_2=w^{-1}}.
       \end{split}
   \end{equation}
   \Bk
 \end{remark}

 \section{Asymptotics for The Free Space Part of
   $u^{l,r}$}\label{sec3.207}

 Theorem~\ref{thm2.75} shows that the asymptotic expansions satisfied by the
 kernels $k_C,k_D$ as $|x_2|+|y_2|\to\infty$ force the sources, $(\sigma,\tau),$
 to also have asymptotic expansions, provided that the data allows it. Using
 these expansions, the representation formul{\ae} for the solutions
  \begin{equation}\label{eqn160.08}
    u^{l,r}(x)=
    \cS_{k_1}\tau(x)- \cD_{k_1}\sigma(x)+\int\limits_{-\infty}^{\infty}[w^{l,r}(x;0,y_2)\tau(y_2)-\pa_{y_1}w^{l,r}(x;0,y_2)\sigma(y_2)]dy_2,
   \end{equation}
  and the Sommerfeld formula, see~\cite{MorseFeshbach},
  \begin{equation}\label{eqn26.51}
  \cF_{x_2}[(i/4)H^{(1)}_0(k_1|x|)](\xi)=\frac{ie^{i|x_1|\sqrt{k_1^2-\xi^2}}}{2\sqrt{k_1^2-\xi^2}},
\end{equation}
   we derive asymptotic expansions for $u^{l,r}(r\eta)$ along radial lines
   through the origin. As a consequence of  these expansions we see that,
   except for the guided modes, our solutions satisfy a standard Sommerfeld
   radiation condition.  Within the channels, $\{\eta_2=0\},$ the guided modes
   satisfy appropriate outgoing conditions.

      Let $\eta\in S^1\subset\bbR^2$ be a unit vector, and set
 \begin{equation}\label{eqn235.85}
   \begin{split}
   u^{l,r}_0(r\eta)&=\cS_{k_1}\tau(r\eta)=\frac{i}{4}\int_{-\infty}^{\infty}
   H^{(1)}_0(k_1|r\eta-(0,y_2)|)\tau(y_2)dy_2,\\
    u^{l,r}_1(r\eta)&=\cD_{k_1}\sigma(r\eta)=-\frac{ik_1}{4}\int_{-\infty}^{\infty}
    \frac{r\eta_1 [\pa_z H^{(1)}_0](k_1|r\eta-(0,y_2)|)\sigma(y_2)dy_2}{\sqrt{r^2-2r\eta_2y_2+y_2^2}}.
    \end{split}
 \end{equation}
 As usual,  the sub- or super-scripts $l\leftrightarrow
 \{r\eta:\:\eta_1<0\}$ and $r\leftrightarrow \{r\eta:\:\eta_1>0\}.$

 In this section we analyze the behavior of the free space contribution to the solution,
 $-\cD_{k_1}\sigma(r\eta)+\cS_{k_1}\tau(r\eta),$ assuming that $\eta_1\neq 0.$
 The more difficult case to analyze, where $\eta_2\to {\pm 1},$ is deferred to
 Section~\ref{sec7.4.83}.  These estimates are summarized in the following
 theorem.
 \begin{theorem}\label{thm3.75}
   Suppose that the data $(g,h)$ satisfies~\eqref{eqn217.63} for an
   $N\geq 2,$ then, for an $M$ depending on $N,$ and $\eta_2\in [-1,1],$ we have the uniform
   asymptotic expansions
   \begin{equation}\label{eqn237.81}
     \begin{split}
       u_0^{l,r}(r\eta)&=\frac{e^{irk_1}}{\sqrt{r}}\sum_{j=0}^{M}\frac{a_{0j}^{l,r}(\eta)}{r^j}+o(r^{-(M+1)}),\\
         u_1^{l,r}(r\eta)&=\frac{e^{irk_1}}{\sqrt{r}}\sum_{j=0}^{M}\frac{a_{1j}^{l,r}(\eta)}{r^j}+o(r^{-(M+1)}).
     \end{split}
   \end{equation}
  The coefficients are smooth where $\eta_1\neq 0,$ and have smooth extensions
  to $\eta_2\to \pm 1.$ The functions $\pa_r^j u^{l,r}_{k}(r\eta),$ with
  $j\in\bbN,$ have similar
  expansions obtained by differentiating the expansions
  in~\eqref{eqn237.81}. Here $M\to\infty$ as $N\to\infty.$
 \end{theorem}
 \begin{remark}
  For simplicity we assume that $N$ in~\eqref{eqn217.63}, and therefore $M$
  in~\eqref{eqn237.81}, can be taken arbitrarily large. \Rd  A similar result, but
  without the uniformity as $\eta_1\to 0,$ appears in~\cite{BBD_CW_F2022}.\Bk
 \end{remark}

 \Rd The proof of this theorem is given in the following two sections. In
 Section~\ref{sec4} we prove the analogous results for the portion of the
 solution coming from the perturbation terms,
 $\cW^{l,r}\tau-\cW^{l,r\,'}\sigma.$ As expected these proofs rely on stationary
 phase calculations. There are several difficulties that arise.  The first is
 that the integral extends over the whole real line, and it is difficult to
 control the $r\to\infty$ asymptotics of the unbounded part of the integral. To
 handle this we use the Plancherel formula to replace the unbounded portion of
 the integral with its Fourier transform, which effectively handles this
 problem.

 The Fourier transforms of the sources $\sigma$ and $\tau$ are computed using
 their asymptotic expansions, see Lemma~\ref{lem4.201}. The Fourier
 representation presents a somewhat different difficulty, which is that the
 integrand in the Fourier representation has square root singularities at
 $\xi=\pm k_1.$ Upon changing variables to regularize the integrand, the domain
 of integration is replaced with a pair of line segments, in what we now
 recognize as a complex contour integral. The image of the singularity in the $\xi$-variable is the
 intersection of these line segments, see Figure~\ref{fig1.2}.  In this representation the stationary
 phase point occurs in the interior of one of these segments, but as $\eta_1\to
 0,$ the stationary point moves to the intersection point of the two
 segments. The integrand has an analytic continuation to a region bounded by
 these segments. In order to prove asymptotics, with uniform errors as
 $\eta_1\to 0$ we need to deform the contour. Somewhat counterintuitively, we can
 obtain uniform estimates by deforming the contour so that the stationary point
 {\em always} occurs at the intersection of two segments, see
 Lemma~\ref{lem5.202}. There are several cases requiring different contour
 deformations, which largely explains the length of this and the following
 section. While not much detail is given, the idea of using contour deformation
 as a way to establish uniformity in asymptotic expansions for layered media
 appears in~\cite{Christiansen}.

 \Bk

  \subsection{Estimates with $\eta_1\neq 0$}
In this section we assume that $\eta_1\neq 0;$ in the next section we show that
the coefficients have smooth extensions to $\eta_2\to \pm 1,$ with uniform
error estimates.  To obtain the desired asymptotics, we split the
$y_2$-integrals in~\eqref{eqn235.85} into a part with $|y_2|$ small and parts
with $|y_2|$ unbounded. The arguments are quite different for these 2 cases.  We
now drop the $l,r$ sub- and super-scripts, and focus on the $r$-case, i.e.~$\eta_1>0.$

Assume that $\supp q\subset [-d,d];$ let $\varphi_{\pm}\in\cC^{\infty}(\bbR),$ be monotone increasing,
with $\varphi_-(y_2)=\varphi_+(-y_2),$ where
 \begin{equation}\label{eqn1.100}
   \varphi_{+}(y_2)=\begin{cases} &0\text{  for }y_2<d+1,\\
   &1\text{ for }y_2>d+2.
   \end{cases}
 \end{equation}
 Let $\varphi_0(y_2)=1-(\varphi_+(y_2)+\varphi_-(y_2));$ with $\epsilon\in\{0,+,-\},$ define
 \begin{equation}\label{eqn238.200}
   \begin{split}
     &u_{0}^{\epsilon}(r\eta)=\int_{-\infty}^{\infty}H^{(1)}_0(|r\eta-(0,y_2)|)\varphi_{\epsilon}(y_2)\tau(y_2)dy_2,\\
   &u_{1}^{\epsilon}(r\eta)=\int_{-\infty}^{\infty}\pa_{y_1}H^{(1)}_0(|r\eta-(0,y_2)|)\varphi_{\epsilon}(y_2)\sigma(y_2)dy_2,
   \end{split}
 \end{equation}
 so that
 \begin{equation}
   u_j(r\eta)=u_j^-(r\eta)+u_j^0(r\eta)+u_j^+(r\eta),\text{ for }j=0,1.
 \end{equation}

   We begin with $u^0_{j}(r\eta),$ $j=0,1,$ for which this result is standard. To obtain the expansion we use
   the large $|z|$ asymptotics of the Bessel function:
 $$H^{(1)}_0(z)\sim C\frac{e^{iz}}{\sqrt{z}}\cdot
   \sum_{j=0}^{\infty}\frac{a_j}{z^j}.$$
   Using this expansion we see that
   \begin{equation}
     \begin{split}
       u^0_0(r\eta)&=\int_{-(d+2)}^{d+2}H^{(1)}_0(|r\eta-(0,y_2)|)\varphi_{0}(y_2)\tau(y_2)dy_2\\ &\sim\frac{1}{\sqrt{r}}
       \int_{-(d+2)}^{d+2}\frac{e^{ik_1r\left(1-\frac{2\eta_2y_2}{r}+\frac{y_2^2}{r^2}\right)^{\frac
             12}}}{ \sqrt{1-\frac{2\eta_2y_2}{r}+\frac{y_2^2}{r^2}}
       }\sum\limits_{j=0}^{\infty}\frac{a_j}{r^j\left(1-\frac{2\eta_2y_2}{r}+\frac{y_2^2}{r^2}\right)^{\frac j2}}\varphi_{0}(y_2)\tau(y_2)dy_2
     \end{split}
   \end{equation}
   Using the convergent expansion
   \begin{equation}
      \sqrt{1-\frac{2\eta_2y_2}{r}+\frac{y_2^2}{r^2}}=1+\sum_{j=0}^{\infty}C_j\left(\frac{2\eta_2y_2}{r}-\frac{y_2^2}{r^2}\right)^j,
   \end{equation}
   the expansion for its reciprocal, and the fact that the integral is over a
   bounded interval, we easily obtain the desired expansion for this term. The
   smooth dependence on $\eta_2\in[-1,1]$ and the uniformity of the error terms
   is also clear. A similar argument, using the fact that
\begin{equation}
  \pa_{y_1}H^{(1)}_0(|r\eta-(0,y_2)|)=-\frac{r\eta_1}{\sqrt{r^2-2r\eta_2y_2+y_2^2}}\pa_zH^{(1)}_0(|r\eta-(0,y_2)|),
\end{equation}
applies to $u^0_1(r\eta).$ Altogether we conclude that
   \begin{equation}
     u^0_{k}(r\eta)\sim
     \frac{e^{ik_1r}}{\sqrt{r}}\sum_{j=0}^{\infty}\frac{a^0_{kj}(\eta)}{r^j},\quad k=0,1,
   \end{equation}
   uniformly as $\eta_2\to\pm 1.$

   To treat the unbounded terms we use the Fourier representations for the single
   and double layer kernels, $H^{(1)}_0(|x-(0,y_2)|), \pa_{y_1}H^{(1)}_0(|x-(0,y_2)|):$
   \begin{equation}\label{eqn113.763}
     \begin{split}
       H^{(1)}_0(k_1|x-(0,y_2)|)&=\frac{i}{4\pi}\int_{-\infty}^{\infty}
       \frac{e^{i\xi(x_2-y_2)+i\sqrt{k_1^2-\xi^2}|x_1|}d\xi}{\sqrt{k_1^2-\xi^2}},\\
     \pa_{y_1} H^{(1)}_0(k_1|x-(0,y_2)|)&=\frac{1}{4\pi}
     \int_{-\infty}^{\infty}e^{i\xi(x_2-y_2)+i\sqrt{k_1^2-\xi^2}|x_1|}d\xi,
     \end{split}
   \end{equation}
   \Rd see~\cite[Chap. 7.2]{MorseFeshbach}. \Bk 
   To use this representation, requires the Fourier transforms of the sources
$\varphi_{\pm}(y_2)\sigma(y_2),$ $\varphi_{\pm}(y_2)\tau(y_2),$
   \begin{equation}\label{eqn89.42}
     \begin{split}
       \hsigma_{\pm}(\xi)&=\lim_{R\to\infty}\int_{-R}^{R}e^{-iy_2\xi}\varphi_{\pm}(y_2)\sigma(y_2)dy_2,\\
       \htau_{\pm}(\xi)&=\int_{-\infty}^{\infty}e^{-iy_2\xi}\varphi_{\pm}(y_2)\tau(y_2)dy_2.
     \end{split}
   \end{equation}
   The integrals defining $\htau_{\pm}(\xi)$ are absolutely convergent, whereas
   $\hsigma_{\pm}$ is defined as the indicated limit. These are computed using
   the asymptotic expansions in~\eqref{eqn230.63}, and the following lemma.
   \begin{lemma}\label{lem4.201}
     For $j\in\bbN\cup \{0\},$ let
     \begin{equation}
       F^{\pm}_j(\xi)=\int_{-\infty}^{\infty}\frac{e^{-iy_2\xi}\varphi_{\pm}(y_2)dy_2}{|y_2|^{j+\frac 12}}.
     \end{equation}
     The functions $F^{\pm}_j(\xi)$ are smooth away from $\xi=0,$ rapidly
     decaying along with all derivatives as $|\xi|\to\infty.$ They have analytic continuations to
     $\mp\Im\xi\geq 0,$ which decay like $e^{-d|\Im\xi|}.$  There are constants
     $a_j^{\pm}$ so that,
     near to $\xi=0,$ they have expansions of the form
     \begin{equation}
       F^{\pm}_j(\xi)=\frac{a_j^{\pm}\xi^{j}}{\sqrt{\xi}}+\psi^{\pm}_{j}(\xi),
     \end{equation}
     Here $\psi^{\pm}_{j}$ are entire functions.
     The $\sqrt{\xi}>0,$ for $0<\xi;$ for $F_j^+,$ the  $\sqrt{\xi}=-i\sqrt{-\xi},$ for
     $0>\xi,$ and for $F_j^-,$ the  $\sqrt{\xi}=i\sqrt{-\xi},$ for
     $0>\xi.$ 
   \end{lemma}
   \begin{proof}
     We first observe that
     \begin{equation}
       F^-_j(\xi)=\int_{-\infty}^{0}\frac{e^{-iy_2\xi}\varphi_{-}(y_2)dy_2}{(-y_2)^{j+\frac
           12}}=\int^{\infty}_{0}\frac{e^{iy_2\xi}\varphi_{+}(y_2)dy_2}{y_2^{j+\frac
           12}}=F_j^+(-\xi).
     \end{equation}
     so it suffices to consider $F^+_j(\xi).$ From the formula it is clear that the function
     $F^+_j(\xi)$ extends analytically to $\Im\xi\leq 0,$ and decays like
     $e^{-d|\Im\xi|}.$ For $\xi\neq 0$ we can integrate by parts arbitrarily
     often
     \begin{equation}
       F^+_j(\xi)=\frac{1}{(i\xi)^l}\int_{0}^{\infty}e^{-iy_2\xi}\pa^l_{y_2}\left(\frac{\varphi_+(y_2)}{y_2^{j+\frac 12}}\right)dy_2,
     \end{equation}
     from which the smoothness away from $\xi=0,$ and rapid decay statements are clear.

     If $j>0,$ then we observe that $y_2^{-(j+\frac 12)}=C_j\pa_{y_2}^jy_2^{-\frac
       12},$ and therefore integration by parts shows that
     \begin{equation}
       F_j^+(\xi)=C_j(-i\xi)^j\lim_{R\to\infty}\int_0^{R}\frac{e^{-iy_2\xi}\varphi_+(y_2)dy_2}{y_2^{\frac 12}}+\psi_j(\xi),
     \end{equation}
     where $\psi_j(\xi)$ is an entire function. To compute this limit observe
     that
     \begin{equation}
       \int_0^{R}\frac{e^{-iy_2\xi}(1-\varphi_+(y_2))dy_2}{y_2^{\frac 12}}
     \end{equation}
     is an entire function, independent of $R>d+2,$ and therefore
     \begin{equation}
         F_j^+(\xi)=C_j(-i\xi)^j\lim_{R\to\infty}\int_0^{R}\frac{e^{-iy_2\xi}dy_2}{y_2^{\frac 12}}+\tpsi_j(\xi),
     \end{equation}
     for $\tpsi_j(\xi)$  an entire function. For $\xi>0,$
     \begin{equation}
       \lim_{R\to\infty}\int_0^{R}\frac{e^{-iy_2\xi}dy_2}{y_2^{\frac 12}}=
       \lim_{R\to\infty}\frac{1}{\sqrt{\xi}}
       \int_0^{R\xi}\frac{e^{-iw}dw}{w^{\frac 12}}=e^{-\frac{\pi i}{4}} \sqrt{\frac{\pi}{\xi}},
     \end{equation}
     and for $\xi<0,$
     \begin{equation}
       \lim_{R\to\infty}\int_0^{R}\frac{e^{-iy_2\xi}dy_2}{y_2^{\frac 12}}=
       e^{\frac{\pi i}{4}} \sqrt{\frac{\pi}{|\xi|}}= e^{-\frac{\pi i}{4}} \sqrt{\frac{\pi}{\xi}},
     \end{equation}
     with $\sqrt{\xi}$ as defined above.
   \end{proof}

   Using this lemma and the asymptotic expansions we see that $\hsigma_+(\xi)$
   and $\htau_+(\xi)$ are smooth away from $\xi=k_1,$ have exponentially
   decaying, analytic continuations to the lower half plane, for $\xi$ near to $k_1,$ and any $N,$ 
   \begin{equation}\label{eqn252.101}
     \begin{split}
     \hsigma_+(\xi)&=\frac{p^{N}(\xi-k_1)}{\sqrt{\xi-k_1}}+\psi^{[N]}(\xi),\\
     \htau_+(\xi)&= \sqrt{\xi-k_1}q^{N}(\xi-k_1)+\theta^{[N]}(\xi),
     \end{split}
   \end{equation}
   for polynomials $p^{N},q^{N}$ and $\cC^{N+1}$-functions $\psi^{[N]},
   \theta^{[N]}.$ Similarly $\hsigma_-(\xi)$
   and $\htau_-(\xi)$ are smooth away from $\xi=-k_1,$ have exponentially
   decaying, analytic continuations to the upper half plane,  for $\xi$ near to $-k_1,$ and any $N,$ 
   \begin{equation}\label{eqn253.101}
     \begin{split}
     \hsigma_-(\xi)&=\frac{p^{N}(\xi+k_1)}{\sqrt{\xi+k_1}}+\psi^{[N]}(\xi),\\
     \htau_-(\xi)&= \sqrt{\xi+k_1}q^{N}(\xi+k_1)+\theta^{[N]}(\xi),
     \end{split}
   \end{equation}
     for polynomials $p^{N},q^{N}$ and $\cC^{N+1}$-functions $\psi^{[N]},
   \theta^{[N]}.$  With these computations, we can now analyze
   $u^{\pm}_0(r\eta), u^{\pm}_{1}(r\eta).$ 

   We use stationary phase to analyze the unbounded terms, which, in the Fourier
   representation, take the form:
   \begin{equation}\label{eqn130.666}
     \begin{split}
        u^{\pm}_0(r\eta)=\frac{i}{4\pi}\int_{-\infty}^{\infty}\frac{e^{ir(\eta_2\xi+\eta_1\sqrt{k_1^2-\xi^2})}\htau_{\pm}(\xi)
          d\xi}{\sqrt{k_1^2-\xi^2}},\\
        u^{\pm}_1(r\eta)=\frac{1}{4\pi}\int_{-\infty}^{\infty}e^{ir(\eta_2\xi+\eta_1\sqrt{k_1^2-\xi^2})}\hsigma_{\pm}(\xi)
        d\xi.
        \end{split}
   \end{equation}
   The change in the order of integrations needed to prove this formula for
   $u^{\pm}_0$ is easily justified: inserting the Fourier representation for
   $H^{(1)}_0$ from~\eqref{eqn113.763} leads an absolutely convergent double
   integral, to which Fubini's theorem immediately applies. The justification
   for the formula for $u^{\pm}_1$ is given after equation~\eqref{eqn133.40}.
   
Fix an $0<\epsilon\ll 1;$ if we assume that $|\eta_1|>\epsilon>0,$ then we can
choose $\mu>0,$ so that the stationary phase occurs at $k_1\eta_2\in
(-k_1+\mu,k_1-\mu).$ We divide these integrals into  parts, $u^{\pm}_{k0},$
supported in $(-k_1+\mu/2,k_1-\mu/2),$ which contains the stationary phase;
parts $u^{\pm}_{k\pm},$ supported in $[\pm k_1-\mu,\pm k_1+\mu],$ and  parts
$u^{\pm}_{k\infty}$ supported where $|\xi|>k_1+\mu/2.$ \Rd The unbounded parts have
smooth rapidly decaying integrands, that do not contain any points of stationary
phase, \Bk  hence it is not difficult to
show that
\begin{equation}
  u^{\pm}_{k\infty}(r\eta)=O(r^{-N})\text{ for any }N>0.
\end{equation}
It is also standard to show that the contributions from the stationary phase are
given by the  asymptotic expansions
\begin{equation}
     u^{\pm}_{k0}(r\eta)\sim\frac{e^{ik_1r}}{\sqrt{r}}\sum_{j=0}^{\infty}\frac{a^{\pm}_{kj}(\eta)}{r^j},
\end{equation}
with the coefficients $\{a^{\pm}_{kj}(\eta)\}$ smooth functions of $\eta,$ and
uniformly bounded error terms so long as $|\eta_1|>\epsilon>0.$ This leaves the
contributions from the singularities of the integrand at $\xi=\pm k_1.$

   We begin with $u^+_0(r\eta).$ With $\mu>0$ as above, let $\psi$ be a smooth function supported in
$[k_1-\mu,k_1+\mu],$ equal to 1 in $[k_1-\mu/2,k_1+\mu/2],$ and set
   \begin{equation}
      u^+_{0+}(r\eta)=\frac{i}{4\pi}\int_{-\infty}^{\infty}\frac{e^{ir(\eta_2\xi+\eta_1\sqrt{k_1^2-\xi^2})}\htau_+(\xi)
       \psi(\xi)d\xi}{\sqrt{k_1^2-\xi^2}}.
   \end{equation}
   We let $t=\sqrt{k_1^2-\xi^2},$ where $\xi<k_1,$ and $s=\sqrt{\xi^2-k_1^2},$
   where $\xi>k_1,$ to obtain
   \begin{equation}
     \begin{split}
       u^+_{0+}(r\eta)=&\frac{i}{4\pi}\int_{0}^{\sqrt{2\mu k_1-\mu^2}}\frac{e^{ir(\eta_2\sqrt{k_1^2-t^2}+\eta_1t)}
         \htau_+(\sqrt{k_1^2-t^2})
          \psi(\sqrt{k_1^2-t^2})dt}{\sqrt{k_1^2-t^2}}+\\
       &\frac{i}{4\pi}\int^{0}_{\sqrt{2k_1\mu+\mu^2}}
       \frac{e^{r(i\eta_2\sqrt{k_1^2+s^2}-\eta_1s)}\htau_+(\sqrt{k_1^2+s^2})
       \psi(\sqrt{k_1^2+s^2})ids}{\sqrt{k_1^2+s^2}}.
     \end{split}
   \end{equation}
   A moment's consideration shows that this sum of integrals is simply the
   contour integral of the ``$t$''-term on the contour in the complex plane:
   \begin{equation}
     \Lambda_{\mu}=[i\sqrt{2k_1\mu+\mu^2}, i0]\cup[0,\sqrt{2\mu k_1-\mu^2}],
   \end{equation}
   that is:
   \begin{equation}\label{eqn259.101}
      u^+_{0+}(r\eta)=\frac{i}{4\pi}\int_{\Lambda_{\mu}}\frac{e^{ir(\eta_2\sqrt{k_1^2-t^2}+\eta_1t)}
         \htau_+(\sqrt{k_1^2-t^2})
          \psi(\sqrt{k_1^2-t^2})dt}{\sqrt{k_1^2-t^2}}
   \end{equation}
   In Figure~\ref{fig1.2} $\Lambda_{\mu}$ is shown as the thick black ``L.'' The
   stationary point, shown in Figure~\ref{fig1.2}[a] as a black dot,  lies along
   the real axis where $t=k_1\eta_1.$

\begin{figure}[h]
  \centering
  \begin{subfigure}[t]{.45\textwidth}
     \includegraphics[width= 5.5cm]{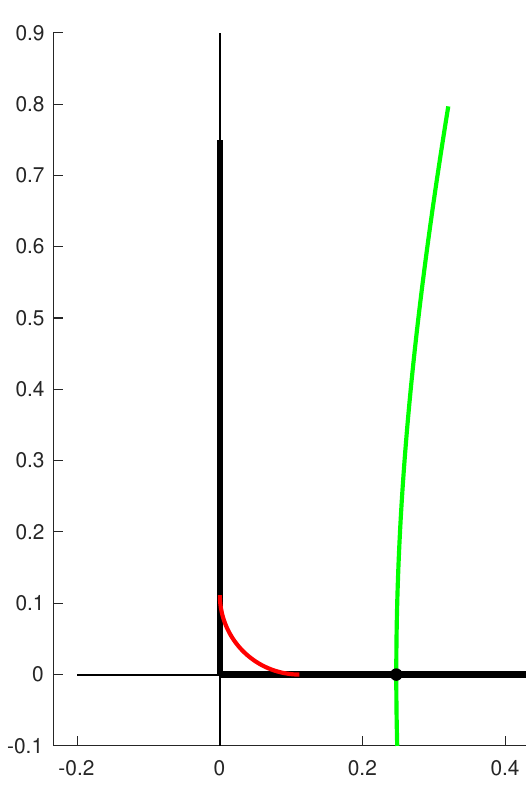}
    \caption{Deformation for the integral in~\eqref{eqn259.101},
      $\theta=0.25.$ The deformation $\Lambda_{\mu 1}$ includes the
      red curve.}
  \end{subfigure}\quad
   \begin{subfigure}[t]{.45\textwidth}
     \includegraphics[width= 5.5cm]{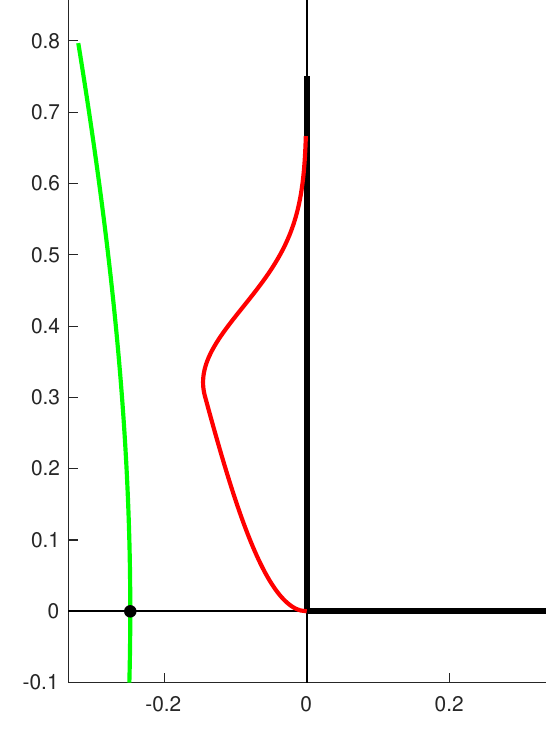}
    \caption{Deformation for the integral in~\eqref{eqn267.200}. The
      deformation $\Lambda_{\mu 2}$ includes the red curve}
    \end{subfigure}
    \caption{Plots of $\Lambda_{\mu}$ and the deformations needed to analyze the
      integrals in~\eqref{eqn259.101} and~\eqref{eqn267.200}. The green curves
      are the right, resp. left boundary of $Q_+.$}
   \label{fig1.2}
\end{figure}
To understand the asymptotics of
   this term we need to deform the contour, being careful to keep the real
   part of the phase, $i(t\eta_1+\sqrt{k_1^2-t^2}\eta_2),$ non-positive. We
   introduce the following change of variables $t=k_1\sin(\theta+z),$ where
\begin{equation}\label{eqn98.207}
  (\eta_1,\eta_2)=(\sin\theta,\cos\theta).
\end{equation}
\Rd Note that this definition of the polar angle differs from the usual choice: here
the angle measured clockwise from the positive $x_2$-axis. \Bk For $\eta_2>0,$ we
take $\theta\in (0,\frac{\pi}{2}).$ The hypothesis that $\eta_1$ is bounded away
from zero implies that $\theta$ is bounded away from $0.$

Note that $z=x+iy$ takes complex values. Using the fact that
\begin{equation}\label{eqn101.206}
  \sin(\theta+x+iy)=\sin(\theta+x)\cosh y+i\cos(\theta+x)\sinh y,
\end{equation}
we see that the contour $\Lambda_{\mu}$ corresponds to
$[-\theta+i\phi_0,-\theta+i0]\cup[-\theta,\theta_0-\theta],$ for a $\phi_0,\theta_0>0.$ In terms of these
variables the phase is
$$i(t\eta_1+\sqrt{k_1^2-t^2}\eta_2)=ik_1\cos(z),$$
which satisfies
\begin{equation}
 i \cos(x+iy)=i\cos x\cosh y+\sin x\sinh y.
\end{equation}
We see that the real part is non-positive for $t$ in the set
\begin{equation}
\left\{k_1\sin(\theta+x+iy):\:  x\in [-\pi,0],\quad y\in
[0,\infty)\right\}.
\end{equation}
On the other hand, the argument of $\htau_+$ is
\begin{equation}
  \sqrt{k_1^2-t^2}=k_1\cos(\theta+x+iy)=k_1[\cos(\theta+x)\cosh(y)-i\sin(\theta+x)\sinh(y)],
\end{equation}
which has non-positive imaginary part for $x\in [-\theta,\pi-\theta], y\geq 0.$ The intersection is
\begin{equation}
Q_+=\left\{k_1\sin(\theta+x+iy):\:  x\in [-\theta,0],\quad y\in
[0,\infty)\right\}
\end{equation}

We consider  deformations, $\Lambda_{\mu 1},$ of $\Lambda_{\mu},$ which replace
the corner of $\Lambda_{\mu}$ near $0$ with a smooth interpolant between the
$x$-axis and the $y$-axis, lying in $Q_+.$ An example is shown as the red curve in Figure~\ref{fig1.2}[a]. The
green curve is the right boundary of $Q_+.$ We can assume that $\psi\equiv
1$ in the support of the deformation, and therefore Cauchy's theorem implies
that the integral on the right hand side of~\eqref{eqn259.101} can be replaced
with
  \begin{equation}\label{eqn265.101}
      u^+_{0+}(r\eta)=\frac{i}{4\pi}\int_{\Lambda_{\mu 1}}\frac{e^{ir(\eta_2\sqrt{k_1^2-t^2}+\eta_1t)}
         \htau_+(\sqrt{k_1^2-t^2})
          \psi(\sqrt{k_1^2-t^2})dt}{\sqrt{k_1^2-t^2}}.
   \end{equation}
This is the integral of a smooth compactly supported function on a smooth arc
and there is no stationary phase within the support of the integrand, which is
therefore $O(r^{-N}),$ for any $N.$

We next consider the part of the integral near to $-k_1,$
   \begin{equation}
      u^+_{0-}(r\eta)=\frac{i}{4\pi}\int_{-(k_1+\mu)}^{\mu-k_1}\frac{e^{ir(\eta_2\xi+\eta_1\sqrt{k_1^2-\xi^2})}\htau_+(\xi)
       \psi(-\xi)d\xi}{\sqrt{k_1^2-\xi^2}}.
   \end{equation}
We can change variables as before and rewrite this integral as a contour
integral over $\Lambda_{\mu}:$
\begin{equation}\label{eqn267.200}
    u^+_{0-}(r\eta)=-\frac{i}{4\pi}\int_{\Lambda_{\mu}}\frac{e^{ir(\eta_ 1t-\eta_2\sqrt{k_1^2-t^2})}\htau_+(-\sqrt{k_1^2-t^2})
       \psi(\sqrt{k_1^2-t^2})dt}{\sqrt{k_1^2-t^2}};
\end{equation}
the phase is stationary at $-k_1\eta_1\notin\Lambda_{\mu}.$ As before,
we let $t=k_1\sin(\theta+z),$ 
the contour $\Lambda_{\mu}$ corresponds to $[-\theta+i\phi_0,-\theta+i0]\cup
[-\theta,\theta_0-\theta]. $ In this variable the phase is
\begin{equation}
 i(\eta_ 1t-\eta_2\sqrt{k_1^2-t^2})=-ik_1\cos(2\theta+z);
\end{equation}
with $z=x+iy,$ 
\begin{equation}\label{eqn269.200}
  -ik_1\cos(2\theta+z)=-ik_1\cos(2\theta+x)\cosh(y)-k_1\sin(2\theta+x)\sinh(y).
\end{equation}
If we deform the contour keeping
\begin{equation}
  -2\theta\leq x\leq -2\theta+\pi\text{ and }y\geq 0,
\end{equation}
then the real part of the phase remains non-positive.

The function $\htau_+(\xi)$ has an analytic extension to the lower half plane;
its argument is
\begin{equation}\label{eqn271.200}
  -\sqrt{k_1^2-t^2}=-k_1\cos(\theta+x)\cosh(y)+ik_1\sin(\theta+x)\sinh(y).
\end{equation}
If we deform the contour to a smooth curve, $\Lambda_{\mu 2},$ keeping 
 \begin{equation}
  -2\theta\leq x\leq -\theta\text{ and }y\geq 0,
 \end{equation}
 then $\Im-\sqrt{k_1^2-t^2}\leq 0,$ the integrand is analytic and the
 real part of the phase remains non-positive, see the red curve in
 Fig.~\ref{fig1.2}[b]. The stationary phase occurs where
 $t=-k_1\sin\theta,$ which lies outside the domain of integration, and
 therefore we conclude that this term is $O(r^{-N})$ for any $N.$ As
 noted earlier, all other portions of the integral defining
 $u^+_0(r\eta)$ are easily seen to be $O(r^{-N}).$

 We now turn to
  \begin{equation}
     u^-_0(r\eta)=\frac{i}{4\pi}\int_{-\infty}^{\infty}\frac{e^{ir(\eta_2\xi+\eta_1\sqrt{k_1^2-\xi^2})}\htau_-(\xi)
       d\xi}{\sqrt{k_1^2-\xi^2}}.
   \end{equation}
   The phase is stationary where $\xi=k_1\eta_2,$ and
   the integrand has singularities where $\xi=\pm k_1.$ If $|\eta_1|>\epsilon,$
   then the stationary point remains separated from the singularities, and
   contributes a standard asymptotic expansion
   \begin{equation}
     u^-_{00}(r\eta)\sim\frac{e^{ik_1r}}{\sqrt{r}}\sum_{j=0}^{\infty}\frac{a^-_{0j}(\eta)}{r^j}.
   \end{equation}
   We again need to examine the contributions from small neighborhoods of $\pm
   k_1.$ The principal differences with the previous case are that
   $\htau_-(\xi)$ has an analytic extension to $\Im \xi>0,$ and is singular at
   $\xi=-k_1.$

   As before, the contribution from near to $k_1$ is given by the contour
   integral
   \begin{equation}
      u^-_{0+}(r\eta)=\frac{i}{4\pi}\int_{\Lambda_{\mu}}\frac{e^{ir(\eta_2\sqrt{k_1^2-t^2}+\eta_1t)}
         \htau_-(\sqrt{k_1^2-t^2})
          \psi(\sqrt{k_1^2-t^2})dt}{\sqrt{k_1^2-t^2}}.
   \end{equation}
   Setting $t=\sin(\theta+z)$ we see that the phase is given by
   \begin{equation}
     i(\eta_2\sqrt{k_1^2-t^2}+\eta_1t)=ik_1\cos(x+iy)=ik_1\cos(x)\cosh(y)+k_1\sin(x)\sinh(y).
   \end{equation}
   Hence, the real part of the phase is non-positive if
   \begin{equation}
     -\pi\leq x\leq 0\text{ and }0\leq y,\text{ or }y=0.
   \end{equation}
   The argument of $\htau_-$ is
   \begin{equation}
     \sqrt{k_1^2-t^2}=k_1\cos(\theta+x+iy)=k_1[\cos(\theta+x)\cosh(y)-i\sin(\theta+x)\sinh(y)].
   \end{equation}
   To have $\Im\sqrt{k_1^2-t^2}\geq 0,$ we need to take
   \begin{equation}
     -\pi-\theta\leq x\leq-\theta\text{ and }0\leq y,\text{ or }y=0.
   \end{equation}
   In the intersection $-\pi\leq x\leq -\theta,\, 0\leq y.$ We can deform 
   $\Lambda_{\mu}$ to a smooth curve, like $\Lambda_{\mu 2},$ see the red curve in
   Figure~\ref{fig1.2}[b], along which the
   real part of the phase is non-positive, and the $\Im \sqrt{k_1^2-t^2}$ is
   non-negative, and therefore
   \begin{equation}
       u^-_{0+}(r\eta)=\frac{i}{4\pi}\int_{\Lambda_{\mu 2}}\frac{e^{ir(\eta_2\sqrt{k_1^2-t^2}+\eta_1t)}
         \htau_-(\sqrt{k_1^2-t^2})
          \psi(\sqrt{k_1^2-t^2})dt}{\sqrt{k_1^2-t^2}},
   \end{equation}
which is $O(r^{-N}).$

This leaves the contribution from near to $-k_1:$
\begin{equation}
    u^-_{0-}(r\eta)=-\frac{i}{4\pi}\int_{\Lambda_{\mu}}\frac{e^{ir(\eta_ 1t-\eta_2\sqrt{k_1^2-t^2})}\htau_-(-\sqrt{k_1^2-t^2})
       \psi(\sqrt{k_1^2-t^2})dt}{\sqrt{k_1^2-t^2}}.
\end{equation}
Using the calculations above for the phase,~\eqref{eqn269.200} and argument of
$\htau_-,$~\eqref{eqn271.200}, we see that we can deform this contour,
$t=\sin(\theta+x+iy),$ keeping $x\geq -\theta, y\geq 0,$ to a smooth curve, so
that the real part of the phase is non-positive and the argument of $\htau_-$
lies in the upper half plane. Contours of this type are shown in red in
Fig.~\ref{fig1.2}[a]. The integral on the deformed contour is the integral of a
smooth function on a smooth curve, which avoids the stationary phase and is
therefore $O(r^{-N}).$ This completes the proof that we get a complete
asymptotic expansion
\begin{equation}
  u_0(r\eta)\sim\frac{e^{ik_1r}}{\sqrt{r}}\sum_{j=0}^{\infty}\frac{a_{0j}(\eta)}{r^j},
\end{equation}
with smooth coefficients provided $\eta_1, \eta_2>0.$

We now consider the contributions of $u_1^{\epsilon}(r\eta),$ for
$\epsilon\in\{+,-\}.$ We use the Fourier transform to represent these two
terms
  \begin{equation}\label{eqn147.45}
     u^{\pm}_1(r\eta)=\frac{1}{4\pi}\int_{-\infty}^{\infty}e^{ir(\eta_2\xi+\eta_1\sqrt{k_1^2-\xi^2})}\hsigma_{\pm}(\xi)
       d\xi.
   \end{equation}
The change of order of integrations to get from the formula
in~\eqref{eqn238.200} to this representation requires some justification as the
double integral, using the Fourier representation of $
\pa_{y_1}H^{(1)}_0(|r\eta-(0,y_2)|),$ is no longer absolutely convergent. From
the asymptotic expansion of $\sigma(y_2),$ it suffices to show that
\begin{multline}\label{eqn133.40}
  \lim_{R\to\infty}\int_{-R}^R\left[\int_{-\infty}^{\infty}e^{i\xi(x_2-y_2)+ix_1\sqrt{k_1^2-\xi^2}}d\xi\right]\frac{\varphi_{\pm}(y_2)e^{\pm
    ik_1y_2}}{\sqrt{|y_2|}}dy_2=\\
  \int_{-\infty}^{\infty}
  \left[\int_{-\infty}^{\infty}\frac{e^{-iy_2\xi}\varphi_{\pm}(y_2)}{\sqrt{|y_2|}}dy_2\right]e^{i\xi
    x_2+ix_1\sqrt{k_1^2-\xi^2}}d\xi.
\end{multline}
To that end, for each fixed $R,$ we change the order of the integrations in the
left hand side. Using analyticity and Cauchy's theorem we can
then deform the contour  of the $\xi$-integrand into the lower  half
plane near to $\xi=k_1,$  for $+,$  or into the upper half
plane near to $\xi=-k_1,$  for $-.$  Let the deformed contours be denoted
$\Gamma_{\pm}.$ We consider the $+$ case where  we need to estimate
\begin{equation}
  \int_{\Gamma_+}
  \left[\int_{R}^{\infty}\frac{e^{i(k_1-y_2)\xi}}{\sqrt{|y_2|}}dy_2\right]e^{i\xi
    x_2+ix_1\sqrt{k_1^2-\xi^2}}d\xi.
\end{equation}
The inner integral is uniformly bounded by $C/\sqrt{R};$ for $x_1>0$ the
$\xi$-integral is absolutely convergent and therefore this term is bounded by
$C/\sqrt{R},$ which justifies the change in the order integrations
\begin{multline}\label{eqn133.401}
  \lim_{R\to\infty}\int_{-R}^R\left[\int_{-\infty}^{\infty}e^{i\xi(x_2-y_2)+ix_1\sqrt{k_1^2-\xi^2}}d\xi\right]\frac{\varphi_{\pm}(y_2)e^{\pm
    ik_1y_2}}{\sqrt{|y_2|}}dy_2=\\
\int_{\Gamma_+}
  \left[  \int_{0}^{\infty}\frac{e^{-iy_2\xi}\varphi_{+}(y_2)}{\sqrt{|y_2|}}dy_2\right]e^{i\xi
    x_2+ix_1\sqrt{k_1^2-\xi^2}}d\xi.
\end{multline}
Using the analyticity of the $\xi$-integrand and the results of Lemma~\ref{lem4.201} we
can deform the integral back to the real axis. The $-$ case is essentially the
same, which completes the justification of~\eqref{eqn147.45}.

 There are several minor differences differences between $u^{\pm}_0$ and
 $u^{\pm}_1,$  which require comment, but the foregoing analyses of $u^{\pm}_0$
 apply, {\em mutatis mutandis}, to  $u^{\pm}_1$ as well. The main differences
 between the integrands defining these functions are near $\xi=\pm k_1:$
 \begin{equation}
   \frac{\htau_{\pm}(\xi)d\xi}{\sqrt{k_1^2-\xi^2}}=\frac{\sqrt{k_1\mp \xi}\,\gamma_{\pm}(\xi)d\xi}{\sqrt{k^2_1-\xi^2}}
 \end{equation}
 whereas
 \begin{equation}
   \hsigma_{\pm}(\xi)d\xi=\frac{\sqrt{k_1\pm\xi}\,\tgamma_{\pm}(\xi)d\xi}{\sqrt{k^2_1- \xi^2}},
 \end{equation}
 where $\gamma_{\pm}(\xi), \tgamma_{\pm}(\xi)$ are smooth near to $\xi=\pm k_1.$ We change variables
 near these points setting $t=\sqrt{k_1^2-\xi^2},$ so that these terms become
 the contour integrals
 \begin{equation}\label{eqn290.203}
   \begin{split}
     u^{+}_{1+}(r\eta)&=\int_{\Lambda_{\mu}}\frac{e^{ir(\eta_2\sqrt{k_1^2-t^2}+\eta_1t)}\hsigma_+(\sqrt{k_1^2-t^2})
       \psi(\sqrt{k_1^2-\mu^2})tdt}{\sqrt{k_1^2-t^2}},\\
     u^{+}_{1-}(r\eta)&=-\int_{\Lambda_{\mu}}\frac{e^{ir(\eta_1t-\eta_2\sqrt{k_1^2-t^2})}\hsigma_+(-\sqrt{k_1^2-t^2})
       \psi(\sqrt{k_1^2-\mu^2})tdt}{\sqrt{k_1^2-t^2}},\\
     u^{-}_{1+}(r\eta)&=\int_{\Lambda_{\mu}}\frac{e^{ir(\eta_2\sqrt{k_1^2-t^2}+\eta_1t)}\hsigma_-(\sqrt{k_1^2-t^2})
       \psi(\sqrt{k_1^2-\mu^2})tdt}{\sqrt{k_1^2-t^2}},\\
     u^{-}_{1-}(r\eta)&=-\int_{\Lambda_{\mu}}\frac{e^{ir(\eta_1t-\eta_2\sqrt{k_1^2-t^2})}\hsigma_-(-\sqrt{k_1^2-t^2})
       \psi(\sqrt{k_1^2-\mu^2})tdt}{\sqrt{k_1^2-t^2}}.
   \end{split}
 \end{equation}
 
 First note that, with the square root as defined in Lemma~\ref{lem4.201}, for
 $t\in\Lambda_{\mu},$ we have for $\hsigma_+$ that
 \begin{equation}\label{eqn128.210}
   \sqrt{\sqrt{k_1^2-t^2}-k_1}=\frac{-it}{\sqrt{\sqrt{k_1^2-t^2}+k_1}},
 \end{equation}
and for $\hsigma_-$ that
 \begin{equation}\label{eqn128.2100}
   \sqrt{k_1-\sqrt{k_1^2-t^2}}=\frac{t}{\sqrt{\sqrt{k_1^2-t^2}+k_1}}.
 \end{equation}
 These functions extend smoothly to a neighborhood, $W,$ of $\Lambda_{\mu}$
 in $\Re t\geq 0.$ Using~\eqref{eqn252.101} and~\eqref{eqn253.101}, we see that
 this implies that both functions
 $\hsigma_+(\sqrt{k_1^2-t^2})t$ and $\hsigma_-(-\sqrt{k_1^2-t^2})t$ are smooth
 in $W;$ for $\hsigma_+$
 \begin{equation}
   \frac{t}{\sqrt{\sqrt{k_1^2-t^2}-k_1}}=\frac{t\sqrt{\sqrt{k_1^2-t^2}+k_1}}{-it}=i\sqrt{\sqrt{k_1^2-t^2}+k_1},
 \end{equation}
  and for $\hsigma_-$
 \begin{equation}
   \frac{t}{\sqrt{k_1-\sqrt{k_1^2-t^2}}}=\frac{t\sqrt{\sqrt{k_1^2-t^2}+k_1}}{t}=\sqrt{\sqrt{k_1^2-t^2}+k_1}.
 \end{equation}
 This shows that there are no issues of integrability or smoothness near $t=0$ in the
 integrals defining $ u^{+}_{1+}(r\eta)$ and $ u^{-}_{1-}(r\eta).$
 
 We deform the contours in the formul{\ae} for $u^{\pm}_{1\pm}(r\eta)$
 to conclude that these terms are $O(r^{-N}),$ for any $N.$ This
 completes the proof that
 \begin{equation}
   u_1(r\eta)\sim\frac{e^{ik_1r}}{\sqrt{r}}\sum_{j=0}^{\infty}\frac{a_{1j}(\eta)}{r^j},
 \end{equation}
 with the coefficients smooth functions of $\eta,$ where $\eta_1,\eta_2 >0.$
 
 Similar arguments apply to analyze these terms where $\eta_2<0.$ We
 now take $\eta=(\sin\theta,\cos\theta),$ for $\theta\in
 (\frac{\pi}{2},\pi).$ After changing to the $t$-variable, the
 stationary phases for $u^+_{0+}, u^-_{0+}$ occur at $t=-k_1\eta_1,$
 which is outside the contour $\Lambda_{\mu},$ whereas, for $u^+_{0-},
 u^-_{0-}$ occur at $k_1\eta_1,$ which lies within
 $\Lambda_{\mu}.$ Using contour deformation arguments like those used
 above, we can show that the contributions of all of these terms are
 $O(r^{-N}),$ for any $N>0.$ The details of these arguments are given
 in the next section, where we show that the asymptotic expansions are
 uniformly valid as $\eta_2\to \pm 1.$ There is no essential difference
 if $\eta_1<0.$

 \subsection{Uniform Estimates as $\eta_1\to 0$}\label{sec7.4.83}
 To complete the analysis of $u_0^{l,r}(r\eta), u_1^{l,r}(r\eta)$ we need to
 consider what happens as $\eta_1\to 0^{+};$ as
 $\eta=(\sin\theta,\cos\theta),$ this corresponds to $\theta\to 0^+,$ or
 $\theta\to\pi^-,$ depending on whether $\eta_2>0,$ or $\eta_2<0.$ \Rd  As in
 the previous section, we only consider the $x_1>0$-case. The other case is
 proved by essentially identical arguments.\Bk

 The contour deformations from the previous section can be adapted to prove
 uniform estimates. In the integrals considered above, the stationary phase, in
 the $t$-variable, occurs at either $k_1\eta_1$ or $-k_1\eta_1,$ where the
 contour of integration, $\Lambda_{\mu}$ consists of 2 line segments meeting at
 $0.$ It is complicated to prove uniform estimates because the stationary phase
 moves as $\eta$ varies, and, as $\eta_1\to 0,$ ($\theta\to 0,$ or $\pi$)
 it converges to the point on $\Lambda_{\mu}$ where the two line segments meet.  If
 the stationary phase always occurred at such an intersection, then 
 Lemma~\ref{lem5.202} below can be used  to prove uniform estimates. As we show
 below, this can be accomplished by deforming the contours. 

 We begin with the lemma.
 \begin{lemma}\label{lem5.202}
   Let $f(x,\theta)\in\cC^{\infty}_c([0,1)\times [0,\delta]),$  for a
     $\delta>0,$ the integral
     \begin{equation}
       F(r,\theta)=\int_{0}^{1}e^{ir x^2}f(x,\theta)dx\sim
       \sum_{j=0}^{\infty}\frac{a_j(\theta)}{r^{\frac{(j+1)}{2}}},
     \end{equation}
     defines a $\cC^{\infty}$-function. The coefficients satisfy
     $a_{j}(\theta)\in\cC^{\infty}([0,\delta]).$ \Rd The error terms
     are uniformly bounded for $\theta\in [0,\delta].$ For any
     $l,m$ the derivatives of $F,$ $\pa_{\theta}^l\pa_r^mF(r,\theta),$
     have uniform asymptotic expansions as $r\to\infty$ obtained by
     differentiating the expansion for $F$ term-by-term.\Bk
 \end{lemma}

 \begin{proof}
   The smoothness of $F$ is immediate from the formula defining it.  For each
   $N>0$ we have the Taylor expansion, in $x,$ for $f$
     \begin{equation}
       f(x,\theta)=\sum_{j=0}^{2N} \frac{\pa_x^jf(0,\theta)}{j!}x^j+x^{2N+1}R_N(x,\theta),
     \end{equation}
     where $R_N\in\cC^{\infty}([0,1)\times [0,\delta]).$ Let
       $\psi\in\cC^{\infty}_c([0,1))$ equal $1$ on $\supp f(\cdot,\theta),$ for
         $\theta\in[0,\delta].$ We can rewrite $F(r,\theta)$ as
       \begin{equation}\label{eqn171.263}
         F(r,\theta)=\sum_{j=0}^{2N}\frac{\pa_x^jf(0,\theta))}{j!}\int_{0}^1e^{irx^2}x^j\psi(x)dx+
         \int_{0}^1 e^{irx^2}x^{2N+1}R_N(x,\theta)\psi(x)dx.
       \end{equation}
       \Rd 

        It is a classical result,
       proved by letting $y=x^2$ and integration by parts that, for $j\in\bbN\cup\{0\},$
       \begin{equation}
         \begin{split}
           \int_{0}^1e^{irx^2}x^{2j+1}\psi(x)dx&=
           \left(\frac{1}{i}\pa_r\right)^j\int_{0}^1e^{irx^2}x\psi(x)dx\\
        &=
           \left(\frac{1}{i}\pa_r\right)^j\left[\frac{i}{2r}+\frac{i}{2r}\int_0^1e^{iry}\tpsi(y)dy\right],
           \end{split}
       \end{equation}
       where $\tpsi\in \cC^{\infty}_c(0,1).$ The error term is therefore
       $O(r^{-M})$ for any $M>0.$ For even order terms we use a
       standard contour deformation argument, see~\cite{ZworskiSA2012}, to
       conclude that
        \begin{equation}\begin{split}
          \int_{0}^1e^{irx^2}x^{2j}\psi(x)dx&=
          \left(\frac{1}{i}\pa_r\right)^j\left[\int_{0}^1e^{irx^2}\psi(x)dx\right]\\
          &=
          \left(\frac{1}{i}\pa_r\right)^j\left[\frac{c_0}{\sqrt{r}}+\int_{\Upsilon}
            (\psi(x)-1)e^{irx^2}dx\right],
          \end{split}
        \end{equation}
        where $\Upsilon$ is a contour in the complex plane starting at $0$ that lies along the
        real axis within the support of $\psi$ and smoothly interpolates, within
        the first quadrant, to the line
        $z=\{(t+1)e^{\frac{\pi i}{4}}:t\in[0,\infty)\},$ so that for large
          arguments the integrand is $-e^{-r(t+1)^2}.$ From this formula it is
          clear that, for any $j,$ the remainder terms are $O(r^{-M})$ for any $M>0.$

          These formul{\ae} show that there are universal
          constants $\{\tc_j\}$ so that
        \begin{equation}\label{eqn174.263}
          F(r,\theta)=\sum_{j=0}^{2N}\frac{\tc_j
            \pa_x^jf(0,\theta)}{r^{\frac{j+1}{2}}}+
           \int_{0}^1 e^{irx^2}x^{2N+1}R_N(x,\theta)\psi(x)dx+O(r^{-M})\text{
             for any }M.
        \end{equation}
              The smoothness of the coefficients of the expansion for
              $\theta\in[0,\delta]$ is immediate from~\eqref{eqn174.263}. To
              complete the proof of the lemma we need to estimate the 
              $R_N$-term in~\eqref{eqn174.263}.

                Integrating by parts $(N+1)$ times, we argue as in the proof of
                Lemma~\ref{lem0} to show that
       \begin{equation}
       \int_{0}^1 e^{irx^2}x^{2N+1}R_N(x,\theta)\psi(x)dx=O(r^{-(N+1)}),
       \end{equation}
     uniformly for $\theta\in [0,\delta].$ From~\eqref{eqn171.263} and the
     smoothness of $R_N(x,\theta)$ it is clear that
       \begin{multline}
         \pa_{\theta}^l\pa_r^{m} F(r,\theta)=
         \sum_{j=0}^{2N}\frac{\pa_{\theta}^l\pa_x^jf(0,\theta))}{j!}\int_{0}^1(ix^2)^me^{irx^2}x^j\psi(x)dx+\\
         \int_{0}^1
         (ix^2)^me^{irx^2}x^{2N+1}\pa_{\theta}^lR_N(x,\theta)\psi(x)dx.
       \end{multline}
       The final statement of the lemma follows easily from this formula, the
       theorem of Coddington and Levinson quoted at the end of
       Section~\ref{ss.asympt}, and the foregoing observations. \Bk
 \end{proof}
 \begin{remark}\label{rmk7.43} We apply this lemma to integrals of analytic functions over a
   smooth family of contours $\{\Lambda^{\theta}_{\mu\pm}:\:\theta\in
   [0,\delta]\}.$ By explicitly parametrizing the contours, removing the phase
   factor $e^{irk_1},$ and choosing changes
   of variable, which depend smoothly on $\theta,$ the phase can be normalized to
   $x^2.$ Hence these integrals can be reduced to the form in the lemma, that is  to  integrals
   of functions depending smoothly on $\theta$ over a fixed interval. We leave
   the details of this reduction to the interested reader.  In our applications
   of this lemma we have a pair of contours meeting at a point; we are always
   able to conclude that the coefficients of terms with integral exponents, 
   $\{r^{-j}:\: j\in\bbN\},$ cancel and these terms are therefore absent.
 \end{remark}

 In the case we are now considering the
 stationary point, in the original $\xi$ variable,   is moving toward an endpoint of the interval
 $[-k_1,k_1].$ In light of that, we modify the definitions of
 $u^{\epsilon}_j,$ for $\epsilon\in \{0,+,-\}.$ We assume that
 $|\eta_1|<\epsilon\ll 1,$ and choose $\mu>0,$ so that
 \begin{equation}
   |k_1\eta_1|<\frac{\mu}{2}.
 \end{equation}
 With this choice we easily show that
 \begin{equation}
   u^0_{j}(r\eta)=O(r^{-N})\text{ for any }N,
 \end{equation}
 uniformly as $\eta_1\to 0^{+}.$

 The integrals defining $u^{\pm}_j$ now contain the stationary
 point. It  turns out that some parts of these terms are still uniformly
 $O(r^{-N}),$ whereas others  have expansions like those in
 Theorem~\ref{thm3.75}.  For example, we easily show that the unbounded contributions,
 \begin{equation}
   u^{\pm}_{j\infty}(r\eta)=O(r^{-N})\text{ for any }N>0, j=0,1,
 \end{equation}
 uniformly as $\eta_1\to 0^+.$ Below we show that the contour in the integrals
 defining the functions $u^{\pm}_{j\pm},$ $j=0,1$ can be deformed either to a
 fixed smooth curve, containing, or avoiding the stationary phase point, or to a
 smooth family of intersecting curves passing through the stationary phase
 points. Applying Lemma~\ref{lem5.202} in the latter case gives the desired
 uniform expansions as $\eta_1\to 0^{+}.$
 In this section we follow the order of terms considered in the previous section,
 with more detailed discussions for the $\eta_2<0$ case.

 We start with $\eta_2>0,$ and $u^{+}_{0+}:$
 \begin{equation}\label{eqn300.203}
   u^{+}_{0+}(r\eta)=\frac{i}{4\pi}\int_{\Lambda_\mu}\frac{e^{ir(\eta_2\sqrt{k_1^2-t^2}+
         \eta_1t)}\htau_{+}(\sqrt{k_1^2-t^2})\psi(\sqrt{k_1^2-t^2})dt}{\sqrt{k_1^2-t^2}}.
 \end{equation}
 For $\eta_1>0,$ integration on the deformed contour $\Lambda_{\mu 1}$
 shows that
 \begin{equation}\label{eqn301.203}
   u^{+}_{0+}(r\eta)\sim\frac{e^{ik_1r}}{\sqrt{r}}\sum_{j=0}^{\infty}\frac{a_{0j}^+(\eta)}{r^j},
 \end{equation}
 with $a^+_{0j}\in\cC^{\infty}((0,\delta])$ for a $\delta>0.$ We can
   further deform the contour, setting
   \begin{equation}
     t=k_1\sin(\theta(1+x(s))+iy(s)).
   \end{equation}

   \begin{figure}[h]
  \centering
  \begin{subfigure}[t]{.45\textwidth}
     \includegraphics[height= 8cm]{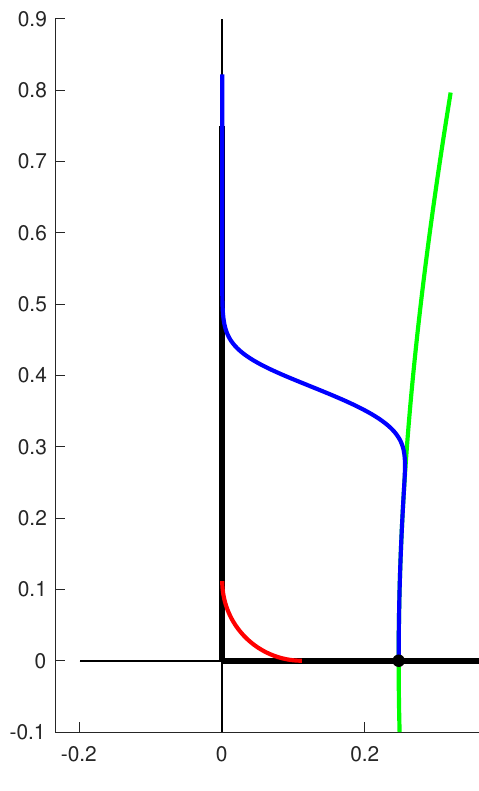}
    \caption{
      $\theta=0.25.$ }
  \end{subfigure}\quad
   \begin{subfigure}[t]{.45\textwidth}
     \includegraphics[height= 8cm]{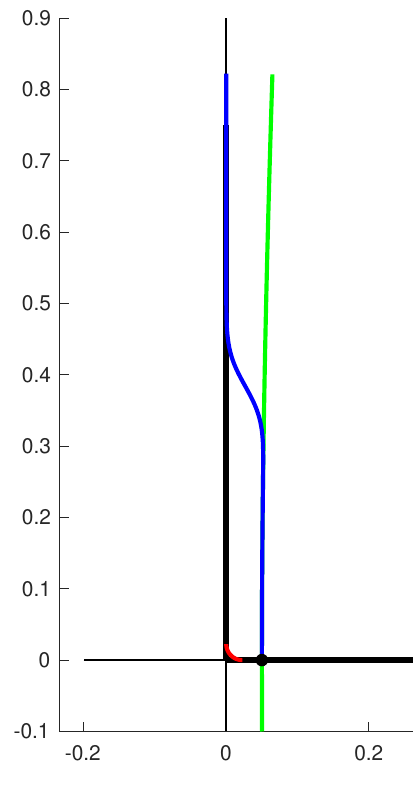}
    \caption{$\theta=.05.$}
    \end{subfigure}
    \caption{Contour deformations for the integral
      in~\eqref{eqn300.203}. The deformation $\Lambda^{\theta}_{\mu+}$
      includes the blue curve; the stationary point is the black dot.}
   \label{fig6.203}
\end{figure}

If $-\theta\leq x\leq 0$ and $0\leq y,$ then the real part of the
phase is non-positive and $\Im\sqrt{k_1^2-t^2}\leq 0.$ More
explicitly, we choose $x\in\cC^{\infty}([0,1]),$ a monotone increasing
functions with
\begin{equation}
  x(s)=
  \begin{cases}
    -1\text{ for }s\in [0,\frac 14],\\
     0\text{ for }s\in [\frac 34,1],
  \end{cases}
\end{equation}
and $y(s)=(1-s)\phi_0,$ where $k_1\sinh(\phi_0)=\mu.$ The smooth family of
deformed curves, $\Lambda^{\theta}_{\mu+},$
is then given by
\begin{equation}
  t\in\{k_1\sin(\theta(1+x(s))+iy(s)):\: s\in [0,1]\}\cup
  [k_1\sin\theta,\mu],\text{ for }\theta\in [0,\delta].
\end{equation}
Examples are shown as the blue curves in Figure~\ref{fig6.203}.

Cauchy's theorem implies that
\begin{equation}\label{eqn305.203}
   u^{+}_{0+}(r\eta)=\frac{i}{4\pi}\int_{\Lambda^{\theta}_{\mu+}}\frac{e^{ir(\eta_2\sqrt{k_1^2-t^2}+
         \eta_1t)}\htau_{+}(\sqrt{k_1^2-t^2})\psi(\sqrt{k_1^2-t^2})dt}{\sqrt{k_1^2-t^2}}.
\end{equation}
As follows from~\eqref{eqn128.210}, $\htau_+(\sqrt{k_1^2-t^2})$ is a smooth
function in a closed set, $W$ which includes the region swept out by the
deformations, $\{\Lambda_{\mu+}^{\theta}:\: \theta\in [0,\delta]\}.$ Therefore, with appropriate changes of variable,
we can apply Lemma~\ref{lem5.202} to the two segments of
$\Lambda_{\mu+}^{\theta}$ that meet at $\eta_1=k_1\sin\theta,$ to conclude that
the asymptotic expansion in~\eqref{eqn301.203} holds uniformly down to
$\theta=0,$ and the coefficients are smooth functions of $\theta\in[0,\delta].$

We next consider
\begin{equation}\label{eqn306.203}
  u^+_{0-}(r\eta)=-\frac{i}{4\pi}\int_{\Lambda_{\mu}}
  \frac{e^{ir(\eta_ 1t-\eta_2\sqrt{k_1^2-t^2})}\htau_+(-\sqrt{k_1^2-t^2})
       \psi(\sqrt{k_1^2-t^2})dt}{\sqrt{k_1^2-t^2}}.
\end{equation}
From our previous analysis, we know that, if $\eta_1>0,$ then
$u^+_{0-}(r\eta)=O(r^{-N})$ for any $N,$ as the stationary point lies
at $-k_1\eta_1,$ which does not belong to the contour $\Lambda_{\mu 2}.$ For this case we set
\begin{equation}
  t=k_1\sin(\theta(1+x(s))+iy(s));
\end{equation}
if $-2\leq x(s)\leq -1$ and $0\leq y(s),$ then the real
part of the phase is non-positive and $\Im-\sqrt{k_1^2-t^2}\leq 0.$
The smooth family of deformed contours,
\begin{equation}
\Lambda_{\mu-}^{\theta}=\{k_1\sin(\theta(1+x(s))+iy(s)):\: s\in
       [0,1]\}\cup [-k_1\sin\theta,\mu],\text{ with }\theta\in [0,\delta],
\end{equation}
with corner at $-k_1\sin\theta,$ is given by $x\in\cC^{\infty}([0,1]),$ a
monotone decreasing function with
\begin{equation}
  x(s)=
  \begin{cases}
    &-1\text{ for }s\in [0,\frac 14],\\
      &-2\text{ for }s\in [\frac 34,1],
  \end{cases}
\end{equation}
and $y(s)=(1-s)\phi_0.$ Examples are shown in Figure~\ref{fig7.203};
the blue curves are parts of $\Lambda_{\mu -}^{\theta}.$
 \begin{figure}[h]
  \centering
  \begin{subfigure}[t]{.45\textwidth}
     \includegraphics[height= 8cm]{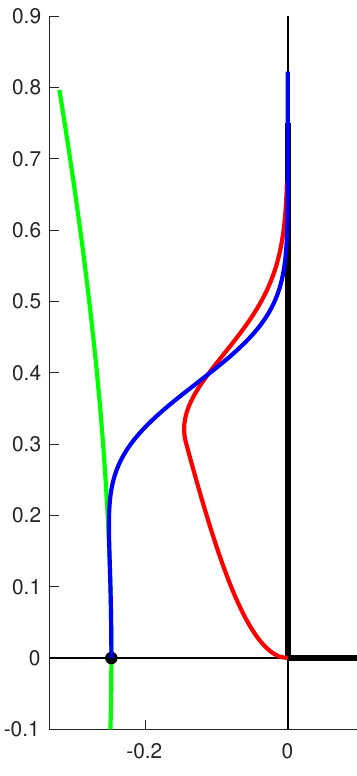}
    \caption{
      $\theta=0.25.$ }
  \end{subfigure}\quad
   \begin{subfigure}[t]{.45\textwidth}
     \includegraphics[height= 8cm]{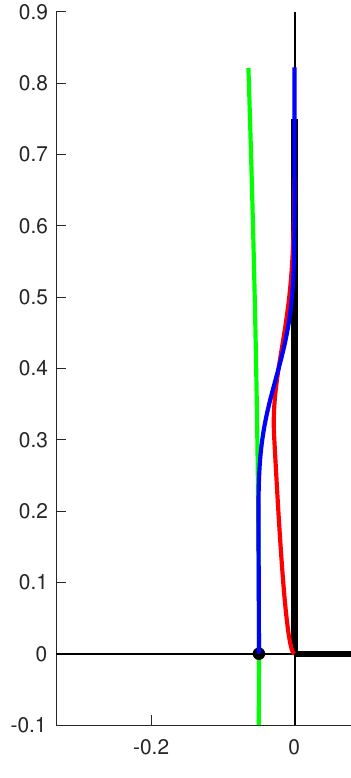}
    \caption{$\theta=.05.$}
    \end{subfigure}
    \caption{Contour deformations for the integral
      in~\eqref{eqn306.203}. The deformation $\Lambda^{\theta}_{\mu-}$
      includes the blue curve.}
   \label{fig7.203}
 \end{figure}

 The function $\htau_+(\xi)$ is singular where $\xi=k_1,$ so
 $\htau_+(-\sqrt{k_1^2-t^2})$ is smooth throughout the regions swept out by the
 deformations, $\{\Lambda_{\mu-}^{\theta}:\: \theta\in [0,\delta]\}.$ Hence we
 can once again apply Lemma~\ref{lem5.202} to the smooth family of integrals
 \begin{equation}\label{eqn310.204}
   u^+_{0-}(r\eta)=-\frac{i}{4\pi}\int_{\Lambda_{\mu-}^{\theta}}
  \frac{e^{ir(\eta_ 1t-\eta_2\sqrt{k_1^2-t^2})}\htau_+(-\sqrt{k_1^2-t^2})
       \psi(\sqrt{k_1^2-t^2})dt}{\sqrt{k_1^2-t^2}},
 \end{equation}
 to conclude that $u^{+}_{0-}(r\eta)$ has asymptotic expansions
 uniformly down to $\eta_1=0.$ The novelty here is that 
 $u^{+}_{0-}(r\eta)=O(r^{-N})$ for any $N.$ Hence for $\eta_1>0,$ all terms of
 the expansions from the two segments cancel, and we conclude that
 $u^{+}_{0-}(r\eta)=O(r^{-N}),$ for any $N>0,$ uniformly as $\eta_1\to 0^+.$

 We now turn to
 \begin{equation}
       u^-_{0+}(r\eta)=\frac{i}{4\pi}\int_{\Lambda_{\mu}}\frac{e^{ir(\eta_2\sqrt{k_1^2-t^2}+\eta_1t)}
         \htau_-(\sqrt{k_1^2-t^2})
          \psi(\sqrt{k_1^2-t^2})dt}{\sqrt{k_1^2-t^2}}.
 \end{equation}
 We deform the contour with $t=k_1\sin(\theta+x(s)+iy(s)).$ To keep
 the real part of the phase non-positive we need to have
 \begin{equation}
   -\pi\leq x(s)\leq 0\text{ and }0\leq y(s);
 \end{equation}
 to insure that $\Im\sqrt{k_1^2-t^2}\geq 0,$ we need to have
 \begin{equation}
   -\pi-\theta\leq x(s)\leq-\theta\text{ and }0\leq y(s).
 \end{equation}
 We now use a \emph{fixed} contour like $\Lambda_{\mu 2},$ (see the red curve in
 Figure~\ref{fig1.2}][b]) with $-\frac{\pi}{2}<x(s)+\theta\leq 0,$ as described
     above to obtain
  \begin{equation}
       u^-_{0+}(r\eta)=\frac{i}{4\pi}\int_{\Lambda_{\mu2}}\frac{e^{ir(\eta_2\sqrt{k_1^2-t^2}+\eta_1t)}
         \htau_-(\sqrt{k_1^2-t^2})
          \psi(\sqrt{k_1^2-t^2})dt}{\sqrt{k_1^2-t^2}}.
  \end{equation}
  The stationary points, $\{k_1\sin\theta:\:\theta\in[0,\delta]\},$ are now
  interior points of the contour, $\Lambda_{\mu2},$ and therefore we have a
  standard asymptotic expansion
  \begin{equation}
    u^-_{0+}(r\eta)=\frac{e^{ik_1r}}{\sqrt{r}}\sum_{j=0}^{\infty}\frac{a_{0j}^{-}(\eta)}{r^j},
  \end{equation}
  where the coefficients are smooth as $\eta_1\to 0^+.$

  The leaves
  \begin{equation}
    u^-_{0-}(r\eta)=-\frac{i}{4\pi}\int_{\Lambda_{\mu}}
    \frac{e^{ir(\eta_ 1t-\eta_2\sqrt{k_1^2-t^2})}\htau_-(-\sqrt{k_1^2-t^2})
       \psi(\sqrt{k_1^2-t^2})dt}{\sqrt{k_1^2-t^2}};
  \end{equation}
  the stationary point lies at $t=-k_1\sin\theta,$ which is outside
  the contour $\Lambda_{\mu}.$ If we deform the contour with
  $t=k_1\sin(\theta+x(s)+iy(s)),$ keeping $-\theta\leq x(s)\leq
  \pi-2\theta,$ and $0\leq y(s),$ then the real part of the phase will
  remain non-positive and the $\Im -\sqrt{k_1^2-t^2}\geq 0.$ Hence we
  can fix the choice of a curve like $\Lambda_{\mu 1},$ lying in the
  first quadrant, such that, for all $\eta_1>0,$ we have
  \begin{equation}
    u^-_{0-}(r\eta)=-\frac{i}{4\pi}\int_{\Lambda_{\mu 1}}
    \frac{e^{ir(\eta_ 1t-\eta_2\sqrt{k_1^2-t^2})}\htau_-(-\sqrt{k_1^2-t^2})
       \psi(\sqrt{k_1^2-t^2})dt}{\sqrt{k_1^2-t^2}}.
  \end{equation}
  The stationary point is a fixed positive distance from this contour
  as $\theta\to 0,$ and therefore
  \begin{equation}
    u^-_{0-}(r\eta)=O(r^{-N})\text{ for any }N>0,
  \end{equation}
  uniformly as $\eta_1\to 0^+.$

  Recalling that $t\hsigma_+(\sqrt{k_1^2-t^2})$ and
  $t\hsigma_-(-\sqrt{k_1^2-t^2})$ have the same regularity properties as
  $\htau_+(\sqrt{k_1^2-t^2})$ and $\htau_-(-\sqrt{k_1^2-t^2})$ near to $t=0,$ it
  follows from the formul{\ae} in~\eqref{eqn290.203} and the foregoing arguments
  that we have expansions
  \begin{equation}
    \begin{split}
      &u^+_{1+}(r\eta)\sim
   \frac{e^{ik_1r}}{\sqrt{r}}\sum_{j=0}^{\infty}\frac{a^+_{1j}(\eta)}{r^j},\quad
   u^+_{1-}(r\eta)=O(r^{-N})\text{ for any } N>0,\\
   &u^-_{1+}(r\eta)\sim
   \frac{e^{ik_1r}}{\sqrt{r}}\sum_{j=0}^{\infty}\frac{a^-_{1j}(\eta)}{r^j},\quad
   u^-_{1-}(r\eta)=O(r^{-N})\text{ for any } N>0,
    \end{split}
  \end{equation}
  with all expansions uniformly valid as $\eta_1\to 0^+.$

  This completes the proof of Theorem~\ref{thm3.75}, assuming the $\eta_2>0,$
  and $\eta_1\to 0^+.$ In the remainder of this section we explain the
  modifications that are needed if $\eta_2<0,$ and $\eta_1\to 0^+.$
  The proof with $\eta_1<0$ is essentially identical, and is left to the
  reader.

    The most significant difference that arises if $\eta_2<0$  is that phase functions,
  $\eta_2\sqrt{k_1^2-t^2}\pm \eta_1t,$ have  their stationary points at
  $\mp k_1\eta_1,$ rather than $\pm k_1\eta_1.$ We begin with
  \begin{equation}
    u^+_{0+}(r\eta)=\int_{\Lambda_{\mu}}\frac{e^{ir(\eta_2\sqrt{k_1^2-t^2}+\eta_1t)}\htau_+(\sqrt{k_1^2-t^2})dt}
    {\sqrt{k_1^2-t^2}}.
  \end{equation}
  As before $(\eta_1,\eta_2)=(\sin\theta,\cos\theta),$ but now
  $\theta\in [\frac{\pi}{2},\pi),$ and we are interested in the limit
    $\theta\to\pi^-.$ For the discussion that follows we fix a
    $0<\delta\ll 1,$ and consider $\theta\in[\pi-\delta,\pi].$

    We let $t=\sin(\theta+z),$ the phase is
    \begin{equation}\label{eqn321.204}
      i(\eta_2\sqrt{k_1^2-t^2}+\eta_1t)=ik_1\cos(x)\cosh(y)+k_1\sin(x)\sinh(y),
    \end{equation}
    and therefore the real part of the phase is non-positive if
    $-\pi\leq x\leq 0,$ and $0\leq y.$ On the other hand
    \begin{equation}\label{eqn322.204}
      \sqrt{k_1^2-t^2}=k_1\cos(\theta+x+iy)=k_1\cos(\theta+x)\cosh(y)-ik_1\sin(\theta+x)\sinh(y),
    \end{equation}
 so $\Im\sqrt{k_1^2-t^2}\leq 0,$ if $-\theta\leq x\leq \pi-\theta,$
 and $0\leq y.$ For $\theta\in[\pi-\delta,\pi],$ we can therefore
 deform the contour to a fixed contour like $\Lambda_{\mu 1},$ (see the red
 curve in Fig.~\ref{fig1.2}[a]) where
 the corner at $0$ is replaced by a smooth curve lying in the first
 quadrant. This contour is a positive distance to the stationary point, $-k_1\sin\theta,$ for
 $\theta\in [\pi-\delta,\pi],$ and therefore
 \begin{equation}
   u^+_{0+}(r\eta)=O(r^{-N})\text{ for any }N>0,
 \end{equation}
 uniformly as $\eta_1\to 0^+.$

 We next consider
 \begin{equation}
    u^+_{0-}(r\eta)=\int_{\Lambda_{\mu}}\frac{e^{ir(\eta_1 t-\eta_2\sqrt{k_1^2-t^2})}\htau_+(-\sqrt{k_1^2-t^2})dt}
    {\sqrt{k_1^2-t^2}};
  \end{equation}
the stationary phase occurs at $k_1\eta_1\in\Lambda_{\mu}.$ With
$t=\sin(\theta+z)$ the phase is
\begin{equation}
  i(\eta_1 t-\eta_2\sqrt{k_1^2-t^2})=-ik_1\cos(2\theta+x)\cosh(y)-k_1\sin(2\theta+x)\sinh(y).
\end{equation}
If $-2\theta\leq x\leq \pi-2\theta,$ and $0\leq y,$ then the real part
of the phase is non-positive. We also have
\begin{equation}
  -\sqrt{k_1^2-t^2}=-ik_1\cos(\theta+x)\sinh(y)+ik_1\sin(\theta+x)\sinh(y)
\end{equation}
To have $\Im -\sqrt{k_1^2-t^2}\leq 0$ we need to have $-\theta-\pi\leq x\leq
-\theta,$ so both conditions hold if $-2\theta\leq x\leq -\theta.$ For
$\theta\in [\pi-\delta,\pi]$ we fix a deformation of $\Lambda_{\mu}$ into the
second quadrant of the type $\Lambda_{\mu 2};$ (see the red curve in
Fig.~\ref{fig1.2}[b]). It is easy to see that along such a curve we can arrange
to have $-2\theta\leq x\leq -\theta,$ and $0\leq y.$ Cauchy's theorem implies
that
  \begin{equation}
    u^+_{0-}(r\eta)=\int_{\Lambda_{\mu 2}}\frac{e^{ir(\eta_1 t-\eta_2\sqrt{k_1^2-t^2})}\htau_+(-\sqrt{k_1^2-t^2})dt}
    {\sqrt{k_1^2-t^2}};
  \end{equation}
the stationary points at $k_1\sin\theta$ are interior points of this
contour, and therefore we have a asymptotic expansion
\begin{equation}
  u^+_{0-}(r\eta)\sim\frac{e^{ik_1r}}{\sqrt{r}}\sum_{j=0}^{\infty}\frac{a_{0j}^+(\eta)}{r^j},
\end{equation}
with coefficients that are smooth as $\eta_1\to 0^+,$ and uniformly bounded
error terms.

We next turn to
\begin{equation}
       u^-_{0+}(r\eta)=\frac{i}{4\pi}\int_{\Lambda_{\mu}}\frac{e^{ir(\eta_2\sqrt{k_1^2-t^2}+\eta_1t)}
         \htau_-(\sqrt{k_1^2-t^2})
          \psi(\sqrt{k_1^2-t^2})dt}{\sqrt{k_1^2-t^2}};
 \end{equation}
with $t=\sin(\theta+z),$ we use the computations in~\eqref{eqn321.204}
and~\eqref{eqn322.204} to see that the real part of the phase is
non-positive if $-\pi\leq x\leq 0,$ and
$\Im\sqrt{k_1^2-t^2}\geq 0$ if $-\theta-\pi\leq x\leq -\theta.$ The
stationary point is at $t=-k_1\sin\theta=k_1\sin(\theta-\pi).$ If
$\eta_1>0,$ then the stationary point does not lie on $\Lambda_{\mu}$
and therefore $u^-_{0+}(r\eta)=O(r^{-N}).$

To see that this statement
is true uniformly down to $\eta_1=0^+,$ we let
$x\in\cC^{\infty}([0,1])$ be a monotone decreasing function with
\begin{equation}
  x(s)=\begin{cases}
  &1\text{ for }s\in [0,\frac 14],\\
  &0\text{ for }s\in [\frac 34,1],
  \end{cases}
\end{equation}
and $y(s)=(1-s)\phi_0.$ We then define the smooth family of contours
$\Lambda^{\theta}_{\mu-}:$
\begin{equation}
  \Lambda^{\theta}_{\mu-}=\{\sin(\theta-\pi+(\pi-\theta)x(s)+iy(s):\:s\in
         [0,1]\}\cup [-k_1\sin\theta,\mu],
\end{equation}
with the segments meeting at the stationary point $-k_1\sin\theta.$ These are
like the blue curves shown in Figure~\ref{fig7.203}. The function
$\htau_-(\sqrt{k_1^2-t^2})$ is smooth in the set swept out by these contours, $\{\Lambda^{\theta}_{\mu-}:\:\theta\in[\pi-\delta,\pi]\}.$
This contour satisfies the conditions above, hence we can apply Cauchy's
theorem to conclude that, for $\eta_1\neq 0,$
\begin{equation}
       u^-_{0+}(r\eta)=\frac{i}{4\pi}\int_{\Lambda^{\theta}_{\mu-}}\frac{e^{ir(\eta_2\sqrt{k_1^2-t^2}+\eta_1t)}
         \htau_-(\sqrt{k_1^2-t^2})
          \psi(\sqrt{k_1^2-t^2})dt}{\sqrt{k_1^2-t^2}}.
\end{equation}
Using Lemma~\ref{lem5.202}, as in the analysis
of~\eqref{eqn310.204}, we conclude that
\begin{equation}
  u^-_{0+}(r\eta)=O(r^{-N})\text{ for any }N>0,
\end{equation}
uniformly as $\eta_1\to 0^+.$

The final case is
\begin{equation}
       u^-_{0-}(r\eta)=\frac{i}{4\pi}\int_{\Lambda_{\mu}}\frac{e^{ir(\eta_1t-\eta_2\sqrt{k_1^2-t^2})}
         \htau_-(-\sqrt{k_1^2-t^2})
          \psi(\sqrt{k_1^2-t^2})dt}{\sqrt{k_1^2-t^2}},
\end{equation}
with stationary point at $t=k_1\sin\theta.$ Using the estimates from
above we see that, if $t=k_1\sin(\theta+x+iy),$ then the real part of the phase is non-positive if
$-2\theta\leq x\leq \pi-2\theta,$ $0\leq y,$ and $\Im-\sqrt{k_1^2-t^2}\geq 0$ if
$-\theta\leq x\leq \pi-\theta,$ and both conditions hold if
$-\theta\leq x\leq \pi-2\theta.$ As before we can construct a smooth
family of contours $\Lambda^{\theta}_{\mu+},$ which satisfy these
conditions and contain 2 segments that meet at
$k_1\sin\theta=k_1\sin(\theta+(\pi-2\theta)).$ These are like the blue curves
shown in Figure~\ref{fig6.203}. By Cauchy's theorem
\begin{equation}
       u^-_{0-}(r\eta)=\frac{i}{4\pi}\int_{\Lambda^{\theta}_{\mu+}}\frac{e^{ir(\eta_1t-\eta_2\sqrt{k_1^2-t^2})}
         \htau_-(-\sqrt{k_1^2-t^2})
          \psi(\sqrt{k_1^2-t^2})dt}{\sqrt{k_1^2-t^2}}.
\end{equation}
It follows from~\eqref{eqn128.210} and~\eqref{eqn253.101} that the function
$\htau_-(-\sqrt{k_1^2-t^2})$ is smooth in the set swept out by these contours,
$\{\Lambda^{\theta}_{\mu+}:\:\theta\in[\pi-\delta,\pi]\};$ arguing as above, it
follows from Lemma~\ref{lem5.202} that
\begin{equation}
  u^-_{0-}(r\eta)=\frac{e^{ik_1r}}{\sqrt{r}}\sum_{j=1}^{\infty}\frac{a^{-}_{0j}(\eta)}{r^j},
\end{equation}
where the coefficients are smooth as $\eta_1\to 0^+,$  and the error
estimates are uniform down to $\eta_1=0^+.$
Essentially identical arguments apply to estimate $u_1(r\eta),$
and to the case $\eta_1<0.$ Altogether this completes the proof of
Theorem~\ref{thm3.75}.

\Rd

We finish the proof of Theorem~\ref{thm3.75} by showing that the asymptotic expansions
for the free space contributions to $u^{l,r}(r\eta)$ can be differentiated with
respect to $r$ to obtain asymptotic expansions for $\pa_r^ju^{l,r}_k(r\eta),$
for any $j\in\bbN$ and $k=0,1.$ These contributions to $u^{l,r}(r\eta)$ satisfy
the Helmholtz equation in a half space, which can be written in polar
coordinates as
\begin{equation}
  \pa_r^2u^{l,r}_{k}(r\eta)+\frac{1}{r}\pa_r
  u^{l,r}{k}(r\eta)+\frac{1}{r^2}\pa_{\eta}^2u^{l,r}_{k}(r\eta)+k_1^2u^{r,l}_k(r\eta)=0.
\end{equation}
We have already shown that the expansions can be differentiated with respect to $\eta.$
Using this equation and a simple inductive argument it suffices to show that
$\pa_ru^{l,r}_k(r\eta)$ has such an expansion.

This claim is immediate from the proof of Theorem~\ref{thm3.75} with $j=0,$ and
the fact that we can differentiate the various formul{\ae} for $u^{l,r}_0,
u^{l,r}_1$ under the integral sign leading to integrals of exactly the same type
as those estimated in the proof of this theorem.  This is immediately clear for
the compactly supported part of the integrals, that is $u^{l,r; 0}_k.$
Differentiating the formul{\ae} for $u^{l,r; \pm}_k$ in~\eqref{eqn130.666} with
respect to $r$ introduces a factor of $i(\eta_2\xi+\eta_1\sqrt{k_1^2-\xi^2})$ in
the numerator.  This changes nothing about the analysis of the bounded parts of
these integrals. The functions $\htau_{\pm}(\xi),\hsigma_{\pm}(\xi)$ are rapidly
decreasing and therefore the unbounded parts of these integrals, with the
additional factor of $i(\eta_2\xi+\eta_1\sqrt{k_1^2-\xi^2}),$ remain uniformly
$O(r^{-N})$ for any $N>0$ and $\eta_2\in [-1,1].$ Since these functions have asymptotic
expansions, applying the theorem of Coddington and Levinson to the functions
$\sqrt{r}e^{-ik_1r}u^{l,r}_{k}(r\eta)$ shows that these expansions can only
be obtained by differentiating the expansions for $u^{l,r}_0(r\eta), u^{l,r}_1(r\eta).$

\Bk

\begin{remark}
  Polar coordinates in this section are defined by
  $r=\sqrt{x_1^2+x_2^2},\,\tan\theta=\frac{x_1}{x_2}.$  We can equally well
  define polar coordinates with respect to any point $(0,D)$ on the $x_2$-axis,
  that is $r_D=\sqrt{x_1^2+(x_2-D)^2},\,\tan\theta_D=\frac{x_1}{x_2-D},$ and
  prove essentially identical asymptotic expansions with respect to this choice
  of polar coordinates. This proves useful in the next section when we consider
  the perturbation of the continuous spectral part of the solution.
\end{remark}


\section{Estimates for $\cW^{l,r}\tau-\cW^{l,r\,'}\sigma$}\label{sec4}

We now prove estimates for the perturbation terms
$\cW^{l,r}\tau-\cW^{l,r\,'}\sigma.$ This requires consideration of many special
cases as the kernels of $\cW^{l,r}$ depend in a non-trivial way on whether or
not the arguments belong to $\supp q_{l,r}.$ To prove that the solutions satisfy
the radiation conditions within the channel, we also need to consider the
behavior of the solution as $|x_1|\to\infty,$ with $|x_2|$ bounded. \Rd In this
section we use standard polar coordinates
\begin{equation}
  \eta=(\cos\theta,\sin\theta).
\end{equation}
\Bk

 The operator $\cW$ can be split into a continuous spectral part, $w^c_{0+},$
 and a guided mode part, $w^g_{0+}.$ The guided mode contribution is a finite
 sum, and, where $x_1>0,$ the guided modes take the very simple form
 $\{v_n(x_2)e^{ i\xi_n x_1}\},$ with $k_1<\xi_n<k_2$ and $\{v_n(x_2)\}$
 exponentially decaying. Estimating these contributions is quite simple.
 We consider the contributions of the continuous spectrum,
 $u_{c0}(x)-u_{c1}(x),$ where
 \begin{equation}\label{eqn0.1}
   \begin{split}
     &u_{c0}(x)\overset{d}{=}
     \int_{-\infty}^{\infty}w^c_{0+}(x;0,y_2)\tau(y_2)dy_2\\
     &u_{c1}(x)\overset{d}{=}
     \int_{-\infty}^{\infty}\pa_{y_1}w^c_{0+}(x;0,y_2)\sigma(y_2)dy_2.
     \end{split}
 \end{equation}
 The kernel $w^c_{0+}$ has a representation in terms of the Fourier transform
 \begin{equation}
   w^c_{0+}(x;0,y_2)=\frac{1}{2\pi}\int_{\Gamma_{\nu}^+}\tw(\xi,x_2;0,y_2)e^{ix_1\xi}d\xi,
 \end{equation}
 where, for $x_1>0,$ the integral is over the contour $\Gamma_{\nu}^+,$ see
 Figure~\ref{FTcontour}, which includes semi-circles in the upper half plane
 centered on the zeros of the Wronskian. In the following pages we prove
 asymptotic results for $u_{c0}$ and $u_{c1},$ which rely on the asymptotic
 expansions for the sources,~\eqref{eqn230.63}, which are proved above in
 Theorem~\ref{thm2.75}.  The maximum $N$ for which such an expansion holds
 depends on the data $g(x_2), h(x_2).$
   \begin{figure}
  \centering
    \includegraphics[width= 10cm]{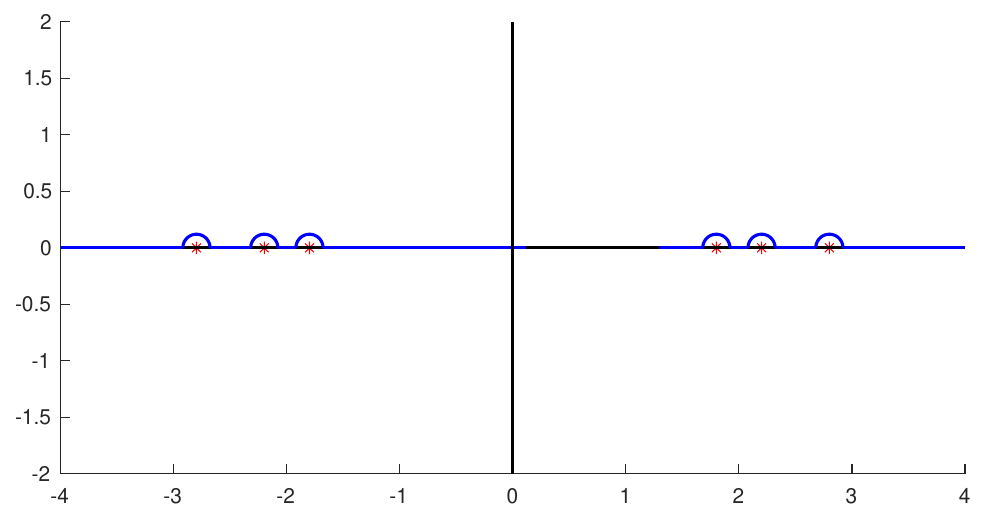}
    \caption{The contour $\Gamma^+_{\nu}$ shown in blue. The roots of
      Wronskian $\{\pm\xi_n\}$ are shown as red asterisks.}
   \label{FTcontour}
   \end{figure}

  We prove the following asymptotic expansions for the perturbation terms. 
 \begin{theorem}\label{thm3.200}
   Suppose that $\sigma,\tau$ have asymptotic expansions as in~\eqref{eqn230.63}
   for a sufficiently large $N.$ If $\eta_2\neq 0,$ then for $M$ depending on $N,$
   \begin{equation}\label{eqn7.7}
     \begin{split}
       u^{l,r}_{c0}(r\eta)=\frac{e^{ik_1r}}{r^{\frac
           12}}\sum_{j=0}^M\frac{a^{l,r}_j(\eta)}{r^{j}}+O(r^{-(M+1)}),\\
         u^{l,r}_{c1}(r\eta)=\frac{e^{ik_1r}}{r^{\frac
           12}}\sum_{j=0}^M\frac{b^{l,r}_j(\eta)}{r^{j}}+O(r^{-(M+1)}).
     \end{split}
   \end{equation}
   The coefficients are smooth functions of $(\eta_1,\eta_2)$ where
   $l\leftrightarrow \eta_1< 0,$ $r\leftrightarrow \eta_1> 0,$ resp. They have
   smooth extensions as $\eta_2\to 0^{\pm},$ and the error terms are bounded
   uniformly as $\eta_2\to 0^{\pm}.$\footnote{See Remark~\ref{rmk9.200}.} The
   functions $\pa_r^j u^{l,r}_{c0}(r\eta), \pa_r^j u^{l,r}_{c1}(r\eta)$ with
   $j\in\bbN,$ have similar expansions obtained by differentiating the
   expansions in~\eqref{eqn7.7}.

   For fixed $x_2$ we have the expansions
    \begin{equation}\label{eqn8.7}
     \begin{split}
       u^{l,r}_{c0}(x_1,x_2)=\frac{e^{ik_1x_1}}{x_1^{\frac
           12}}\sum_{j=0}^M\frac{\ta^{l,r}_j(x_2)}{x_1^{ j}}+O(x_1^{-(M+1)}),\\
         u^{l,r}_{c1}(x_1,x_2)=\frac{e^{ik_1x_1}}{x_1^{\frac
           12}}\sum_{j=0}^M\frac{\tb^{l,r}_j(x_2)}{x_1^{ j}}+O(x_1^{-(M+1)}).
     \end{split}
    \end{equation}
      The coefficients are continuous functions of $x_2$ and the errors are
   uniformly bounded. The leading coefficients $\ta^{l,r}_0(x_2),\tb^{l,r}_0(x_2)$ are
   constant for $\pm x_2>d.$ The
   functions $\pa_{x_1}^j u^{l,r}_{c0}(x_1,x_2), \pa^j_{x_1}u^{l,r}_{c1}(x_1,x_2)$ with
   $j\in\bbN,$ have similar expansions obtained by differentiating the
   expansions in~\eqref{eqn8.7}.
 \end{theorem}
 \begin{remark}\label{rmk9.200}
 For a subtle reason, with the assumption that the potential is supported in
 $[-d,d],$ the error terms in~\eqref{eqn7.7},
   \begin{equation}
     u^{l,r}_{c0}(r\eta)-\frac{e^{ik_1r}}{r^{\frac
           12}}\sum_{j=0}^M\frac{a^{l,r}_j(\eta)}{r^{j}},\quad  u^{l,r}_{c1}(r\eta)-\frac{e^{ik_1r}}{r^{\frac
           12}}\sum_{j=0}^M\frac{b^{l,r}_j(\eta)}{r^{j}},
   \end{equation}
   are not obviously uniform as $\eta_2\to 0^{\pm}.$ The reason is that the formul{\ae}
   used to derive these expansions, essentially~\eqref{eqn19.7}
   and~\eqref{eqn19.711}, do not apply unless $r|\eta_2|>d.$ This problem is
   easily dealt in several ways.  We could simply replace the polar coordinates
   $(r,\theta)$ centered at $(0,0)$  with polar coordinates centered at $(0,\pm d):$
   $(r_{\pm},\theta_{\pm}),$ with $r_{\pm}=\sqrt{x_1^2+(x_2\mp d)^2}.$ See
   Figure~\ref{figpm}.
 \begin{figure}[h]
  \centering
     \includegraphics[height= 6cm]{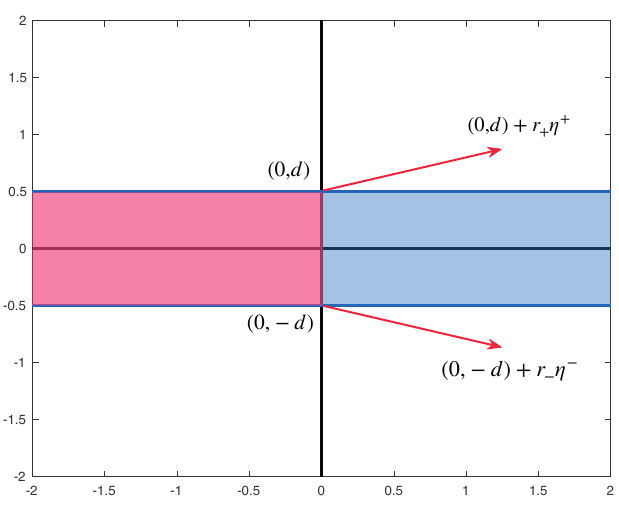}
     \caption{Modified polar coordinates $(r_{\pm},\eta_{\pm}),$
       with $\eta_{\pm}=(\cos(\theta_{\pm}),\sin(\theta_{\pm}).$}
   \label{figpm}
 \end{figure}
 The arguments given to prove Theorem~\ref{thm3.200}
apply equally well in  the new polar coordinates,
   \begin{equation}
     (x_1,x_2)=(r_{\pm}\cos(\theta_{\pm}),r_{\pm}\sin(\theta_{\pm})\pm d).
   \end{equation}
In these coordinates the formul{\ae} used to prove the asymptotic expansions are
valid as soon as $r_{\pm}>0,$ and therefore the error terms are now uniformly
valid as $\theta_{\pm}\to 0^{\pm}.$

Alternatively, for the case of asymptotics in the first quadrant we could assume
that the potential, $q^r$ is supported in $[-2d,0];$ with this assumption the
error terms in~\eqref{eqn8.7} as $\eta_2\to 0^+$ will be uniformly small. We
could equally well have assumed that the potential is supported in $[0,2d]$ to obtain
uniform estimates as $\eta_2\to 0^-.$ In order not to complicate the notation we
prove the theorem using standard polar coordinates.
 \end{remark}
 \begin{remark} The order $M$ in~\eqref{eqn7.7} and~\eqref{eqn8.7} tends to infinity as $N\to\infty.$ For simplicity
   we again assume that the expansions in~\eqref{eqn230.63} are valid with an
   arbitrarily large $N.$
 \end{remark}

 To prove this theorem we split the $y_2$-integrals in~\eqref{eqn0.1},
 as before, into a part with $|y_2|$ small and a part with $|y_2|$
 unbounded. The arguments are quite different for the 2 cases: we
 use direct estimates on the kernels for $|y_2|$ small, and a Fourier
 representation for $|y_2|$ large.  We now drop the $l,r$ sub- and
 super-scripts, and focus on the $r$-case. Assume that $\supp q\subset
 [-d,d];$ let $\varphi_{\pm}\in\cC^{\infty}(\bbR),$ with
 $\varphi_-(y_2)=\varphi_+(-y_2),$ and
 \begin{equation}\label{eqn1.1}
   \varphi_{+}(y_2)=\begin{cases} &0\text{  for }y_2<d+1,\\
   &1\text{ for }y_2>d+2.
   \end{cases}
 \end{equation}
 Let $\varphi_0(y_2)=1-(\varphi_+(y_2)+\varphi_-(y_2));$ with $\epsilon\in\{0,+,-\},$ define
 \begin{equation}
   \begin{split}
     &u_{c0}^{\epsilon}(r\eta)=\int_{-\infty}^{\infty}w^c_{0+}(r\eta;0,y_2)\varphi_{\epsilon}(y_2)\tau(y_2)dy_2,\\
   &u_{c1}^{\epsilon}(r\eta)=\int_{-\infty}^{\infty}\pa_{y_1}w^c_{0+}(r\eta;0,y_2)\varphi_{\epsilon}(y_2)\sigma(y_2)dy_2.
   \end{split}
 \end{equation}

\subsection{Large $|y_2|$ contributions}
If $|y_2|>d,$ then the $y_2$-dependence of $\tw(\xi,x_2;0,y_2)$
reduces to the factor $e^{i\sqrt{k_1^2-\xi^2}|y_2|}.$ We begin with
the large $|y_2|$ estimates. This is accomplished by first integrating
in the $y_2$-variable, which amounts to taking the Fourier transforms
of the `tails' of the sources,
$$\varphi_{\pm}(y_2)\sigma(y_2),  \varphi_{\pm}(y_2)\tau(y_2).$$

Let $\varphi_{\pm}$ be defined in~\eqref{eqn1.1}, and set
\begin{equation}
  \sigma_{\pm}(y_2)=\varphi_{\pm}(y_2)\sigma(y_2),\quad  \tau_{\pm}(y_2)=\varphi_{\pm}(y_2)\tau(y_2).
\end{equation}
In this section we use  slightly different definitions of
$\hsigma_{\pm},\htau_{\pm}$ from that given in~\eqref{eqn89.42}, used in the previous section:
\begin{equation}
  \begin{split}
    \hsigma_{\pm}(\xi)=\int_{-\infty}^{\infty}e^{\pm i\xi
      y_2}\varphi_{\pm}(y_2)\sigma(y_2)dy_2,\\
    \htau_{\pm}(\xi)=\int_{-\infty}^{\infty}e^{\pm i\xi y_2}\varphi_{\pm}(y_2)\tau(y_2)dy_2;
  \end{split}
\end{equation}
this choice of signs simplifies subsequent computations.

As noted, we assume that $N$ can be taken arbitrarily large
in~\eqref{eqn230.63}, then Lemma~\ref{lem4.201} shows that
$\hsigma_{\pm}(\xi),\htau_{\pm}(\xi)$ are smooth and rapidly decreasing as
$|\xi|\to\infty.$ However, we need to shift the frequency by $\pm k_1$ to get
the correct small $|\xi|$ behavior.  Lemma~\ref{lem4.201} shows that, for any
$N,$ there are polynomials $p^N_{\pm}, q^N_{\pm},$ and analytic functions,
$\hsigma^N_{\pm a}(\xi),\tau^N_{\pm a}(\xi),$ so that, for $\xi$ near to $\pm
k_1,$ 
\begin{equation}
  \begin{split}
    \hsigma_{\pm}(\xi)&=\frac{p^N_{\pm}(\xi+k_1)}{(\xi+k_1)^{\frac 12}}+\hsigma^N_{\pm
      a}(\xi)+O((\xi+k_1)^{N}),\\
    \htau_{\pm}(\xi)&=q^N_{\pm}(\xi+k_1)(\xi+k_1)^{\frac 12}+\htau^N_{\pm a}(\xi)+O((\xi+k_1)^{N}).
  \end{split}
\end{equation}
The $O$-terms are $\cC^{N+1}$-functions.
The singularities in these functions occur only at $\xi=-k_1.$ These functions
have analytic continuations to the upper half plane and decay exponentially as
$\Im\xi\to\infty.$

The arguments in this section are very similar to those in
Section~\ref{sec7.4.83}. In addition to the slight change in the definitions of
$\hsigma_{\pm},\htau_{\pm},$  the principal difference is in the definition of the
phase. In the $\xi$-variable, the phase in Section~\ref{sec3.207} is
\begin{equation}
  i(\eta_2\xi+\eta_1\sqrt{k_1^2-\xi^2}),
\end{equation}
whereas in this section it is
\begin{equation}
  i(\eta_2\sqrt{k_1^2-\xi^2}+\eta_1\xi).
\end{equation}
As a consequence we use standard polar coordinates
\begin{equation}
  (\eta_1,\eta_2)=(\cos\theta,\sin\theta),
\end{equation}
rather than the normalization in~\eqref{eqn98.207}.

The precise form of the integral depends on whether $|x_2|$ is greater or less
than $d;$ we start with $\eta_2\neq 0,$ so that $|x_2|=r|\eta_2|.$ \Rd The function
$\fA_0(\xi,\upsilon)$ is defined in  equation (211) in Part I. Assuming
that $q(x_2)=(k_2^2-k_1^2)\chi_{[-d,d]}(x_2),$ it is given by
\begin{equation}
  \fA_0(\xi,\upsilon)=\int_{-d}^{d}q(z_2)e^{\upsilon(z_2-d)}\left[\cos[(d+z_2)\sqrt{k_2^2-\xi^2}]+
    \upsilon\frac{\sin[(d+z_2)\sqrt{k_2^2-\xi^2}]}{\sqrt{k_2^2-\xi^2}}\right]dz_2;
\end{equation}
and the Wronskian $\fW(\xi,\upsilon)$  given in equation (209) of Part I is
\begin{equation}
   \fW(\xi,\upsilon)=-e^{-2i\upsilon d}\left[(2\xi^2-k_1^2-k_2^2)\frac{\sin 2d\sqrt{k_2^2-\xi^2}}{\sqrt{k_2^2-\xi^2}}
    -2i\upsilon\cos 2d\sqrt{k_2^2-\xi^2}\right].
\end{equation}
\Bk These are entire functions of $(\xi,\upsilon).$

In this case we interchange the order of integrations to see that
\begin{equation}\label{eqn19.7}
  \begin{split}
    u^{\pm}_{c0}(r\eta)=& \int\limits_{-\infty}^{\infty}\left[
    \frac{1}{2\pi}\int\limits_{\Gamma_{\nu}^+}\frac{e^{i\left[\sqrt{k_1^2-\xi^2}(|x_2|\pm y_2)+\xi
          x_1\right]}   \fA_0(\xi,-i\sqrt{k_1^2-\xi^2})d\xi }
  {\fW(\xi,\sqrt{k_1^2-\xi^2})\sqrt{k_1^2-\xi^2}}\right]\varphi_{\pm}(y_2)\tau(y_2)dy_2\\
   =& \frac{1}{2\pi}\int\limits_{\Gamma_{\nu}^+}\frac{e^{ir[\sqrt{k_1^2-\xi^2}|\eta_2|+\xi\eta_1]}\fA_0(\xi,-i\sqrt{k_1^2-\xi^2})
  \htau_{\pm}(\sqrt{k_1^2-\xi^2})}
   {\fW(\xi,\sqrt{k_1^2-\xi^2})\sqrt{k_1^2-\xi^2}}d\xi.
   \end{split}
\end{equation}
\Rd The functions $\htau_{\pm}$ extend analytically to the upper half plane; they are
evaluated at $\sqrt{k_1^2-\xi^2}.$ Note that for real $\xi$ with $|\xi|>k_1,$
$\sqrt{k_1^2-\xi^2}=i\sqrt{\xi^2-k_1^2}$ lies on the positive imaginary
axis. For sufficiently small $\nu>0$ and $\xi\in\Gamma^+_{\nu},$ lying over
$[-k_2,-k_1]\cup [k_1,k_2],$ the $\Im(i\sqrt{\xi^2-k_1^2})\geq 0,$ so the
$\xi$-integrand in~\eqref{eqn19.7} makes sense. \Bk To see that this change in the order of
integrations is justified we observe that the integrand in the first line is
estimated for $\xi \in\Gamma_{\nu}$ by
\begin{equation}\label{eqn203.41}
  \frac{M|e^{i\sqrt{k_1^2-\xi^2}(|x_2|+|y_2|-2d)}|\varphi_{\pm}(y_2)}
       {(1+|\xi|)^2|\sqrt{k_1^2-\xi^2}|(1+|y_2|)^{\frac
      32}},
\end{equation}
for a constant $M.$  We have used the estimates, for $\xi \in\Gamma_{\nu},$ 
\begin{equation}\label{eqn204.41}
  \begin{split}
    |\fA_0(\xi,-i\sqrt{k_1^2-\xi^2})|&\leq
    \frac{M|e^{2d\sqrt{\xi^2-k_1^2}}|}{1+|\xi|}\text{ (221) from Part I},\\
    |\fW(\xi,\sqrt{k_1^2-\xi^2})|&\geq  M(1+|\xi|)\text{\phantom{mm} (190) from Part I}.
  \end{split}
\end{equation}
As $|x_2|+|y_2|\geq 2d$ the fact that the expression in~\eqref{eqn203.41} is
integrable in $(\xi,y_2)$ justifies the equality in~\eqref{eqn19.7}. 

We give the details for estimating $u^{+}_{c0}(r\eta);$ the estimates for
$u^{-}_{c0}(r\eta)$ are essentially identical.  Where $|\xi|>k_1,$ the square
root satisfies $\sqrt{k_1^2-\xi^2}=i\sqrt{\xi^2-k_1^2},$ with $\sqrt{x}>0$ for
$x\in (0,\infty).$ There is a stationary phase at $\xi=k_1\eta_1;$ the
contributions from the rest of integral are rapidly decreasing. The
$\Re\sqrt{k_1^2-\xi^2}\geq 0$ throughout the domain of this integral, which
means that the singularity of $\htau_+(t)$ at $-k_1$ plays no role.

There is no difficultly as $\eta_1\to 0.$ To prove uniform estimates as
$\eta_1\to 1,$ we need to examine the portions of this integral near to $\pm
k_1.$ Choose a $0<\mu<k_1/2,$ so that $k_1+\mu<\xi_1-\nu,$ where $\xi_1$ is the
smallest guided mode frequency.  We also assume that $\delta>0$ is fixed so that
if $|\eta_2|<\delta,$ then $k_1\eta_1\in[k_1-\mu/2,k_1].$ Let
$\psi\in\cC^{\infty}_c((k_1-\mu,k_1+\mu))$ which equals $1$ in
$[k_1-\mu/2,k_1+\mu/2].$ Letting $t=\sqrt{k_1^2-\xi^2}$ where $\xi< k_1$ and
$\tau=\sqrt{\xi^2-k_1^2}$ where $\xi>k_1$ we see that
\begin{multline}\label{eqn16.2}
 u^+_{c00}(r\eta)= \frac{i}{2\pi}\int_{k_1-\mu}^{k_1+\mu}\frac{e^{ir[\sqrt{k_1^2-\xi^2}|\eta_2|+\xi\eta_1]}\fA_0(\xi,-i\sqrt{k_1^2-\xi^2})
  \htau_+(\sqrt{k_1^2-\xi^2})\psi(\xi)}
       {\fW(\xi,\sqrt{k_1^2-\xi^2})\sqrt{k_1^2-\xi^2}}d\xi\\=
        \frac{i}{2\pi}\int_{0}^{\sqrt{2k_1\mu-\mu^2}}\frac{e^{ir[t|\eta_2|+\sqrt{k_1^2-t^2}\eta_1]}\fA_0(\sqrt{k_1^2-t^2},-it)
  \htau_+(t)\psi(\sqrt{k_1^2-t^2})}
             {\fW(\sqrt{k_1^2-t^2},t)\sqrt{k_1^2-t^2}}dt+\\
              \frac{i}{2\pi}\int_{0}^{\sqrt{2k_1\mu+\mu^2}}\frac{e^{ir[i\tau|\eta_2|+\sqrt{k_1^2+\tau^2}\eta_1]}\fA_0(\sqrt{k_1^2+\tau^2},\tau)
  \htau_+(i\tau)\psi(\sqrt{k_1^2+\tau^2})}
       {\fW(\sqrt{k_1^2+\tau^2},i\tau)i\sqrt{k_1^2+\tau^2}}d\tau.   
\end{multline}
As before, the sum of two integrals on the right hand side are
simply the complex contour integral ``$t$''-integral over the contour
$$t\in
\Lambda_{\mu}=[i\sqrt{2k_1\mu+\mu^2},i0]\cup [0,\sqrt{2k_1\mu-\mu^2}],$$ shown
as the thick black ``L''s in Figure~\ref{fig6.203}. The stationary point, shown
as a black dot, now lies along the real axis where $t=k_1|\eta_2|.$

To prove uniform asymptotics we need to deform the contour, being careful to
keep the real part of the phase, $i(t|\eta_2|+\sqrt{k_1^2-t^2}\eta_1),$
non-positive. We use the change of variables $t=k_1\sin(\theta+x+iy),$ where $
(\eta_1,|\eta_2|)=(\cos\theta,\sin\theta),$ which, as noted earlier, differs from
the normalization used in the previous section.  Computing as in~\eqref{eqn101.206},
we see that the contour $\Lambda_{\mu}$ corresponds to
$[-\theta+i\phi_0,-\theta+i0]\cup[-\theta,\theta_0-\theta],$ for a
$\phi_0,\theta_0>0.$ As before, the phase is $k_1\cos(z)$ and therefore
\begin{equation}
  i\cos(x+iy)=i\cos x\cosh y+\sin x\sinh y.
\end{equation}
We see that the real part is non-positive for $-\pi\leq x\leq 0,$ and $0< y$ and
where-ever $y=0.$
The $t$ variable
\begin{equation}
  k_1\sin(\theta+x+iy)=k_1[\sin(\theta+x)\cosh(y)+i\cos(\theta+x)\sinh(y)]
\end{equation}
has non-negative imaginary part if
$x\in[-\frac{\pi}{2}-\theta,\frac{\pi}{2}-\theta],$ $0< y,$ or $y=0.$
In the $t$-variable these conditions corresponds to
\begin{equation}
Q_+=\{k_1\sin(\theta+x+iy):\:  x\in [-\frac{\pi}{2}-\theta,0],\quad 0<y\text{ or
  }y=0\}.
\end{equation}

We consider two deformations of $\Lambda_{\mu}.$ The first deformation,
$\Lambda^{\theta}_{\mu 1},$ replaces the corner of $\Lambda_{\mu}$ near $0$ with a smooth
interpolant between the $x$-axis and the $y$-axis, lying in $Q_+$ intersected
with the first quadrant. Examples are shown as the red curves in
Figure~\ref{fig6.203}. The green curves are the right boundaries of $Q_+.$ Using
Cauchy's theorem we see that the integral on the right hand side
of~\eqref{eqn16.2} can be replaced with
\begin{equation}
  u^+_{c00}(r\eta)= \frac{i}{2\pi}\int_{\Lambda^{\theta}_{\mu 1}}\frac{e^{ir[t|\eta_2|+\sqrt{k_1^2-t^2}\eta_1]}\fA_0(\sqrt{k_1^2-t^2},-it)
  \htau_+(t)\psi(\sqrt{k_1^2-t^2})}
             {\fW(\sqrt{k_1^2-t^2},t)\sqrt{k_1^2-t^2}}dt
\end{equation}
This is the integral of a smooth compactly supported function on a smooth arc,
which therefore has a complete asymptotic expansion arising from the stationary
phase at $k_1|\eta_2|=k_1\sin\theta,$
\begin{equation}
  u^+_{c00}(r\eta)=\frac{e^{ik_1 r}}{r^{\frac
      12}}\sum_{j=0}^{N}\frac{a^+_{0j}(\eta)}{r^j}+O\left(r^{-(N+\frac 32)}\right).
\end{equation}
The coefficients are smooth functions of $|\eta_2|,$ hence the best we can hope for is
that they extend smoothly as $\eta_2\to 0^{\pm},$  see Remark~\ref{rmk9.200}.

To show that the coefficients $\{a^+_{0j}(\eta)\}$ extend  to
$\eta_2\to 0^{\pm},$ and the error terms are
uniformly bounded as $\eta_2\to 0^{\pm},$ we use a second deformation,
$\Lambda^{\theta}_{\mu 2},$ given by
\begin{equation}
  \Lambda^{\theta}_{\mu 2}=\{k_1\sin(\theta g(s)+is):\: s\in [0,s_0]\}\cup [k_1\sin\theta,\sqrt{2k_1\mu-\mu^2}].
\end{equation}
Here $\sinh(s_0)=\sqrt{2\mu k_1+\mu^2}$ and $g\in\cC^{\infty}([0,s_0])$ is a monotone function satisfying
\begin{equation}
  g(s)=
  \begin{cases}
    &0\text{ for }s\in [\frac{3s_0}{4},s_0]\\
      &1\text{ for }s\in [0,\frac{s_0}{4}].
  \end{cases}
\end{equation}
Examples are shown in blue in Figure~\ref{fig6.203}.  The arcs
$\Lambda^{\theta}_{\mu 2}$ depend smoothly on $\theta\in [0,\epsilon]$ for an
$\epsilon>0.$ In the integral along $\Lambda^{\theta}_{\mu 2}$ the stationary
phase occurs at $k_1\sin\theta$ which is the common endpoint of the two smooth
arcs making up the contour $\Lambda^{\theta}_{\mu 2}.$ Applying
Lemma~\ref{lem5.202} it is clear that coefficients extend smoothly as
$\eta_2=0^{\pm}.$ Using the modified polar coordinates described in
Remark~\ref{rmk9.200} the implied constants in the error terms are easily seen
to be uniformly bounded as $\eta_2\to 0^{\pm}.$

We also need to consider what happens at the opposite end of the interval,
$[-(k_1+\mu),-k_1+\mu].$ It is not hard to see that this integral can be
represented as the contour integral
\begin{equation}
   -\frac{i}{2\pi}\int_{\Lambda_{\mu}}\frac{e^{ir[t|\eta_2|-\sqrt{k_1^2-t^2}\eta_1]}\fA_0(-\sqrt{k_1^2-t^2},-it)
  \htau_+(t)\psi(-\sqrt{k_1^2-t^2})}
             {\fW(-\sqrt{k_1^2-t^2},t)\sqrt{k_1^2-t^2}}dt.
\end{equation}
The phase has a stationary point at $-k_1|\eta_2|,$ which lies outside the domain
of the integration, but approaches it as $|\eta_2|\to 0.$ Using the
variables $t=k_1\sin(\theta+z),$ we see that the phase
\begin{multline}
  i(t|\eta_2|-\sqrt{k_1^2-t^2}\eta_1)=-ik_1\cos(2\theta+z)=\\
  -i\cos(2\theta+x)\cosh y-\sin(2\theta+x)\sinh y,
\end{multline}
which shows that the real  part is non-positive where
\begin{equation}
  x\in\left[-2\theta,\pi-2\theta\right],\quad 0\leq y,\text{ or }y=0.
\end{equation}
The $\Im t\geq 0$ if $ x\in\left[-2\theta,\frac{\pi}{2}-2\theta\right],\quad 0\leq y.$
We can therefore replace $\Lambda_{\mu}$ with a contour $\Lambda^{\theta_0}_{\mu
  1},$ for a fixed $\theta_0>0,$ along which $\Im t\geq 0.$ In light of this, it
is clear that, as $\eta_2\to 0^{\pm},$ the contribution from this interval is
uniformly $O(r^{-N}),$ for all $N>0.$ The remaining contributions from $\Gamma^+_{\nu}$ are easily shown to be
uniformly $O(r^{-N}),$ for any $N>0.$

As the double integral on the left hand side in~\eqref{eqn19.711} is not
absolutely convergent, a somewhat more involved argument is required to show
that
\begin{equation}\label{eqn19.711}
  \begin{split}
    u^{\pm}_{c1}(r\eta)=& \int\limits_{-\infty}^{\infty}\left[
    \frac{1}{2\pi}\int\limits_{\Gamma_{\nu}^+}\frac{\xi^2e^{i\left[\sqrt{k_1^2-\xi^2}(|x_2|\pm y_2)+\xi
          x_1\right]}   \fA_0(\xi,-i\sqrt{k_1^2-\xi^2})d\xi }
  {\fW(\xi,\sqrt{k_1^2-\xi^2})\sqrt{k_1^2-\xi^2}}\right]\varphi_{\pm}(y_2)\sigma(y_2)dy_2\\
   =& \frac{1}{2\pi}\int\limits_{\Gamma_{\nu}^+}\frac{\xi^2e^{ir[\sqrt{k_1^2-\xi^2}|\eta_2|+\xi\eta_1]}\fA_0(\xi,-i\sqrt{k_1^2-\xi^2})
  \hsigma_{\pm}(\sqrt{k_1^2-\xi^2})}
   {\fW(\xi,\sqrt{k_1^2-\xi^2})\sqrt{k_1^2-\xi^2}}d\xi.
   \end{split}
\end{equation}
\Rd As before, for small enough $\nu>0,$ $\Im\sqrt{k_1^2-\xi^2}\geq 0$ on
$\Gamma_{\nu}^+,$ so the integrand in~\eqref{eqn19.711} make sense.  \Bk The
$\xi$-integral on the first line of~\eqref{eqn19.711} is
$O((|x_2|+|y_2|)^{-\frac 32}),$ and therefore
\begin{equation}
  \begin{split}
    u^{\pm}_{c1}(r\eta)=&\lim_{R\to\infty} \int\limits_{-R}^{R}\left[
    \frac{1}{2\pi}\int\limits_{\Gamma_{\nu}^+}\frac{\xi^2e^{i\left[\sqrt{k_1^2-\xi^2}(|x_2|\pm y_2)+\xi
          x_1\right]}   \fA_0(\xi,-i\sqrt{k_1^2-\xi^2})d\xi }
         {\fW(\xi,\sqrt{k_1^2-\xi^2})\sqrt{k_1^2-\xi^2}}\right]\varphi_{\pm}(y_2)\sigma(y_2)dy_2\\
    =&\lim_{R\to\infty} 
    \frac{1}{2\pi}\int\limits_{\Gamma_{\nu}^+}\left[\int\limits_{-R}^{R}e^{i\sqrt{k_1^2-\xi^2}|y_2|}\varphi_{\pm}(y_2)\sigma(y_2)dy_2\right]\times\\
    &\phantom{mmmmmmmmmmm}\frac{\xi^2e^{i(\sqrt{k_1^2-\xi^2}|x_2|+\xi  x_1)}   \fA_0(\xi,-i\sqrt{k_1^2-\xi^2})d\xi }
         {\fW(\xi,\sqrt{k_1^2-\xi^2})\sqrt{k_1^2-\xi^2}}
  \end{split}
\end{equation}
To handle the limit we first observe that the only term in the asymptotic
expansion of $\sigma(y_2)$ for which the double integral in~\eqref{eqn19.711} is not absolutely
convergent is the $y_2^{-\frac 12}$-term. We treat the $+$-case, the other case is
essentially identical. It suffices to consider
\begin{multline}\label{eqn218.42}
  \int\limits_{d}^{R}e^{i\sqrt{k_1^2-\xi^2}y_2}\varphi_{+}(y_2)\frac{e^{ik_1y_2}dy_2}{\sqrt{y_2}}=\\
  \hsigma_{+}(\sqrt{k_1^2-\xi^2})-
  \lim_{R_1\to\infty} \int\limits_{R}^{R_1}e^{i(\sqrt{k_1^2-\xi^2}+k_1)y_2}\frac{dy_2}{\sqrt{y_2}}+\lot
\end{multline}
Here $\lot$ are terms for which is clear that we can interchange the order of
integrations in~\eqref{eqn19.711}.

An integration by parts shows that the limit of the integral on the right hand
side of~\eqref{eqn218.42}, as $R_1\to\infty,$ is bounded by
\begin{equation}
 \frac{2|e^{i\sqrt{k_1^2-\xi^2}R}|}{\sqrt{R}|\sqrt{k_1^2-\xi^2}+k_1|}.
\end{equation}
Using this estimate and those  in~\eqref{eqn204.41} we see that the error term
\begin{multline}
  \left|\,\int\limits_{\Gamma_{\nu}^+}\left[\int\limits_{R}^{\infty}e^{i(\sqrt{k_1^2-\xi^2}+k_1)|y_2|}\frac{dy_2}{\sqrt{y_2}}\right]
   \frac{\xi^2e^{i(\sqrt{k_1^2-\xi^2}|x_2|+\xi  x_1)}   \fA_0(\xi,-i\sqrt{k_1^2-\xi^2})d\xi }
        {\fW(\xi,\sqrt{k_1^2-\xi^2})\sqrt{k_1^2-\xi^2}}\right|\\
        \leq
        \frac{C}{\sqrt{R}}\int\limits_{\Gamma_{\nu}^+}\frac{\left|e^{i\sqrt{k_1^2-\xi^2}(R+|x_2|-2d)}\right||\xi|^2|d\xi|}
        {(1+|\xi|)^2|\sqrt{k_1^2-\xi^2}||\sqrt{k_1^2-\xi^2}|+k_1|},
\end{multline}
from which it is clear that the limit as $R\to\infty$ is zero, proving that
the identity in~\eqref{eqn19.711} is correct.

The functions $\hsigma_{\pm}(\xi)$ have the same analytic extension properties as
$\htau_{\pm}(\xi),$ and are also singular at $\xi=-k_1.$ The only other
difference in this term is the factor of $\xi^2,$ which has no significant
effects. Thus we easily show that, for $\eta_2\neq 0,$ $u^{\pm}_{1c}(r\eta)$ have
standard asymptotic expansions, which are uniformly valid as $\eta_2\to 0^{\pm}.$ Hence we have
\begin{equation}\label{eqn32.6}
    u^{\pm}_{ck}(r\eta)=\frac{e^{ik_1 r}}{r^{\frac
        12}}\sum_{j=0}^{N}\frac{a^{\pm}_{kj}(\eta)}{r^j}+O\left(r^{-(N+\frac
      32)}\right),\text{ for }k=0,1.
\end{equation}
To get uniform error estimates as $\eta_2\to 0^{\pm}$ we need to use the
modified polar coordinates, $(r_{\pm},\eta^{\pm})$ as described in
Remark~\ref{rmk9.200}.

\subsection{Small $|y_2|$ contributions}
     
To estimate the contributions to the asymptotics from $y_2$ in bounded
intervals, we consider the behavior of the perturbation kernels
$w_{0+}^c(r\eta;0,y_2),$  $\pa_{x_1}w_{0+}^c(r\eta;0,y_2)$ as $r\to\infty,$
with $y_2$ bounded. As noted in Remark~\ref{rmk9.200}, we should replace standard
polar coordinates with $(r_{\pm},\theta_{\pm}),$ where $\pm x_2>d,$ to obtain
uniform error terms as $\theta_{\pm}\to 0^{\pm},$ but to avoid more complicated
notation we continue to use standard polar coordinates.

The principal contribution again comes from the stationary
phase that now occurs where $\xi=\eta_1 k_1.$ We first assume that $\eta_2\neq 0.$
Thus, along the rays $x=r\eta,$ the second component $|x_2|$ is eventually
larger than $d.$ For $x_2,y_2>d,$ $k=0,1,2,$ we start with the integral 
\begin{equation}\label{eqn182.62}
  w_{0+}^{[k]}(r\eta;0,y_2)=\frac{i}{2\pi}\int_{\Gamma_{\nu}^+}\xi^k
  \frac{e^{i[\sqrt{k_1^2-\xi^2}(r\eta_2+y_2)+\xi
      r\eta_1]}\fA_0(\xi,-i\sqrt{k_1^2-\xi^2})}
  {\fW(\xi,\sqrt{k_1^2-\xi^2})\sqrt{k_1^2-\xi^2}}d\xi.
\end{equation}
As noted above, the stationary phase, as $r\to\infty,$  occurs where
\begin{equation}
  \frac{-\xi}{\sqrt{k_1^2-\xi^2}}\eta_2+\eta_1=0,
\end{equation}
which implies that
\begin{equation}
  \xi=\eta_1 k_1.
\end{equation}
Using the analysis from the previous section we can show that this integral has
asymptotic expansion
\begin{equation}\label{eqn182.622}
  w_{0+}^{[k]}(r\eta;0,y_2)=\frac{e^{ik_1r}}{\sqrt{r}}\left[\sum_{l=0}^{N}\frac{a_{kl}(\eta;0,y_2)}{r^l}+O(r^{-(N+1)})\right],
\end{equation}
with coefficients that are smooth functions of $(\eta, y_2),$ uniformly as $\eta_1\to 0^+,1^-.$
Similar considerations apply when  $x_2<0,$ we need to replace $\eta_2$
in~\eqref{eqn182.622} with $-\eta_2,$ which gives the same asymptotic
result. If $y_2<0$ then we replace $y_2$ with $-y_2.$ Altogether we
get an asymptotic expansion with the same form.

The analysis where $|y_2|\leq d$ and $x_2>d,$ begins with
\begin{equation}\label{eqn182.623}
  w_{0+}^{[k]}(r\eta;0,y_2)=\frac{i}{2\pi}\int_{\Gamma_{\nu}^+}\xi^k
  \frac{e^{i[\sqrt{k_1^2-\xi^2}r\eta_2+\xi
      r\eta_1]}\fC(\xi,\sqrt{k_1^2-\xi^2};y_2)}
  {\fW(\xi,\sqrt{k_1^2-\xi^2})\sqrt{k_1^2-\xi^2}}d\xi,
\end{equation}
where $\fC(\xi,\upsilon;y_2)$ is defined in equation (250) of~\cite{EpWG2023_1}, and
is an entire function of $(\xi,\upsilon).$ Once again there is a stationary phase at
$\xi=k_1\eta_1,$ and the argument used above, along with the estimates for $\fC$
from~\cite{EpWG2023_1} apply to show that
\begin{equation}\label{eqn264.666}
    w_{0+}^{[k]}(r\eta;0,y_2)=\frac{e^{ik_1r}}{\sqrt{r}}\left[\sum_{l=0}^{N}\frac{a_{kl}(\eta;0,y_2)}{r^l}+O(r^{-(N+1)})\right],
\end{equation}
with coefficients that are smooth functions of $(\eta_1,y_2)\in
[0,1]\times[-d,d].$ As these estimates are uniform in $y_2$ over bounded
intervals, they show that, for $\eta_2\neq 0,$
\begin{equation}\label{eqn49.6}
  u^0_{cj}(r\eta)=\frac{e^{ik_1r}}{\sqrt{r}}\left[\sum_{k=0}^N\frac{a_{jk}(\eta)}{r^k}+O(r^{-(N+1)})\right],
\end{equation}
uniformly as $\eta_2\to\pm 1.$

\Rd Similarly to the free space contributions, the asymptotic expansions for
$u^{\pm}_{ck}(r\eta)$ can be differentiated with respect to $r$ to obtain
asymptotic expansions for $\pa_r^ju^{\pm}_{ck}(r\eta).$ Once again these
functions satisfy the free space Helmholtz equation, so it is again just a
matter of checking this statement for $j=1.$ The right hand sides of the
formul{\ae}~\eqref{eqn19.7} and~\eqref{eqn19.711} can be differentiated with
respect to $r$ introducing a factor of $i(\sqrt{k_1^2-\xi^2}|\eta_2|+\xi\eta_1)$
into the numerators of these integrands. As before this changes nothing about
the analysis of either the principal terms or the error terms in these
expansions.  The formul{\ae},~\eqref{eqn182.62} and~\eqref{eqn182.623}, for
$w^{[k]}_{0+}(r\eta;0,y_2)$ can be differentiated with respect to $r$
introducing a factor of $i(\sqrt{k_1^2-\xi^2}\eta_2+\xi\eta_1)$ in the numerator
of the integrand. Once again, this has no effect on the analysis. Since these
functions have asymptotic expansions, the theorem of Coddington and Levinson
applies as before to show that these can only be obtained by differentiating the
expansions for $w^{[k]}_{0+}(r\eta;0,y_2),$ leading to the same conclusion for
$u^0_{cj}(r\eta).$ \Bk

\begin{remark}\label{rmk11.46}
 As noted above, in order to get uniform error estimates in the asymptotic
 formul{\ae} they should be stated in terms polar coordinates
 $(r_{\pm},\theta_{\pm})$ centered on $(0,\pm d).$ With these choices the
 functions $u^{\pm}_{cj}((\pm d,0)+r_{\pm}\eta^{\pm})$ are given by a single
 formula (essentially~\eqref{eqn19.7}) for all $r_{\pm}\geq 0;$ the asymptotic
 formul{\ae} are of the same form, assuming that $\pm\eta^{\pm}_2\geq 0,$
   \begin{equation}\label{eqn7.7.1}
     \begin{split}
       u^{l,r}_{c0}((\pm d,0)+r_{\pm}\eta^{\pm})=\frac{e^{ik_1r_{\pm}}}{r_{\pm}^{\frac
           12}}\sum_{k=0}^M\frac{\Ha^{l,r}_k(\eta^{\pm})}{r_{\pm}^{
           k}}+O(r_{\pm}^{-(M+1)}),\\
         u^{l,r}_{c1}((\pm d,0)+r_{\pm}\eta^{\pm})=\frac{e^{ik_1r_{\pm}}}{r_{\pm}^{\frac
           12}}\sum_{k=0}^M\frac{\hb^{l,r}_k(\eta^{\pm})}{r_{\pm}^{
           k}}+O(r_{\pm}^{-(M+1)});
     \end{split}
   \end{equation}
   the coefficients are smooth where $\eta_1^{\pm}\in [0,1]$ in the $r$-case
   and $\eta_1^{\pm}\in [-1,0]$ in the $l$-case; the errors 
   are uniformly bounded  as $\eta^{\pm}_2\to 0^{\pm}.$
\end{remark}

\subsection{Estimates with $x_1\to\infty,$ and $|x_2|$ bounded}
We now prove estimates with $|x_2|$ remaining bounded as $x_1\to\infty.$ We
begin with the contribution of the unbounded part of the integral. We need to
distinguish the cases $|x_2|\leq d$ and $|x_2|>d.$ If $|x_2|>d,$ then the
unbounded part of the integral takes the form
\begin{equation}\label{eqn268.666}
  u^{\pm}_{c0}(x_1,x_2)
  =\frac{1}{2\pi}\int_{\Gamma_{\nu}^+}\frac{e^{i[\sqrt{k_1^2-\xi^2}|x_2|+\xi x_1]}\fA_0(\xi,-i\sqrt{k_1^2-\xi^2})
  \htau_{\pm}(\sqrt{k_1^2-\xi^2})}
  {\fW(\xi,\sqrt{k_1^2-\xi^2})\sqrt{k_1^2-\xi^2}}d\xi.
\end{equation}
This integrand does not appear to have a stationary phase as $x_1\to\infty,$ but is singular at
$\xi=\pm k_1.$ Focusing, as before, on $[k_1-\mu,k_1+\mu],$ the contribution is
given by the contour integral:
\begin{equation}
  \frac{i}{2\pi}\int_{\Lambda_{\mu}}\frac{e^{i[t|x_2|+\sqrt{k_1^2-t^2}x_1]}\fA_0(\sqrt{k_1^2-t^2},-it)
  \htau_{\pm}(t)\psi(\sqrt{k_1^2-t^2})}
             {\fW(\sqrt{k_1^2-t^2},t)\sqrt{k_1^2-t^2}}dt,
\end{equation}
which does have a stationary point, as $x_1\to\infty,$ at $t=0.$  We can deform
$\Lambda_{\mu}$ to a fixed contour
like the red curve in Figure~\ref{fig7.203}, which lies in the upper half
plane, on which $\Re i\sqrt{k_1^2-t^2}\leq 0.$ This implies that these integrals have
asymptotic expansions of the form
\begin{equation}
  \frac{e^{ik_1x_1}}{\sqrt{x_1}}\sum_{j=0}^{\infty}\frac{a_{j}(x_2)}{x_1^{j}},
\end{equation}
with the coefficients smooth functions for $|x_2|\geq d.$ Note that $a_0(x_2)$ is constant
 where $\pm x_2>d.$ 
 
If we examine the contribution from $[-(k_1+\mu),-k_1+\mu]$ in the same way we
get
\begin{equation}
   \frac{i}{2\pi}\int_{\Lambda_{\mu}}\frac{e^{i[t|x_2|-\sqrt{k_1^2-t^2}x_1]}\fA_0(-\sqrt{k_1^2-t^2},-it)
  \htau_+(t)\psi(\sqrt{k_1^2-t^2})}
             {\fW(-\sqrt{k_1^2-t^2},t)\sqrt{k_1^2-t^2}}dt.
\end{equation}
If $\theta\in [0,\frac{\pi}{2}],$ then 
\begin{equation}
  -\Re[i\sqrt{k_1^2-r^2e^{2i\theta}}]\leq 0
\end{equation}
as well.  This shows that we can replace $\Lambda_{\mu}$ with
$\Lambda^{\theta_0}_{\mu 1},$ for a fixed $\theta_0>0,$ and thereby avoid the
stationary point at $t=0.$ Hence this part of the integral is $O(x_1^{-N})$ for
all $N,$ uniformly as $\pm x_2\to d^{+}.$ The remaining parts of the integral over $\Gamma_{\nu}^+$ are similarly
rapidly decreasing and therefore for $|x_2|>d,$ we see that, for any $N>0,$
\begin{equation}
  u^{\pm}_{c0}(x_1,x_2)=\frac{e^{ik_1
      x_1}}{\sqrt{x_1}}\sum_{j=0}^{N}\frac{a^{\pm}_{0j}(x_2)}{x_1^{j}}+O\left(x_1^{-\frac{N+2}{2}}\right).
\end{equation}

We now consider what happens if $|x_2|<d.$ In this case the integral takes a
somewhat different form:
\begin{equation}\label{eqn274.666}
  u^{\pm}_{c0}(x_1,x_2)
  =\frac{1}{2\pi}\int_{\Gamma_{\nu}^+}\frac{e^{i\xi  x_1+2id\sqrt{k_1^2-\xi^2}}
    \fB(\xi,\sqrt{k_1^2-\xi^2};x_2) \htau_{\pm}(\sqrt{k_1^2-\xi^2})\psi(\xi)}
  {\fW(\xi,\sqrt{k_1^2-\xi^2})\sqrt{k_1^2-\xi^2}}d\xi,
\end{equation}
where $\fB(\xi,\upsilon;x_2)$ is defined in equation (234) of~\cite{EpWG2023_1}, and is
an entire function of $(\xi,\upsilon).$  As
before the only possible stationary phase contributions come from neighborhoods
of $\pm k_1,$ with the remainder of the integral rapidly decreasing as
$x_1\to\infty.$ The contribution near to $\xi=
k_1$ is given by the contour integral
\begin{equation}
  \frac{i}{2\pi}\int_{\Lambda_{\mu}}\frac{e^{i\sqrt{k_1^2-t^2}x_1+2idt}\fB(t,\sqrt{k_1^2-t^2};x_2)
  \htau_{\pm}(t)\psi(\sqrt{k_1^2-t^2})}
             {\fW(\sqrt{k_1^2-t^2},t)\sqrt{k_1^2-t^2}}dt.
\end{equation}
We can show that this integral has a stationary phase at $t=0$ and therefore an
asymptotic expansion of the form
\begin{equation}
 \frac{e^{ik_1x_1}}{\sqrt{x_1}}\sum_{j=0}^{\infty}\frac{a_j(x_2)}{x_1^{j}},
\end{equation}
with the coefficients depending smoothly on $x_2.$ The contribution from near to
$-k_1$ takes the same form with $\sqrt{k_1^2-t^2}$ replaced with
$-\sqrt{k_1^2-t^2}.$ As before, we can deform the contour away from the
stationary point, and conclude that this integral is $O(x_1^{-N})$ for any
$N>0.$ Altogether this shows that with $|x_2|$ bounded, the functions
$u^{\pm}_{c0}(x_1,x_2)$ has asymptotic expansions
\begin{equation}\label{eqn42.666}
  u^{\pm}_{c0}(x_1,x_2)=\frac{e^{ik_1x_1}}{\sqrt{x_1}}\sum_{j=0}^{N}\frac{a^{\pm}_{0j}(x_2)}{x_1^j}+
  O\left(x_1^{-(N+\frac{3}{2})}\right),
\end{equation}
with uniformly bounded errors and coefficients that are smooth on $
(-\infty,-d]\cup[-d,d]\cup[d,\infty),$ with  finite differentiability at
    $x_2=\pm d.$

Similar arguments apply to show that
\begin{equation}\label{eqn42.6}
  u^{\pm}_{c1}(x_1,x_2)=\frac{e^{ik_1x_1}}{\sqrt{x_1}}\sum_{j=0}^{N}\frac{a^{\pm}_{1j}(x_2)}{x_1^j}+
  O\left(x_1^{-(N+\frac{3}{2})}\right),
\end{equation}
for any $N>0.$ Again with uniformly bounded errors and coefficients that are
smooth on $ (-\infty,-d]\cup[-d,d]\cup[d,\infty),$ with  finite
    differentiability at $x_2=\pm d.$
    
    
To estimate the contribution of the bounded part of the integral, we need to
consider the behavior of kernel function $w^{[k]}_{0+}(x_1,x_2;0,y_2)$ as $x_1\to\infty,$ with
$y_2, x_2$ remaining bounded. The stationary phase now occurs where $\xi=k_1;$ if
$|x_2|>d,$ then the integrand decays exponentially as $|\xi|\to\infty,$ and the
analysis from the previous section shows that
\begin{equation}
  w_{0+}^{[k]}(x_1,x_2;0,y_2)=\frac{e^{ik_1x_1}}{\sqrt{x_1}}\left[\sum_{l=0}^N\frac{b_{kl}(x_2,y_2)}{x_1^l}+O(x_1^{-(N+1)})\right]
\end{equation}
with coefficients smooth away from $x_2,y_2=\pm d,$ where they are finitely
differentiable. Once again $b_{j0}$ does not depend on $x_2>d.$

We now assume that both $|x_2|<d,$ and $|y_2|<d,$ with
$x_1\to\infty.$ In this case
\begin{equation}\label{eqn280.666}
  w_{0+}^{[k]}(x_1,x_2;0,y_2)=\int_{\Gamma_{\nu}^+}\frac{\xi^k e^{i\xi
      x_1}\fD(\xi,\sqrt{\xi^2-k_1^2};x_2,y_2)}
  {\fW(\xi,\sqrt{k_1^2-\xi^2})\sqrt{k_1^2-\xi^2}}d\xi,
\end{equation}
with $\fD(\xi,\upsilon;x_2,y_2)$ defined in equation (266) of~\cite{EpWG2023_1}, and is
again an entire function of $(\xi,\upsilon).$ The singularities at $\pm k_1$
produce, after changing variables, stationary phases at $\xi=\pm k_1.$ As before
we can deform the contour and show that the contribution from $\xi=-k_1$ is
rapidly decreasing. The function $\fD$ involves $\cosh x\sqrt{\xi^2-k_2^2}$
and$\frac{\sinh x\sqrt{\xi^2-k_2^2}}{\sqrt{\xi^2-k_2^2}}$ both of which are
entire functions. For this reason there are no stationary phase contributions
from $\pm k_2.$

We let
$\psi\in\cC^{\infty}_c((-(k_2+1),k_2+1))$ be an even, monotone, non-negative function equal to
  $1$ in the interval $[-(k_2+\frac 12),k_2+\frac 12]$ and set
\begin{equation}
  w_{00+}^{[k]}(x_1,x_2;0,y_2)=\int_{\Gamma_{\nu}^+}\frac{\xi^k e^{i\xi
      x_1}\fD(\xi,-i\sqrt{k_1^2-\xi^2};x_2,y_2)\psi(\xi)d\xi}
  {\fW(\xi,\sqrt{k_1^2-\xi^2})\sqrt{k_1^2-\xi^2}}.
\end{equation}
This term includes the stationary phases. We need to consider the
contributions from near to $\pm k_1.$ For that purpose let
$\varphi\in\cC^{\infty}_c(k_1-\epsilon,k_1+\epsilon),$ equal one in a
small neighborhood of $k_1,$ with
$0<\epsilon<\min\{(\xi_1^{r}-k_1)/2,k_1/2\}.$ Let
$t=\pm\sqrt{k_1^2-\xi^2},$ up to a rapidly decreasing error
$w_{00+}^{[k]}(x_1,x_2;0,y_2),$ is a sum of the integrals
\begin{multline}\label{eqn231.55}
  I_{\pm}(x_1,x_2;y_2)=\\
  \pm\int_{\Lambda_{\mu}}\frac{[\pm(k_1^2-t^2)]^{\frac k2} e^{\pm ix_1\sqrt{k_1^2-t^2}}
    \fD(\pm\sqrt{k_1^2-t^2},-it;x_2,y_2)\varphi(\sqrt{k_1^2-t^2})dt}
  {\fW(\pm\sqrt{k_1^2-t^2},t)\sqrt{k_1^2-t^2}},
\end{multline}
which each have a stationary  phase at $t=0,$ the intersection of the two segments
that make up $\Lambda_{\mu}.$

We first consider $I_-.$ In this case we deform $\Lambda_{\mu}$ by
replacing a neighborhood of the corner at 0 with a smooth curve  lying
in the first quadrant, like the
red curve in Figure~\ref{fig6.203}[a]. This curve is parameterized by
$\rho(s)e^{i\theta(s)},$ with $\theta(s)\in [0,\pi/2].$ Along this
curve we see that $\Im \rho^2(s)e^{2i\theta(s)}\geq 0,$ and therefore
\begin{equation}
  \Re -i\sqrt{k_1^2-\rho^2(s)e^{2i\theta(s)}}\leq 0.
\end{equation}
As the deformed curve avoids the stationary point at 0, we conclude
that, for all $N,$ $I_-(x_1,x_2;y_2)=O(x_1^{-N}).$

To analyze $I_+$ we deform $\Lambda_{\mu}$ by replacing a line segment,
$[0,\delta],$ with a smooth curve lying in the fourth quadrant meeting the $x$-
and $y$-axes smoothly, see the red curve in Figure~\ref{fig9.55}. The deformation has a parameterization
$\rho(s)e^{i\theta(s)}$ with $-\frac{\pi}{2}\leq \theta(s)\leq 0,$ so that
\begin{equation}
  \Re i\sqrt{k_1^2-\rho^2(s)e^{2i\theta(s)}}\leq 0.
\end{equation}
 \begin{figure}[h]
  \centering
     \includegraphics[height= 8cm]{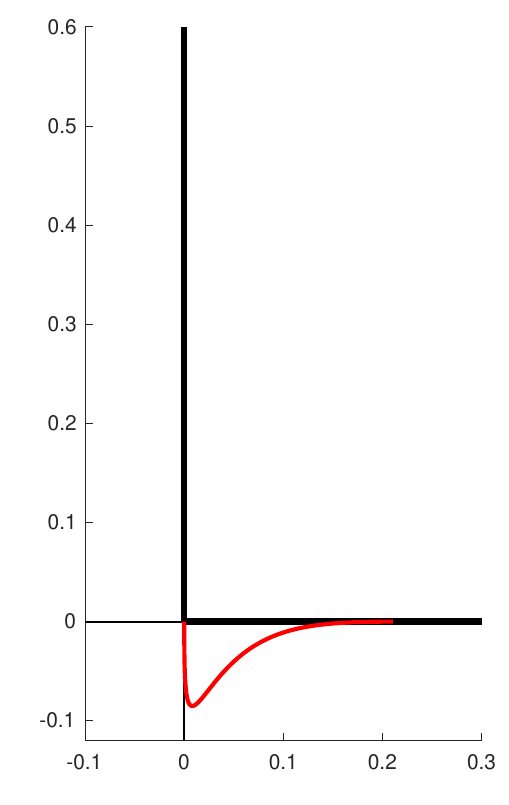}
    \caption{Contour deformation to compute the integral $I_+$ in~\eqref{eqn231.55}}
   \label{fig9.55}
 \end{figure}
The critical point at $t=0$ is an interior point of the
deformed  contour, which shows that we have the  asymptotic
expansions
\begin{equation}
  w_{00+}^{[k]}(x_1,x_2;0,y_2)=\frac{e^{ik_1
      x_1}}{\sqrt{x_1}}\left[\sum_{l=0}^{N}\frac{f_{kl}(x_2;y_2)}{x_1^{l}}\right]+O\left(x_1^{-(N+\frac
    32)}\right)\text{ for any }N,
\end{equation}
where $\{f_{jk}(x_2;y_2)\},$ are continuous functions on
$[-d,d]\times[-d,d],$ which are smooth away from $x_2,y_2=\pm d.$

What requires
further effort is to show that the ``error term''
\begin{equation}
  w_{01+}^{[k]}(x_1,x_2;y_2)=\int_{\Gamma_{\nu}^+}\frac{\xi^k e^{i\xi
      x_1}\fD(\xi,\sqrt{\xi^2-k_1^2};x_2,y_2)(1-\psi(\xi))}
  {\fW(\xi,\sqrt{k_1^2-\xi^2})\sqrt{k_1^2-\xi^2}}d\xi
\end{equation}
is rapidly decreasing as $x_1\to\infty.$ This is obvious so long as $x_2\neq
y_2,$ as
\begin{equation}
  |\fD(\xi,\sqrt{\xi^2-k_1^2};x_2,y_2)|\leq Me^{-\sqrt{\xi^2-k_1^2}|x_2-y_2|}\left[\frac{1}{|1+|\xi|}+|x_2-y_2|\right],
\end{equation}
which is proved in~\cite{EpWG2023_1}.  To prove the rapid decrease where $x_2=y_2$
requires accounting for 
the oscillations from the $e^{i\xi x_1}$-term.

\begin{remark}
  In the remainder of the section we consider a variety of functions that satisfy
  symbolic estimates. We say that a smooth function, $s_l(\xi)$ defined for
  $|\xi|$ sufficiently large, is a ``symbol of order $l$'' if there
  are constants so that
  \begin{equation}
    \limsup_{|\xi|\to\infty}(1+|\xi|)^{m-l}|\pa_{\xi}^ms_l(\xi)|\leq C_m<\infty.
  \end{equation}
  In most cases these functions will depend smoothly on additional parameters,
  and the constants, $\{C_m\},$ in these estimates depend uniformly on these
  parameters. As we never differentiate with respect to them, we often suppress
  the dependence on these parameters.
\end{remark}

The function $\fD(\xi,\sqrt{\xi^2-k_1^2};x_2,y_2)$ can be written as a sum of
terms of the form
\begin{equation}
  e^{-\sqrt{\xi^2-k_i^2}h(x_2,y_2)}s_l(x_2,y_2;\xi),
\end{equation}
where $i=1$ or $2,$ and  $h(x_2,y_2)$ is non-negative, and strictly positive if
$x_2\neq y_2.$ The functions $s_l(x_2,y_2;\xi)$ are  symbols of order
$l\in\{0,-1,-2,-3,-4\}.$ Note, for example that
$e^{(\sqrt{\xi^2-k_2^2}-\sqrt{\xi^2-k_1})d}$ equals $1$ plus a symbol of order
$-1.$ The remainder of the integrand
$$\frac{\xi^j (1-\psi(\xi))}{\fW(\xi,\sqrt{k_1^2-\xi^2})\sqrt{k_1^2-\xi^2}},$$
is a symbol of order $j-2.$

\begin{lemma}\label{lem6.666}
Let $s_l$ be a symbol of order $l\in\bbZ.$   An integral of the form
\begin{equation}
  \int_{-\infty}^{\infty}e^{ix_1\xi}e^{-\lambda\sqrt{\xi^2-k_i^2}} s_{l}(\xi)(1-\psi(\xi))d\xi,
\end{equation}
is uniformly rapidly decreasing in $x_1$ as $\lambda\to 0^+.$ 
\end{lemma}
\begin{proof}
We show this for any $l\in\bbZ$ by
integrating by parts in the most obvious way: for $\lambda>0$ and $m\in\bbN,$ we have
\begin{multline}
  \int_{-\infty}^{\infty}e^{ix_1\xi}e^{-\lambda\sqrt{\xi^2-k_i^2}}
  s_{l}(\xi)(1-\psi(\xi))d\xi=\\
  \frac{1}{(-ix_1)^m}
   \int_{-\infty}^{\infty}e^{ix_1\xi}\pa_{\xi}^m\left[e^{-\lambda\sqrt{\xi^2-k_i^2}}
  s_{l}(\xi)(1-\psi(\xi))\right]d\xi
\end{multline}
If a derivative falls on $1-\psi(\xi),$ then the term is obviously rapidly
decreasing, hence we need to consider:
\begin{equation}\label{eqn59.5}
  \pa_{\xi}^m\left[e^{-\lambda\sqrt{\xi^2-k_i^2}}
    s_{l}(\xi)\right]=\sum_{p=0}^m\left(\begin{matrix}m\\p\end{matrix}\right)\pa_{\xi}^{m-p}s_l(\xi)
    \cdot \pa_{\xi}^{p}e^{-\lambda\sqrt{\xi^2-k_i^2}}.
\end{equation}
We let $s_{l+p-m}=\pa_{\xi}^{m-p}s_l(\xi),$ denote a symbol of order $l+p-m.$

We use the following lemma, which is proved using a simple induction argument
and the fact that
\begin{equation}
  \pa_{\xi}\sqrt{\xi^2-k_i^2}=\frac{\xi}{\sqrt{\xi^2-k_i^2}}
\end{equation}
is a  symbol of order $0$ in $\{\xi:\:|\xi|>k_i+\epsilon\},$ for any $\epsilon>0.$
\begin{lemma}\label{lem7.777}
  For $m\in\bbN$ we have
  \begin{equation}
    \pa_{\xi}^me^{-\lambda\sqrt{\xi^2-k_i^2}} = e^{-\lambda\sqrt{\xi^2-k_i^2}}
    \times \sum_{j=1}^m\lambda^j s_{j-m}(\xi),
  \end{equation}
  where $s_{j-m}$ is a symbol of order $j-m$ in
  $\{\xi:\:|\xi|>k_i+\frac 12\}.$
\end{lemma}

Applying  lemma~\ref{lem7.777} and~\eqref{eqn59.5} we see that
\begin{equation}\label{eqn62.5}
  \pa_{\xi}^m\left[e^{-\lambda\sqrt{\xi^2-k_i^2}}
    s_{l}(\xi)\right]=\sum_{j=0}^m\lambda^j s_{l-m+j}(\xi)
    e^{-\lambda\sqrt{\xi^2-k_i^2}}.
\end{equation}
To complete this analysis we need to estimate
\begin{equation}
 I_j(\lambda)= \left|\int\limits_{k_i+\frac 12}^{\infty}e^{i\xi
    x_1}\lambda^js_{l-m+j}(\xi)e^{-\lambda\sqrt{\xi^2-k_i^2}}d\xi\right|
  \text{ for }j\in\{1,\dots,m\},
\end{equation}
for  $m>l+1.$ If $j=0,$ then it is clear that this integral is uniformly bounded
as $\lambda\to 0^+.$

Using the symbolic estimate
$$|s_{l-m+j}(\xi)|\leq C(\xi^2-k_i^2)^{\frac{l-m+j}{2}},\text{ for
}|\xi|>k_i+\epsilon,$$
we see that
\begin{equation}
  I_j(\lambda)\leq C\int\limits_{k_i+\frac 12}^{\infty}
  \lambda^j(\xi^2-k_i^2)^{\frac{l-m+j}{2}}
  e^{-\lambda\sqrt{\xi^2-k_i^2}}d\xi\leq
C'\int\limits_{\sqrt{k_i+\frac 14}}^{\infty}\lambda^jw^{l-m+j}e^{-\lambda w}dw,
\end{equation}
where we have let $w=\sqrt{\xi^2-k_i^2}.$ If $l-m+j<-1,$ then it is clear that
the limit of $I_j(\lambda)$ is $0$ as $\lambda\to 0^+.$ If $l-m+j=-1,$ then
$j=m-l-1>0,$ and once again it is clear that $\lim_{\lambda\to
  0^+}I_j(\lambda)=0.$ Finally, if $l-m+j\geq 0,$ then we let $\lambda w=x$ to
obtain
\begin{equation}
  \lambda^{m-l-1}\int\limits_{\lambda k_i}^{\infty}x^{l-m+j}e^{-x}dx,
\end{equation}
which again tends to zero as $\lambda\to 0^+.$
\end{proof}

Combining these results shows
that, for all $N>0,$ we have the estimate
\begin{equation}
  w^{[k]}_{01+}(x_1,x_2;y_2)=O(x_1^{-N}),
\end{equation}
uniformly for bounded $x_2,y_2.$

\begin{proposition}
  If $x_2$ lies in a bounded interval, then for bounded $y_2,$ and any $N>0,$ we have the
  asymptotic expansions, for $k=0,1,2,$
  \begin{equation}
w^{[k]}_{0+}(x_1,x_2;y_2)  =\frac{e^{ik_1x_1}}{\sqrt{x_1}}\sum_{l=0}^{N}\frac{f_{kl}(x_2,y_2)}{x_1^{l}}+O(x_1^{-(N+1)}),
  \end{equation}
  where $\{f_{jk}(x_2,y_2)\}$ are continuous bounded functions of $x_2,y_2.$
\end{proposition}

As a corollary of the proposition we have the asymptotic expansions for
$u^0_{ck}(x_1,x_2),$ with $x_2$ bounded:
\begin{equation}\label{eqn67.6}
  u^0_{ck}(x_1,x_2)=\frac{e^{ik_1
      x_1}}{\sqrt{x_1}}\left[\sum_{l=0}^{N}\frac{b_{kl} (x_2)}{x_1^{l}}+
    O(x_1^{-(N+1)})\right], \text{ for }k=0,1.
\end{equation}
Taken together,~\eqref{eqn32.6}, \eqref{eqn42.6}, \eqref{eqn49.6},
and~\eqref{eqn67.6} complete the proof of the theorem in the $j=0$ case.

\Rd To show that the functions $\pa_{x_1}^ju_{c0}(x_1,x_2),
\pa_{x_1}^ju_{c1}(x_1,x_2)$ have asymptotic expansions obtained by
differentiating the expansions for $u_{c0}(x_1,x_2), u_{c1}(x_1,x_2)$ we observe
that the formul{\ae}~\eqref{eqn268.666},~\eqref{eqn274.666}
and~\eqref{eqn280.666} can be differentiated with respect to $x_1$ leading to 
factors of $(i\xi)^j$ in the numerator of the integrand. The bounded
contributions of the first two can be estimated exactly as before. The unbounded
contribution to error terms are uniformly bounded because
$\htau_{\pm}(\xi),\hsigma_{\pm}(\xi)$ are rapidly decreasing as
$|\xi|\to\infty.$

To contribution of~\eqref{eqn280.666} to the error is a bit subtler. It is still uniformly
bounded as Lemma~\ref{lem6.666} shows that such integrals are bounded for $s_l(\xi)$
symbols of arbitrary integral order. The additional factor of $\xi^j$ just multiplies  the
symbols in the $j=0$ case. Again this shows that these functions have asymptotic
expansions; the theorem of Coddington and Levinson shows that these can only
be obtained by differentiating the expansions for $u_{c0}(x_1,x_2), u_{c1}(x_1,x_2).$

\Bk

  \section{Concluding Remarks}
  In this paper we have derived refined estimates for the solution, obtained in
  Part I, to the scattering problem specified by two open semi-infinite
  wave-guides meeting along a common perpendicular line. These estimates show
  that the solution satisfies the usual Sommerfeld radiation condition away from
  the channels. Within the channels the solution splits cleanly into a
  ``radiation'' part and a wave-guide mode part. The radiation part satisfies a
  standard Sommerfeld radiation condition. The wave-guide mode parts are sum of terms
  of the form $\{e^{-i\xi^l_mx_1}v(x_2)\}$ with $\xi^l_m>0$ in the left half
  plane, and $\{e^{i\xi^r_mx_1}v(x_2)\}$ with $\xi^r_m>0$ in the right half
  plane. These contributions are therefore outgoing, in a naive sense. In
  Part III we show that they are also outgoing in the rigorous sense first
  introduced by Isozaki in~\cite{Isozaki94}.

  Nonetheless in Part III we introduce the formalism used in~\cite{Vasy2000},
  which gives radiation conditions for very general open wave guide problems in
  any dimension. The radiation condition implies uniqueness for the outgoing
  solution to $(\Delta+q+k^2)u=f,$ provided that $f$ decays rapidly enough.
  This analysis also shows the existence of a limiting absorption solution to
  this equation, which is shown to satisfy these radiation conditions. \Rd Using
  the PDE uniqueness result we show that the integral equation~\eqref{eqn143.35}
  has a trivial null-space on any of the spaces
  $\cC_{\alpha}(\bbR)\oplus\cC_{\alpha+\frac 12}(\bbR),$ for $0<\alpha<\frac
  12.$ Hence these equations are uniquely solvable for data in these spaces.
  \Bk We then conclude that the solutions we have constructed starting with
  admissible data agree with the limiting absorption solutions. Note that, in
  general the solutions to the PDE, $u^{l,r},$ constructed in each half space
  via~\eqref{eqn160.08} with general data
  $(g,h)\in\cC_{\alpha}(\bbR)\oplus\cC_{\alpha+\frac 12}(\bbR),$ cannot be
  expected to satisfy an outgoing radiation condition.

\end{document}